\documentclass[letterpaper,UKenglish]{lipics-v2016}

\usepackage{stackrel}
\usepackage{url}             
\usepackage{hyperref}             
\usepackage{stmaryrd}
\usepackage{amssymb}
\usepackage{empheq}
\usepackage{verbatim}
\usepackage{graphicx}
\usepackage{color}
\usepackage{xcolor,pgf,tikz,pgflibraryarrows,pgffor,pgflibrarysnakes}
\usetikzlibrary{fit} 
\usetikzlibrary{backgrounds} 
\usepackage{todonotes}

\usepackage{wrapfig}
\usepackage{blkarray}
\usepgflibrary{shapes}
\usetikzlibrary{snakes,automata}

\definecolor{gold}{rgb}{0.99,0.78,0.07}

\usetikzlibrary{patterns}
\tikzstyle{background}=[rectangle,fill=gold!40, inner sep=0.08cm, rounded corners=1mm]

\newcommand\inter[1]{\llbracket #1 \rrbracket}

\newcommand{\matindex}[1]{\mbox{\scriptsize#1}}%

\newcommand{\Ss}{\mathcal{S}}
\newcommand{\Aa}{\mathcal{A}}
\newcommand{\Bb}{\mathcal{B}}
\newcommand{\Ee}{\mathcal{E}}
\newcommand{\Uu}{\mathcal{U}}

\newcommand{\Dd}{\mathcal{D}}
\newcommand{\Mm}{\mathcal{M}}
\newcommand{\Tt}{\mathcal{T}}

\newcommand{\seq}[1]{\langle #1 \rangle}

\newcommand{\N}{\mathbb N}
\newcommand{\Nat}{\mathbb N}

\newcommand{\set}[1]{\left\{ #1 \right\}}

\newcommand{\struc}[1]{\Xi_{#1}}

\newcommand{\sst}{\ensuremath{\mathsf{SST}}}

\newcommand{\fot}{\ensuremath{\mathsf{FOT}}} 
 
\newcommand{\twst}{\ensuremath{\mathsf{2WST}}}
\newcommand{\la}{\ensuremath{sf}}

\newcommand{\uniquepred}{\texttt{u\_pred}}
\newcommand{\uniquesucc}{\texttt{u\_succ}}
\newcommand{\istring}{\texttt{is\_string}_{\#}}
\newcommand{\isstring}{\texttt{is\_string}}
\newcommand{\reach}{\texttt{reach}_{\#}}

\newcommand{\ifirst}{\texttt{first}}

\newcommand{\beh}{\Bb}

\newcommand{\MSO}{\mathrm{MSO}}
\newcommand{\FO}{\mathrm{FO}}
\newcommand{\dom}{\mathrm{dom}}
\newcommand{\nodes}{\mathrm{pos}}

\newcommand{\flows}{\rightsquigarrow}

\newcommand{\varsst}{\mathcal{X}}

\newcommand{\contribute}{\text{useful}}

\newcommand{\useful}{\text{useful}}
\newcommand{\oomit}[1]{}

\newcommand{\linear}{copyless}

\newtheorem{proposition}{Proposition}

\usepackage{tikz}
\tikzstyle{nloc}=[draw, text badly centered, rectangle, rounded corners, minimum size=2em,inner sep=0.5em]
\tikzstyle{loc}=[draw,rectangle,minimum size=1.4em,inner sep=0em]
\tikzstyle{trans}=[-latex, rounded corners]
\tikzstyle{trans2}=[-latex, dashed, rounded corners]

\def\rmdef{\stackrel{\mbox{\rm {\tiny def}}}{=}} 

\newcommand{\Ashu}[1]{}
\newcommand{\Manu}[1]{}
\newcommand{\krish}[1]{}

\title{FO-definable transformations of infinite strings}
\author[1]{V. Dave}
\author[1]{S. Krishna} 
\author[1,2]{A. Trivedi}

\affil[1]{Indian Institute of Technology Bombay  (\texttt{krishnas,vrunda@cse.iitb.ac.in})}
\affil[2]{University of Colorado Boulder (\texttt{ashutosh.trivedi@colorado.edu})}

\authorrunning{Dave, Krishna, and Trivedi}

\Copyright{Dave, Krishna, and Trivedi}

\begin{document}

\maketitle
\begin{abstract}
The theory of regular and aperiodic transformations of finite strings has
recently received a lot of interest.
These classes can be equivalently defined using logic (Monadic second-order
logic and first-order logic), two-way machines (regular two-way and aperiodic
two-way transducers), and one-way register machines (regular streaming string
and aperiodic streaming string transducers).
These classes are known to be closed under operations such as sequential
composition and regular (star-free)  choice; and problems such 
as functional equivalence and  type checking,  are decidable for these classes.
On the other hand, for infinite strings these results are only known for
$\omega$-regular transformations: Alur, Filiot, and Trivedi studied
transformations of infinite strings  and introduced an extension of streaming
string transducers over $\omega$-strings and showed that they capture
monadic second-order definable transformations for infinite strings. 
In this paper we extend their work to recover connection for
infinite strings among first-order logic definable transformations, aperiodic
two-way transducers, and aperiodic streaming string transducers.

\end{abstract}

\section{Introduction}
\label{sec:introduction}
The beautiful theory of regular languages is the cornerstone of theoretical
computer science and formal language theory.
The perfect harmony among the languages of finite words definable using abstract machines
(deterministic finite automata, nondeterministic finite automata, and two-way
automata), algebra (regular expressions and finite monoids), and logic (monadic
second-order logic (MSO)~\cite{Bu60}) did set the stage for the generalizations of the
theory to not only for the theory of regular languages of infinite
words~\cite{Bu62,McN66}, trees~\cite{AM09}, partial orders~\cite{Tho96}, but
more recently for the theory of regular transformations of the finite
strings~\cite{AC10}, infinite strings~\cite{FiliotTrivediLics12,ADT13}, and
trees~\cite{AD12}.
For the theory of regular transformations it has been shown that abstract
machines (two-way transducers~\cite{EH01} and streaming string
transducers~\cite{AC10}) precisely capture the transformations definable via
monadic second-order logic transformations~\cite{Cour94}.
For a detailed exposition on the regular theory of languages and
transformations, we refer to the surveys by Thomas~\cite{Tho96,Tho97} and
Filiot~\cite{Fil15}, respectively.  

There is an equally appealing and rich theory for first-order logic (FO) definable
subclasses of regular languages. 
McNaughton and Papert~\cite{McP71} observed the equivalence between
FO-definability and star-free regular expressions for finite words, while 
Ladner~\cite{Ladner77} and Thomas~\cite{thomas79} extended this connection to
infinite words.
The equivalence of star-free regular expressions and languages defined via
aperiodic monoids is due to  Sch\"utzenberger~\cite{Sch65}  and corresponding
extension to infinite words is due to Perrin~\cite{Perrin84}.
For a detailed introduction to FO-definable language we refer the
reader to Diekert and Gastin~\cite{dg08SIWT}.

The results for the theory of FO-definable transformations are relatively recent. 
While Courcelle's definition of logic based transformations~\cite{Cour94}
provides a natural basis for FO-definable transformations of finite as
well as infinite words, ~\cite{FKT14} observed that over finite words,
streaming string transducers~\cite{AC10} with an appropriate notion of
aperiodicity precisely capture the same class of transformations.
Carton and Dartois~\cite{CD15} introduced aperiodic two-way transducers for finite words and
showed that it precisely captures the notion of FO-definability.
We consider transformations of infinite strings and generalize
these results by showing that appropriate aperiodic restrictions on two-way
transducers  and streaming string transducers on infinite strings
capture the essence of FO-definable transformations.
Let us study an example to see how the following $\omega$-transformation
can be represented using logic, two-way transducers, and streaming string
transducers. 

\begin{example}[Example Transformation] Let $\Sigma = \{a,b,\#\}$. 
Consider an $\omega$-transformation $f_1 : \Sigma^\omega \rightharpoonup
\Sigma^\omega$ such that it replaces any maximal $\#$-free finite string $u$ by
$\overline{u}u$, where $\overline{u}$ is the reverse of $u$.
Moreover $f_1$ is defined only for strings with finitely many $\#'s$, e.g. for
all $w {=} u_1\#u_2\#\dots u_n\# v$ s.t $u_i\in\{a,b\}^*$ 
and $v\in \{a,b\}^\omega$, we have $f_1(w) {=} \overline{u_1}u_1\#\dots
\#\overline{u_n}u_n\#v$.
\end{example}
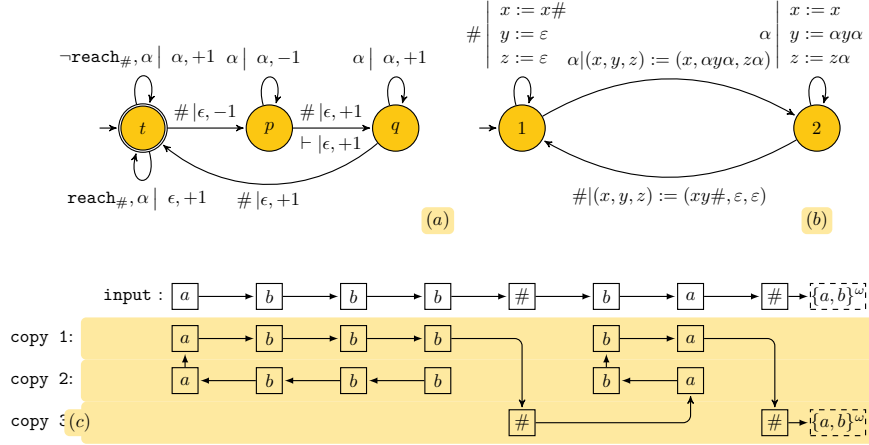
\begin{figure}[t]
  \tikzstyle{trans}=[-latex, rounded corners]
  \begin{center}
    \scalebox{0.7}{
      
      \begin{tikzpicture}[->,>=stealth',shorten >=1pt,auto, semithick,scale=.8]
        \tikzstyle{every state}=[fill=gold]
        
        \node[state] at (0,4) (A) {$p$} ;
        \node[state] at (3,4) (B) {$q$} ;
        \node[state,initial,initial where=left,initial text={}, accepting] at (-3,4) (C) {$t$} ;
        
        \path(A) edge[loop above] node
             {$\alpha\left|\begin{array}{lll} \alpha, -1 \end{array}\right.$}  (A);

             \path(A)  edge node[above] {$\# \left|
               \epsilon, +1 \right.$} node[below] {$\vdash \left| \epsilon,
               +1 \right.$} 
             (B);

             \path(B) edge[loop above] node
                  {$\alpha\left|\begin{array}{llllllll} \alpha, +1 \end{array}\right.$}
                  (B);
                  
                  \path(B) edge [bend left=40] node {$\# \left| \epsilon, +1 \right.$} (C);
                
                \path(C) edge[loop above] node
                     {$\neg \reach, \alpha\left|\begin{array}{llllllll} \alpha, +1 \end{array}\right.$} (C);
                     \path(C) edge[loop below] node
                          {$\reach, \alpha\left|\begin{array}{llllllll} \epsilon, +1 \end{array}\right.$} (C);
                          \path(C) edge node  {$\# \left| \epsilon, -1 \right.$} (A);
                          
                          \node[background] at (4, 1.8) {($a$)};
                          \node[background] at (13, 1.8) {($b$)};
                          \node[background] at (-4.5, -3) {($c$)};
      
      \node[initial,state, initial text={}] at (6, 4) (A1) {$1$} ;
      \node[state] at (13,4) (B1) {$2$} ;

  \path (A1) edge [loop above] node {$\#\left|\begin{array}{llllllll}x:=x\#\\y:=\varepsilon\\z:=\varepsilon\end{array}\right.$} (A1);
  \path (A1) edge [bend left] node [above]
        {$\begin{array}{lllll} \vspace{-2mm}\\ \alpha|(x,y,z):=(x,\alpha
            y \alpha, z \alpha)\end{array}$} (B1);

  \path (B1) edge [loop above] node {$\alpha\left|\begin{array}{lllll} x:=x\\
   y:=\alpha y \alpha \\ z := z\alpha \end{array}\right.$} (A1);

  \path (B1) edge [bend left] node [below]
        {$\begin{array}{l}
            \#|(x,y,z):=(xy\#,\varepsilon,\varepsilon) \\ \vspace{-2mm}\end{array}$} (A1);

        \node[loc] at (-2, 0) (B) {$a$} ;
        \node[loc] at (0,0) (B0) {$b$} ;
        \node[loc] at (2,0) (B1) {$b$} ;
        \node[loc] at (4,0) (B2) {$b$} ;
        \node[loc] at (6,0) (B3) {$\#$} ;
        \node[loc] at (8,0) (B4) {$b$} ;
        \node[loc] at (10,0) (B5) {$a$} ;
        \node[loc] at (12,0) (B6) {$\#$} ;
        \node[loc,dashed] at (13.5,0) (B7) {$\{a,b\}^\omega$} ;
        
        \draw[trans] (B) -- (B0);
        \draw[trans] (B0) -- (B1); 
        \draw[trans] (B1) -- (B2); 
        \draw[trans] (B2) -- (B3); 
        \draw[trans] (B3) -- (B4); 
        \draw[trans] (B4) -- (B5); 
        \draw[trans] (B5) -- (B6); 
        \draw[trans] (B6) -- (B7); 
    
        \node[loc] at (-2, -1) (C) {$a$} ;
        \node[loc] at (0,-1) (C0) {$b$} ;
        \node[loc] at (2,-1) (C1) {$b$} ;
        \node[loc] at (4,-1) (C2) {$b$} ;
        \node[loc] at (8,-1) (C4) {$b$} ;
        \node[loc] at (10,-1) (C5) {$a$} ;
        \node at (13.5,-1) (C7) {$~~~~~~$} ;
        
        \draw[trans] (C) -- (C0);
        \draw[trans] (C0) -- (C1); 
        \draw[trans] (C1) -- (C2); 
        \draw[trans] (C4) -- (C5); 
        
        \node[loc] at (-2, -2) (D) {$a$} ;
        \node[loc] at (0,-2) (D0) {$b$} ;
        \node[loc] at (2,-2) (D1) {$b$} ;
        \node[loc] at (4,-2) (D2) {$b$} ;
        \node[loc] at (8,-2) (D4) {$b$} ;
        \node[loc] at (10,-2) (D5) {$a$} ;
        \node at (13.5,-2) (D7) {$~~~~~~$} ;
        
        \draw[trans] (D) -- (C);
        \draw[trans] (D2) -- (D1);
        \draw[trans] (D1) -- (D0); 
        \draw[trans] (D0) -- (D); 
        \draw[trans] (D5) -- (D4);
        
        \node at (-2,-3) (E) {$~~$} ;
        \node[loc] at (6,-3) (E3) {$\#$} ;
        \node[loc] at (12,-3) (E6) {$\#$} ;
        \node[loc,dashed] at (13.5,-3) (E7) {$\{a,b\}^\omega$} ;
        
        \draw[->,rounded corners] (C2) --  (6, -1) --  (E3);
        \draw[trans] (D4) -- (C4);
        \draw[->, rounded corners] (E3) -- (10, -3) -- (D5);
        \draw[->, rounded corners] (C5) -- (12, -1) -- (E6);
        \draw[trans] (E6) -- (E7);

        \begin{pgfonlayer}{background}
          \node at (-2.2, 0) [label=left:\texttt{input} :] {};
          \node [background, fit=(C) (C7), label=left:\texttt{copy 1}:] {};
          \node [background, fit=(D) (D7), label=left:\texttt{copy 2}:] {};
          \node [background, fit=(E) (E7), label=left:\texttt{copy 3}:] {};
        \end{pgfonlayer}
        
      \end{tikzpicture}
    }
  \end{center}
  \vspace{-1em}
  \caption{Transformation $f_1$ given as (a) two-way
    transducers with look-ahead  (b) streaming string transducers with
    $F(\set{2}) = x z$ is the output associated with Muller set $\set{2}$, and
    (c) FO-definable transformation for the string
    $abbb\#ba\#\{a,b\}^\omega$. Here symbol $\alpha$ stands for
  both symbols $a$ and $b$, and the predicate $\reach$ is the lookahead
  that checks whether string contains a $\#$ in
  future. \label{fig:fo-example}}  
  \vspace{-1em}
\end{figure}

\noindent{\bf Logic based transformations.}
Logical descriptions of transformations of structures---as introduced by
Courcelle~\cite{Cour94}---work 
by introducing a fixed number of copies of the  vertices of the input graph; and
the domain, the labels and the edges of the output graph are defined by MSO
formulae with zero, one or two free variables, respectively, interpreted over
the input graph.  
Figure \ref{fig:fo-example}(c) shows a way to express transformation $f_1$ using
three copies of the input with a) logical formula $\phi_{\text{dom}}$  expressing
domain of the transformation, b) logical formulae $\phi^c_\alpha(i)$ (with one free
variable) for every copy $c \in \set{1, 2}$ and letter $\alpha \in \set{a, b}$
expressing the label of a position  $i$ for copy $c$,  and c) logical formulae
$\phi^{c, d}(i, j)$ with two free variables expressing the edge from position
$i$ of copy $c$ to position $j$ of copy $d$.
The formulae $\phi_\dom$, $\phi^c_a$, and $\phi^{c, d}$ are interpreted over
input structure (in this paper always an infinite string), and it is easy to see
that these formulae for our example can easily be expressed in MSO. 
In this paper we study logical transformations expressible with FO and to cover
a larger class of transformations, we use natural order 
relation $\prec$ for positions instead of the successor relation. 
We will later show that the transformation $f_1$ indeed can be expressed
using FO.

\noindent{\bf Two-Way Transducers.} 
\label{two-way}
For finite string transformations, Engelfriet and Hoogeboom~\cite{EH01} showed
that the finite-state transducers when equipped with a two-way input tape have
the same expressive power as MSO transducers, and Carton and Dartois~\cite{CD15}
recovered this result for FO transducers and two-way transducers with
aperiodicity restriction. 
A crucial property of two-way finite-state transducers exploited in these
proofs~\cite{EH01,CD15} is the fact that transitions capable
of regular (star-free) \emph{look-ahead}  (i.e., transitions that test the
whole input string against a regular property) do not increase the
expressiveness of regular (aperiodic) two-way transducers. 
However, this property does not hold in case of $\omega$-strings.
In Figure~\ref{fig:fo-example}(a), we show a two-way transducer characterizing
transformation $f_1$.
The transducer uses the lookahead $\reach$ to check if the remaining part
of the string contains a $\#$ in future.
A transition labeled $<\phi, \alpha \left | \beta, +1 \right. >$ of the  two-way
transducer should be read as: if the current position on the string
satisfies the look-ahead $\phi$ and the current symbol is $\alpha$ then
output symbol $\beta$ and move the input tape head to the right.
This transducer works by first checking if the string contains a $\#$ in the
future of the current position, if so it moves its head all the way to the
position before $\#$ and starts outputting the symbols in reverse, and when it
sees the end-marker or a $\#$ it prints the string before the $\#$; however, if
there is no $\#$ in future, then the transducer outputs the rest of the string. 
It is straightforward to verify that this transducer characterizes the
transformation $f_1$. 
However, in the absence of the look-ahead a
two-way transducer can not express this transformation. 

\noindent{\bf Streaming String Transducers.}
Alur and  {\v C}ern\'y~\cite{AC10,AC11} proposed a one-way finite-state
transducer model, called the \emph{streaming string transducers} (\sst{}), that
manipulates a finite set of string variables to compute its output, and showed
that they have same expressive power as MSO transducers. 
\sst{}, instead of appending symbols to the output tape, concurrently update all
string variables using a concatenation of string variables and output symbols in
a \emph{copyless} fashion, i.e. no variable occurs more than once in each
concurrent variable update.   
The transformation of a string is then defined using an output (partial)
function $F$ that associates states with a copyless concatenation of string
variables, s.t.\ if the state $q$ is reached after reading the string and
$F(q) {=} xy$, then the output string is the final valuation of $x$ concatenated
with that of $y$. 
\cite{FiliotTrivediLics12} generalized this by introducing
a Muller acceptance condition to give an \sst{} to characterize
$\omega$-transitions. 
Figure~\ref{fig:fo-example}(b) shows a streaming string transducer accepting the
transformation $f_1$.
It uses three string variables and concurrently prepends and/or appends these
variables in a copyless fashion to construct the output.
The acceptance set and the output is characterized by a Muller set (here
$\set{2}$ and its output $xz$), such that if the infinitely visiting states set
is $\set{2}$ then the output is limit of the values of the concatenation $xz$.
Again, it is easy to verify that \sst{} in Figure~\ref{fig:fo-example}(b) captures
the transformation $f_1$. 

\noindent{\bf Contributions and Challenges.}
Our main contributions include the definition of aperiodic streaming
string transducers and aperiodic two-way transducers, and the proof of the
following key theorem connecting FO and transducers for transformations of infinite strings. 

\begin{theorem}
  \label{thm:main}
  Let $F: \Sigma^\omega \to \Gamma^\omega$.
  Then the following assertions are equivalent:
  \begin{enumerate}
  \item
    $F$ is first-order definable.
  \item
    $F$ is definable by some aperiodic two-way transducer with
    star-free look-around.
  \item
    $F$ is definable by some aperiodic streaming string transducers. 
  \end{enumerate}
\end{theorem}

We introduce the notion of transition monoids for automata, \twst{}, and \sst{}
with the Muller acceptance condition; and recover the classical result proving
aperiodicty of a language using the aperiodicity of the transition monoid of its
underlying automaton.
The equivalence between  \fot{} and \twst{} with star-free
look-around (Section~\ref{sec:fot-2wst}), 
crucially uses the transition monoid 
with Muller acceptance, which is necessary 
 to show
aperiodicity of the underlying language of the $\twst{}$. 
On the other hand, while going from aperiodic  \sst{} to \fot{}
(Section~\ref{sec:sst-fot}), the main difficulty is the construction of the \fot{}
using the aperiodicity  of the \sst{}, and while going from \twst{} with star-free
look-around to \sst{} (Section~\ref{sec:2wst-sst}), the hard part is to
establish the aperiodicity of the  \sst{}.
Due to space limitation, we only provide key definitions and sketches of our
results---complete proofs and related supplementary material can be found in the
appendix. 

\section{Preliminaries}
\label{sec:prelims}
A finite (infinite) string over alphabet $\Sigma$ is a finite (infinite)
sequence  of letters from  $\Sigma$.
We denote by $\epsilon$ the empty string.
We write $\Sigma^*$ for the set of finite strings, $\Sigma^{\omega}$ for the set
of $\omega$-strings over $\Sigma$, and $\Sigma^\infty = \Sigma^* \cup
\Sigma^\omega$ for the set of finite and $\omega$-strings.  
A language $L$ over an alphabet $\Sigma$ is defined as a set of  
strings, i.e. $L \subseteq \Sigma^\infty$. 

For a string $s \in \Sigma^{\infty}$ we write $|s|$ for its length; 
note that $|s|=\infty$ for an $\omega$-string $s$.  
Let $\dom(s)=\{1,2,3,\dots,\}$ be the set of positions 
in $s$.
For all $i\in \dom(s)$ we write $s[i]$  for the $i$-th letter of the string $s$. 
For two $\omega$-strings $s,s' \in \Sigma^{\omega}$, we define the distance
$d(s,s')$  as $\frac{1}{2^j}$ where $j$=min$\{k  \mid s[k] \neq s'[k]\}$. We say
that a string $s \in \Sigma^{\omega}$ is the limit of a sequence  
$s_1, s_2, \dots$ of $\omega$-strings $s_i \in \Sigma^{\omega}$ if for every
$\epsilon >0$, there is an index $n_{\epsilon} \in \Nat$ such that  
for all $i \geq  n_{\epsilon}$, we have that $d(s,s_i) \leq \epsilon$.
Such a limit, if exists, is unique and is denoted as $s=\lim_{i \rightarrow
  \infty}s_i$.
For example, $b^{\omega}= \lim_{i \rightarrow \infty}b^ic^{\omega}$. 

\subsection{Aperiodic Monoids for $\omega$-String Languages}
A monoid $\Mm$ is an algebraic structure $(M, \cdot, e)$ with a non-empty set
$M$, a binary  operation $\cdot$, and an identity element $e \in M$ such that 
for all $x, y, z \in M$ we have that  $(x \cdot (y \cdot z)) {=} ((x \cdot y)
\cdot z)$, and $x \cdot e = e  \cdot x$ for all $x \in M$.
We say that a  monoid $(M, \cdot, e)$ is \emph{finite} if the set $M$ is finite.
A monoid that we will use in this paper is the \emph{free} monoid, 
$(\Sigma^*,\cdot, \epsilon)$, which has a set of finite strings  
over some alphabet $\Sigma$ with the empty string $\epsilon$ as the identity. 

We define the notion of acceptance of a language via monoids. 
A morphism (or homomorphism) between two monoids $\Mm = (M, \cdot, e)$ and $\Mm'
= (M', \times, e')$ is a mapping $h: M \rightarrow M'$ such that $h(e) = e'$ and
$h(x \cdot y) = h(x)  \times h(y)$.
Let $h: \Sigma^* \rightarrow M,$ be a morphism from free monoid $(\Sigma^*,
\cdot, \epsilon)$ to a finite monoid $(M, \cdot, e)$.
Two strings $u, v \in \Sigma^{*}$ are said to be similar with respect to
$h$ denoted $u \sim_h v$, if for some  
$n \in \mathbb{N} \cup \{\infty\}$, we can factorize $u, v$ as $u=u_1u_2 \dots u_n$ and       
$v=v_1 v_2 \dots v_n$ with $u_i, v_i \in \Sigma^+$ and $h(u_i)=h(v_i)$ for all $i$. 
Two $\omega$-strings are $h$-similar if we can find factorizations $u_1u_2
\dots$ and $v_1v_2 \dots$ such that $h(u_i) =h(v_i)$ for all $i$. Let $\cong$ be
the transitive closure of $\sim_h$.
$\cong$ is an equivalence relation.
Note that since $M$ is finite, the equivalence relation $\cong$ is of finite index. 
For $w \in \Sigma^{\infty}$ we define $[w]_h$ as the set $\set{u \mid u \cong
  w}$. 
We say that a morphism $h$ accepts a language $L \subseteq \Sigma^{\infty}$ if $w \in L$ implies $[w]_h
\subseteq L$ for all $w \in \Sigma^{\infty}$.

We say that a monoid $(M, ., e)$ is \emph{aperiodic}~\cite{Strau94} if there
exists $n \in \Nat$ such that for all $x \in M$, $x^n = x^{n+1}$. 
Note that for finite monoids, it is equivalent to require that for all $x\in M$,
there exists $n\in \Nat$ such that $x^n= x^{n+1}$.
A language $L \subseteq \Sigma^{\infty}$ is said to be aperiodic iff
it is recognized by some morphism to a finite and aperiodic monoid 
(See Appendix~\ref{app:example-aperiodic}).
 
 \subsection{First-Order Logic for $\omega$-String Languages}                 
A string $s \in \Sigma^{\omega}$ can be represented as a relational structure 
$\struc{s} {=} (\dom(s), \preceq^s, (L^s_a)_{a \in \Sigma})$, called the string
model of $s$, where 
  $\dom(s) = \set{1, 2, \ldots}$ is the set of positions in $s$, 
  $\preceq^s$ is a binary relation over the positions in $s$ characterizing the
  natural order, i.e. $(x, y) \in \preceq^s$ if $x \leq y$;
  $L^s_a$, for all $a \in \Sigma$, are the unary predicates that hold for the
  positions in $s$ labeled with the alphabet $a$, i.e., $L^s_a(i)$ iff 
  $s[i]=a$, for all $i\in \dom(s)$.
When it is clear from context we will drop the superscript $s$ from the
relations $\preceq^s$ and $L^s_a$.

Properties of string models over the alphabet $\Sigma$ can be formalized by
first-order logic denoted by $\FO(\Sigma)$. 
Formulas of $\FO(\Sigma)$ are built up from variables $x, y, \ldots$ ranging
over positions of string models along with \emph{atomic formulae} of the
form   
$x{=} y, x {\preceq} y$, and  $L_a(x)$  for all $a \in \Sigma$  
where formula $x{=}y$ states that variables $x$ and $y$ point to the same
position, the formula $x \preceq y$ states that position corresponding to
variable $x$ is not larger than that of $y$, 
and the formula $L_a(x)$ states that position $x$ has the label $a \in
\Sigma$.
Atomic formulae are connected with \emph{propositional connectives} $\neg$,
$\wedge$, $\lor$, $\to$, and \emph{quantifiers} $\forall$ and $\exists$ that
range over node variables and we use usual semantics for them.  
We say that a variable is \emph{free} in a formula if it does not occur in the
scope of some quantifier.  
A \emph{sentence} is a formula with no free variables.
We write $\phi(x_1, x_2, \ldots, x_k)$ to denote that  at most the variables
$x_1, \ldots, x_k$ occur free in $\phi$.
For a string $s\in \Sigma^*$ and for positions $n_1, n_2, \ldots, n_k \in
\dom(s)$ we say that $s$ with valuation $\nu = (n_1, n_2, \ldots, n_k)$
satisfies the formula $\phi(x_1, x_2, \ldots, x_k)$ and we write 
$(s, \nu) \models \phi(x_1, x_2, \ldots, x_k)$ or  $s \models
\phi(n_1, n_2, \ldots, n_k)$ if formula $\phi$ with $n_i$ as the
interpretation of $x_i$ is satisfied in the string model $\struc{s}$. 
The language defined by an FO sentence $\phi$ is $L(\phi) \rmdef \set{s \in \Sigma^{\omega} \::\:
  \struc{s} \models \phi}$. 
We say that a language $L$ is FO-definable if there is an FO sentence $\phi$
such that $L = L(\phi)$.
The following is a well known result. 
\begin{theorem}\cite{Strau94}
	\label{thm:fo-aperiodic}
	A language $L\subseteq \Sigma^*$ is \ensuremath{\FO}-definable iff it is
	aperiodic.
\end{theorem}

\subsection{Aperiodic Muller Automata for $\omega$-String Languages}
\label{ap-omega}
A deterministic Muller automaton (DMA) is a tuple $\Aa = (Q, q_0, \Sigma, \delta, F)$
where
$Q$ is a finite set of states, $q_0 \in Q$ is the initial state, $\Sigma$
is an input alphabet,
$\delta: Q \times \Sigma \to Q$ is a transition function,
and $F \subseteq 2^Q$ are the accepting (Muller) sets.
For states $q, q' \in Q$ and letter $a \in \Sigma$ we say that $(q, a, q')$ is a
transition of the automaton $\Aa$ if $\delta(q,a ) = q'$ and we write $q
\xrightarrow{a} q'$.  
We say that there is a run of $\Aa$ over a finite string $s = a_1 a_2 \ldots a_n
\in \Sigma^*$ from state $p$ to state $q$  if there is a finite sequence of
transitions $\seq{(p_0, a_1, p_1), (p_1, a_2, p_2), \ldots,  (p_{n-1}, a_n,
  p_n)} \in (Q \times \Sigma \times Q)^*$ with $p = p_0$ and $q = p_n$.
We write $L_{p, q}$ for the set of finite strings such that there is a run of
$\Aa$ over $w$ from $p$ to $q$.
We say that there is a run of $\Aa$ over an $\omega$-string $s = a_1 a_2 \ldots
\in \Sigma^\omega$ if there is a sequence of
transitions $\seq{(q_0, a_1, q_1), (q_1, a_2, q_2), \ldots} \in (Q \times \Sigma
\times Q)^\omega$. 
For an infinite run $r$, we denote by $\Omega(r)$ the set of states that occur
infinitely often  in $r$.
We say that an $\omega$-string $w$ is accepted by a Muller automaton
$\Aa$ if the run  of $\Aa$  on $w$ is such that $\Omega(r) \in F$ and we write
$L(\Aa)$ for the set of all $\omega$-strings accepted by $\Aa$. \\
A Muller automaton $\Aa$ is \emph{aperiodic} iff there exists some
$m {\geq} 1$ s.t. $u^m {\in} L_{p,q}$ iff  $u^{m+1} {\in} L_{p,q}$ for all $u \in
\Sigma^*$ and $p, q \in Q$. 
Another equivalent way to define aperiodicity is using the transition monoid, which, to the best of our knowledge, 
has not been defined in the literature for Muller automata. 
Given a DMA $\mathcal{A} {=} (Q, q_0, \Sigma, \Delta, \set{F_1, \ldots, F_n})$,  
we define the transition monoid $\Mm_\Aa {=} (M_\Aa, \times, \textbf{1})$ of $\mathcal{A}$ 
as follows: 
 $M_\Aa$ is a set of $|Q|\times|Q|$ square matrices over
$(\set{0, 1} \cup 2^Q)^n \cup \set{\bot}$. Matrix multiplication
$\times$ is defined for matrices in $M_\Aa$ with identity element $\textbf{1} \in M_\Aa$, 
where $\textbf{1}$ is the matrix whose diagonal entries are
$(\emptyset, \emptyset, \ldots, \emptyset)$ and non-diagonal entries are all
$\bot$'s. Formally,  $M_\Aa {=} \set{M_s \::\: s \in \Sigma^*}$ is defined using matrices
$M_s$ for strings $s\in \Sigma^*$ s.t. $M_s[p][q] {=} \bot$ if there is no
run from $p$ to $q$ over $s$ in $\Aa$. Otherwise, let $P$ be the set of states 
(excluding $p$ and $q$) witnessed in the unique run from $p$ to $q$.
Then $M_s[p][q] = (x_1, \dots, x_n) \in (\set{0, 1} \cup 2^Q)^n$ where  
(1) $x_i = 0$ iff $\exists r \in P \cup \{p,q\}$, $r \notin F_i$;
(2)  $x_i=1$ iff $P \cup \{p,q\}=F_i$,  and
(3) $x_i=P \cup \{p,q\}$ iff $P \cup \{p,q\} \subset F_i$.
It is easy to see that $M_\epsilon = \textbf{1}$, since $\epsilon$ takes a state to itself and nowhere else. 
The operator $\times$ is simply matrix multiplication for  matrices in $M_\Aa$,
however we need to define addition $\oplus$ and multiplication $\odot$ for
elements $(\set{0, 1} \cup 2^Q)^n \cup \set{\bot}$ of the matrices.
We have $\alpha_1 \odot \alpha_2 = \bot$ if $\alpha_1 = \bot$ or $\alpha_2 =
\bot$, and if $\alpha_1 = (x_1, \ldots, x_n)$ and $\alpha_2 = (y_1, \ldots,
y_n)$ then $\alpha_1 \odot \alpha_2 = (z_1, \ldots, z_n)$ s.t.:
\begin{equation}
\resizebox{0.7 \textwidth}{!}
{
\tag{$\star$}
\label{eq1:product}
$z_i =
\begin{cases}
 0 & \text{if $x_i = 0$ or $y_i = 0$}\\
 1 & \text{if $(x_i = y_i = 1)$ or if $(x_i, y_i \subset F_i \text { and } x_i \cup y_i = F_i)$}\\
 1 & \text{if ($x_i = 1$ and $y_i \subset F_i$) or ($y_i = 1$ and $x_i \subset F_i$)}  \\
 x_i \cup y_i & \text{if $x_i, y_i \subset F_i$ and $x_i \cup y_i \subset F_i$}
\end{cases}
$
}
\end{equation}

Due to determinism, we have that for every
matrix $M_s$ and every state $p$ there is at most one state $q$ such that
$M_s[p][q] \not = \bot$ and hence the only  addition rule we need to
introduce is $\alpha \oplus \bot = \bot \oplus \alpha = \alpha$. 
It is easy to see that $(M_\Aa, \times, \textbf{1})$ is a monoid (a proof is
deferred to the Appendix \ref{app:basic-muller}).
It is straightforward to see that a Muller automaton is aperiodic if and only if
its transition monoid is aperiodic. Appendix \ref{app:aper-equiv} gives a proof showing that a language $L \subseteq \Sigma^{\omega}$ is aperiodic iff 
there is an aperiodic DMA accepting it. 

\section{Aperiodic Transformations}
\label{sec:aperiodic}
In this section we formally introduce three models to express
FO-transformations, and prepare the machinery required to prove their expressive
equivalence in the rest of the paper.
\subsection{First-Order Logic Definable Transformations} 
\label{subsec:fotrans}
Courcelle~\cite{Cour94} initiated the study of structure transformations using
MSO logic.
His main idea was to define a transformation  
$(w, w') \in R$  by defining the string model of $w'$ using a finite number of
copies of positions of the string model of $w$. 
The existence of positions, various edges, and position labels are then given as
$\MSO(\Sigma)$ formulas. 
We study a restriction of his formalism to use first-order logic to express
string transformations. 

\begin{definition}
  An \emph{FO string transducer}  is a tuple
$T {=} (\Sigma, \Gamma, \phi_{\dom}, C, \phi_{\nodes}, \phi_{\preceq})$ where:
\begin{itemize}
\item
  $\Sigma$ and $\Gamma$ are finite sets of input and output alphabets;
\item 
  $\phi_\dom$ is a closed $\FO(\Sigma)$ formula characterizing the domain of
  the transformation;   
\item 
  $C {=} \set{1, 2, \ldots, n}$ is a finite index set;  
\item 
  $\phi_{\nodes} {=} \set{ \phi^c_\gamma(x) : c \in C \text{ and } \gamma \in
    \Gamma}$ is a set of $\FO(\Sigma)$ formulae with a free variable $x$; 
\item 
  $\phi_{\preceq} {=} \set{\phi^{c, d}_\preceq(x, y) : c, d \in C}$ is a set of
  $\FO(\Sigma)$ formulae with two free variables $x$ and $y$.   
\end{itemize}
The transformation $\inter{T}$ defined by $T$ is as follows.  
A string $s$ with $\struc{s} = (\dom(s), \preceq, (L_a)_{a \in \Sigma})$ is in
the domain of $\inter{T}$ if $s \models \phi_{\dom}$ and the output string $w$ with structure \\
$M = (D, \preceq^M, (L^M_\gamma)_{\gamma\in\Gamma})$ is such
that 
\begin{itemize}
\item 
  $D = \set{ v^c \::\: v \in \dom(s), c \in C \text{ and } \phi^c(v)}$ is the
  set of positions where $\phi^c(v) \rmdef
  \lor_{\gamma \in \Gamma} \phi^c_\gamma(v)$; 
\item
  $\preceq^M {\subseteq} D {\times} D$ is the ordering relation between
  positions and it is such that for $v, u \in  dom(s)$ and
  $c, d \in C$ we have that $v^c \preceq^M u^d$ 
  if $w \models \phi^{c, d}_\preceq(v, u)$; and
\item 
  for all $v^c \in D$ we have that $L_\gamma^M(v^c)$ iff $\phi^c_\gamma(v)$.    
\end{itemize}
\end{definition}

Observe that the output is unique and therefore FO transducers implement
functions. A string $s \in \Sigma^{\omega}$ can be represented by its 
string-graph with $dom(s)=\{i \in \mathbb{N}\}$, 
$\preceq=\{(i,j) \mid i \leq j\}$ and $L_a(i)$ iff $s[i] = a$ for all $i$.
From now on, we denote the string-graph of $s$ as $s$ only. 
We say that an FO transducer is a \emph{string-to-string} transducer if its
domain is restricted to string graphs and the output is also a string graph. 
We say that a string-to-string transformation is FO-definable if there exists an FO
transducer implementing the transformation. 
We write $\fot{}$ for the set of FO-definable string-to-string $\omega$-transformations. 

\begin{example}
  Figure~\ref{fig:fo-example}(c) shows a transformation for an \fot{} that
  implements the transformation  $f_1: \Sigma^* \{a,b\}^{\omega} \to \Sigma^{\omega}$, where 
  $\Sigma=\{a,b, \#\}$, by replacing every maximal $\#$ free string $u$ into
  $\overline{u}u$.  
 Let $\istring$ be an FO formula that
  defines a string that contains a $\#$, and let $\reach(x)$ be an FO formula
  that is true at a position which has a $\#$ at a later position.
   To define the \fot{} formally, 
  we have $\phi_{dom}=\istring$,  $\phi^1_{\gamma}(x) = \phi^2_{\gamma}(x) = L_{\gamma}(x) \wedge \neg
  L_{\#}(x) \wedge \reach(x)$, since 
  we only keep the non $\#$ symbols that can ``reach'' a $\#$ in the input
  string in the first two copies.
 $\phi^3_{\gamma}(x) = L_{\#}(x) \vee (\neg L_{\#}(x) \wedge \neg
  \reach(x))$, since we only keep the $\#$'s, and  the infinite suffix from
  where there are no $\#$'s.  
  The full list  of formulae $\phi^{i,j}$ can be seen in Appendix \ref{app:fot-example}.
%
  \end{example}

\subsection{Two-way Transducers (\twst{})}
\label{sec:2wst}
A  \twst{} is a tuple $T=(Q, \Sigma, \Gamma, q_0, \delta, F)$ where $\Sigma,
\Gamma$ are respectively the input and output alphabet,  $q_0$ is the initial
state, $\delta$ is the transition function and $F \subseteq 2^Q$ is the
acceptance set.  
The transition function is given by $\delta : Q \times \Sigma \rightarrow Q
\times \Gamma^* \times \{1,0,-1\}$.
A configuration of the \twst{} is a pair $(q,i)$ where $q \in Q$ and $i \in \Nat$
is the current position of the input string.  
A run $r$ of a \twst{} on a string $s \in \Sigma^{\omega}$ is a sequence of
transitions  
$(q_0,i_0{=}0) \xrightarrow{a_1/c_1,dir}  (q_1, i_1) \xrightarrow{a_2/c_2,dir}  (q_2,
i_2) \cdots$ where $a_i \in \Sigma$ is the input letter read and $c_i \in
\Gamma^*$ is the output string produced during a transition and $i_j$s are the positions updated during a transition for all $j \in dom(s)$.
$dir$ is the direction, $\{1,0,-1\}$.  
W.l.o.g. we can consider the outputs to be over $\Gamma \cup \{\epsilon\}$. 
The output $out(r)$ of a run $r$  is simply a concatenation of the individual outputs,
i.e. $c_1c_2 \dots \in \Gamma^{\infty}$.
We say that the transducer reads the whole string $s$ when
$\sup \set{i_n \mid 0 \leq n <|r|} {=} \infty$.
The output of $s$, denoted $T(s)$ is defined as $out(r)$  only if
$\Omega(r) \in F$ and $r$ reads the whole string $s$.    
We write $\inter{T}$ for the transformation captured by $T$.

\noindent{\bf Transition Monoid.}
The transition monoid  of a \twst{} $T=(Q, \Sigma, \Gamma, q_0, \delta, \{F_1, \dots, F_n\})$ is the transition monoid
of its underlying automaton.
However, since the \twst{} can read their input in both directions,
the transition monoid definition must allow for reading the string starting
from left side and leaving at the left (left-left) and similar other
behaviors (left-right, right-left and right-right).
Following~\cite{CD15}, we define the behaviors 
$\beh_{xy}(w)$ of a string $w$ for $x, y \in \{\ell,r\}$. 
$\beh_{\ell r}(w)$ is a set consisting of pairs $(p,q)$ of states such that 
starting in state $p$ in the left side of $w$ the transducer leaves $w$ in right
side in state $q$.  
In the example in figure \ref{fig:fo-example}(a), we have $\beh_{\ell r}(ab\#) =
\set{(t,t), (p,t), (q,t)}$ and $\beh_{r r}(ab\#) = \set{(q,t), (t,t), (p,q)}$.
Two words $w_1, w_2$ are ``equivalent'' if their left-left, left-right, right-left and
right-right behaviors  are same. That is, $\beh_{xy}(w_1)=\beh_{xy}(w_2)$ for 
$x, y \in \{\ell,r\}$. 
The transition monoid of $T$ is the conjunction of the 4 behaviors, which also
keeps track, in addition,  the set of states witnessed in the run, as shown for
the deterministic Muller automata earlier. For each string $w \in \Sigma^*$, $x, y \in \{\ell,r\}$, and states $p,q$,  
the entries of the matrix $M_{u}^{xy}[p][q]$ are of the form  $\bot$, if there is no run
from $p$ to $q$ on word $u$, starting from the side $x$ of $u$ and leaving it in
side $y$, and is $(x_1,\ldots x_n)$ otherwise, where 
$x_i$ is defined exactly as in section \ref{ap-omega}.
%
 For equivalent words $u_1,u_2$, we have $M_{u_1}^{xy}[p][q]= M_{u_2}^{xy}[p][q]$ for all $x,y \in \{\ell, r\}$ and states $p,q$. 
Addition and multiplication of matrices are defined as in the case of Muller automata. See Appendix \ref{app:2way} for more details.
Note that behavioral composition is quite complex, due to left-right movements. 
In particular, it can be seen from the example that $\beh_{\ell r}(ab\#a\#)=\beh_{\ell r}(ab\#)\beh_{\ell \ell}(a\#)\beh_{r r}(ab\#)\beh_{\ell r}(a\#)$. 
Since we assume that the \twst{} $T$ is deterministic and completely reads the input string $\alpha \in \Sigma^{\omega}$, we can find 
a unique factorization $\alpha= [\alpha_0\dots \alpha_{p_1}][\alpha_{p_1+1} \dots \alpha_{p_2}]\dots$ such that
the run of $\Aa$ on each $\alpha$-block progresses from left to right, and 
each $\alpha$-block will be processed completely. 
That is, one can find a unique sequence of states $q_{p_1}, q_{p_2}, \dots$
such that the \twst{} starting in initial state $q_0$ at the left of the block $\alpha_0\dots \alpha_{p_1}$ leaves it at the right in state $q_{p_1}$, starts the next block $\alpha_{p_1+1} \dots \alpha_{p_2}$ from the left  in state
$q_{p_1}$ and leaves it at the right in state $q_{p_2}$ and so on.  

We consider the languages $L^{xy}_{pq}$ for $x,y \in \{\ell,r\}$, where $\ell,r$
respectively stand for left and right.  
$L^{\ell \ell}_{pq}$ stands for all strings $w$ such that, starting at state $p$
at the left of $w$,  
one leaves the left of $w$ in state $q$. Similarly,  $L^{r\ell}_{pq}$ stands for all strings $w$ such that starting at the right of $w$ 
in state $p$, one leaves the left of $w$ in state $q$. In figure \ref{fig:fo-example}(a), note that starting on the right of $ab\#$ in state $t$, 
we leave it on the right in state $t$, while we leave it on the left in state $p$. 
So $ab\# \in L^{r r}_{tt}, L^{r \ell}_{tp}$.   
Also, $ab \# \in L^{r r}_{pq}$.  

  A  \twst{}  is said to be \emph{aperiodic} iff for all strings $u \in
  \Sigma^*$, all states $p,q$ and  $x,y \in \{l,r\}$, there exists some $m \geq 1$ such that
$u^m \in L_{pq}^{xy}$ iff  $u^{m+1} \in L_{pq}^{xy}$.

\noindent{\bf Star-Free Lookaround.}
We wish to introduce aperiodic \twst{} that are capable of capturing
FO-definable transformations. 
However, as we discussed earlier (see page \pageref{two-way} in the paragraph on two-way transducers) \twst{} without look-ahead are strictly less
expressive than MSO transducers. 
To remedy this we study aperiodic \twst{}s enriched with star-free look-ahead
(star-free look-back can be assumed for free).

An aperiodic $\twst{}$ with star-free look-around ($\twst{}_{\la}$) is a tuple 
$(T, A, B)$ where $A$ is an aperiodic Muller look-ahead automaton and $B$   is an aperiodic look-behind
automaton, resp.,  and $T = (\Sigma, \Gamma, Q, q_0, \delta, F)$ is an aperiodic  
\twst{} as defined earlier except that the transition function 
$\delta: Q \times Q_B \times \Sigma \times Q_A \to Q \times \Gamma
\times  \set{-1, 0, +1}$  may consult look-ahead and look-behind automata to
make its decisions. 
Let $s \in \Sigma^\omega$ be an input string, and 
$L(A, p)$ be the set of infinite strings accepted by $A$ starting in state $p$.
Similarly, let $L(B,r)$ be the set of finite strings accepted by $B$ starting in state $r$. 
We assume that $\twst_{\la}$ are \emph{deterministic} i.e. for every string $s \in
\Sigma^\omega$ and every input position $i \leq |s|$,  there is
exactly one state $p \in Q_A$ and one state $r \in Q_B$ such that 
$s(i)s(i+1)\ldots \in L(A, p)$ and $s(0)s(1)\ldots s(i-1) \in L(B, r)$.
If the current configuration is $(q, i)$ and 
$\delta(q, r, s(i), p) = (q', z, d)$ is a transition, such that the string
$s(i)s(i+1)\ldots \in L(A, p) $ and $s(0)s(1)\ldots s(i-1) \in L(B, r)$, then 
$\twst_{\la}$ writes $z \in \Gamma$ on the output tape and updates its
configuration to $(q', i+d)$. 
 Figure~\ref{fig:fo-example}(a) shows a \twst{} with star-free look-ahead
 $\reach(x)$ capturing the transformation $f_1$ (details in App.~\ref{app:2way}). 

 \subsection{Streaming $\omega$-String Transducers (SST)}\label{subsec:sst}
Streaming string transducers(\sst{}s)  manipulate a finite set of string variables to
compute their output.
In this section we introduce aperiodic \sst{}s for infinite strings. 
Let $\varsst$ be a finite set of variables and $\Gamma$ be a finite alphabet. 
A substitution $\sigma$ is defined as a mapping ${\sigma : \varsst \to (\Gamma \cup
  \varsst)^*}$. 
A valuation is defined as a substitution $\sigma: \varsst \to \Gamma^*$.
Let $\Ss_{\varsst, \Gamma}$ be the set of all substitutions $[\varsst \to (\Gamma \cup
\varsst)^*]$.
Any substitution $\sigma$ can be extended to $\hat{\sigma}: (\Gamma \cup \varsst)^*
\to (\Gamma \cup \varsst)^*$ in a straightforward manner.
The composition $\sigma_1 \sigma_2$  of two substitutions $\sigma_1$ and $\sigma_2$
is defined as the standard function composition  $\hat{\sigma_1} \sigma_2$,
i.e. $\hat{\sigma_1}\sigma_2(x) = \hat{\sigma_1}(\sigma_2(x))$ for all $x \in \varsst$. 
We say that a string $u \in (\Gamma \cup \varsst)^*$ is \emph{\linear{}} (or
linear) if each 
$x \in \varsst$ occurs at most once in $u$. 
A substitution $\sigma$ is \linear{} if $\hat{\sigma}(u)$
is \linear{}, for all linear $u \in (\Gamma \cup \varsst)^*$ .

\begin{definition} A \emph{streaming $\omega$-string transducer} (\sst{}) is a tuple 
$T = (\Sigma, \Gamma, Q, q_0, \delta, \varsst, \rho, F)$
\begin{itemize}
\item
  $\Sigma$ and $\Gamma$ are finite sets of input and output alphabets;
\item 
  $Q$ is a finite set of states with initial state $q_0$;
\item 
  $\delta:Q \times \Sigma \to Q$ is a transition function and
  $\varsst$ is a finite set of variables;
\item 
  $\rho : (Q \times \Sigma) \to \Ss_{\varsst, \Gamma}$ is a variable update
  function to \linear{} substitutions;
\item 
  $F: 2^Q \rightharpoonup \varsst^*$ is an output function such that for all $P
  \in \dom(F)$ the string $F(P)$ is \linear{} of form $x_1 \dots x_n$,  and for
  $q, q' \in P$ and $a \in \Sigma$ s.t.  $q' = \delta(q,  a)$ we have 
  \begin{itemize}
  \item 
    $\rho(q, a)(x_i) = x_i$ for all $i < n$ and 
     $\rho(q, a)(x_n) = x_n u$ for some $u \in (\Gamma \cup \varsst)^*$. 
  \end{itemize}
\end{itemize} 
\end{definition}
The concept of a run of an \sst{} is defined in an analogous manner to that of
a Muller automaton.
The sequence $\seq{\sigma_{r, i}}_{0 \leq i \leq |r|}$ of substitutions induced
by a run $r = q_0 \xrightarrow{a_1} q_1 \xrightarrow{a_2} q_2 \ldots$ is defined
inductively as the following: $\sigma_{r, i} {=} \sigma_{r, i{-}1} \rho(q_{i-1}, a_{i})$
for $0 < i \leq |r|$ and $\sigma_{r, 0} =   x \in X \mapsto \varepsilon$. 
The output $T(r)$ of an infinite run $r$ of $T$ is defined only if $F(r)$
is defined and equals 
$T(r) \rmdef \lim_{i \to \infty} \seq{\sigma_{r, i}(F(r))}$, when the limit exists. If not, we pad $\bot^{\omega}$ to the obtained finite 
string to get   
$\lim_{i \to \infty} \seq{\sigma_{r, i}(F(r)) \bot^\omega}$ 
as the infinite output string. 
%

The assumptions on the output function $F$ in the definition of an \sst{}
ensure that this limit always exist whenever $F(r)$ is defined.  
Indeed, when a run $r$ reaches a point from where it visits only states in $P$,
these assumptions enforce the successive valuations of $F(P)$ to be an
increasing sequence of strings by the prefix relation. 
The padding by unique letter $\bot$ ensures that the output is always an
$\omega$-string. 
The output $T(s)$ of a string $s$ is then defined as the output $T(r)$ of its
unique run $r$.
The transformation $\inter{T}$ defined by an SST $T$ is the partial function
$\set{(s, T(s)) \::\: T(s) \text{ is defined}}$. See Appendix \ref{eg:sst} for an example. 
We remark that for every \sst{} $T  = (\Sigma, \Gamma, Q, q_0, \delta, \varsst, \rho,
F)$, its domain is always an $\omega$-regular language defined by the 
Muller automaton $(\Sigma, Q, q_0, \delta, \dom(F))$, which can be constructed
in linear time.  
However, the range of an \sst{} may not be $\omega$-regular.
For instance, the range of the \sst-definable transformation $a^n\#^\omega \mapsto
a^nb^n\#^\omega$ ($n\geq 0$) is not $\omega$-regular.

\noindent{\bf Aperiodic Streaming String Transducers.}
We define the notion of aperiodic \sst{}s by introducing an appropriate notion
of transition monoid for transducers. 
The transition monoid of an \sst{} $T$ is based on the effect of a string $s$ on
the states as well as on the variables. 
The effect on variables is characterized by, what we call, flow information that
is given as  a relation that describes the number of copies of the
content of a given variable that contribute to another variable after reading a
string $s$.  

Let $T = (\Sigma, \Gamma, Q, q_0, \delta, \varsst, \rho, F)$ be an \sst{}.
Let $s$ be a string in $\Sigma^*$ and suppose that there exists a run $r$ of $T$
on $s$.
Recall that this run induces a substitution $\sigma_r$ that maps each variable
$X \in \varsst$ to a string $u \in (\Gamma \cup \varsst)^*$.
For string variables $X,Y\in \varsst$, states $p,q\in Q$, and $n \in
\Nat$ we say that  $n$ copies of $Y$ flow to $X$ from $p$ to $q$
if there exists a run $r$ on $s \in \Sigma^*$ from $p$ to $q$, and  $Y$ occurs
$n$ times in $\sigma_r(X)$.
We extend the notion of transition monoid for the Muller automata as defined in
Section~\ref{sec:prelims} for the transition monoid for SSTs to equip it with
variables.
Formally, the transition monoid $\Mm_T {=} (M_T, \times, \textbf{1})$ of an \sst{}
$T = (\Sigma, \Gamma, Q, q_0, \delta, \varsst, \rho, \set{F_1, \ldots, F_n})$  is
such that $M_T$ is a set of $|Q\times \varsst|\times|Q \times \varsst|$ square matrices over
$(\Nat \times (\set{0, 1} \cup 2^Q)^n) \cup \set{\bot}$ along with matrix multiplication
$\times$ defined for matrices in $M_T$ and identity element $\textbf{1} \in M_T$
is the matrix whose diagonal entries are
$(1, (\emptyset, \emptyset, \ldots, \emptyset))$ and non-diagonal entries are all
$\bot$'s.
Formally $M_T {=} \set{M_s \::\: s \in \Sigma^*}$ is defined using matrices
$M_s$ for strings $s\in \Sigma^*$ s.t. $M_s[(p, X)][(q, Y)] {=} \bot$ if there is no
run from state $p$ to state $q$ over $s$ in $T$, otherwise
$M_s[(p, X)][(q, Y)] = (k, (x_1, \dots, x_n)) \in (\Nat \times (\set{0, 1} \cup 2^Q)^n)$ where  
$x_i$ is defined exactly as in section \ref{ap-omega}, and 
$k$ copies of variable $X$ flow to variable $Y$ from state $p$
  to state $q$ after reading $s$. 
We write $(p,X) \rightsquigarrow^u_{\alpha} (q,Y)$ for
$M_u[(p, X)][(q, Y)] = \alpha$. 

It is easy to see that $M_\epsilon = \textbf{1}$.
The operator $\times$ is simply matrix multiplication for  matrices in $M_T$,
however we need to define addition $\oplus$ and multiplication $\odot$ for
elements $(\set{0, 1} \cup 2^Q)^n \cup \set{\bot}$ of the matrices.
We have $\alpha_1 \odot \alpha_2 = \bot$ if $\alpha_1 = \bot$ or $\alpha_2 =
\bot$, and if $\alpha_1 = (k_1, (x_1, \ldots, x_n))$ and $\alpha_2 = (k_2, (y_1, \ldots,
y_n))$ then $\alpha_1 \odot \alpha_2 = (k_1\times k_2,  (z_1, \ldots, z_n))$
s.t. for all $1 \leq i \leq n$ $z_i$ are defined as in~(\ref{eq1:product}) from
Section~\ref{ap-omega}.
Note that due to determinism of the \sst{}s we have that for every
matrix $M_s$ and every state $p$ there is at most one state $q$ such that
$M_s[p][q] \not = \bot$ and hence the only  addition rules we need to
introduce is $\alpha \oplus \bot = \bot \oplus \alpha = \alpha$,  $0 \oplus 0=0$, $1 \oplus 1=1$ and 
$\kappa \oplus \kappa=\kappa$ for $\kappa \subseteq Q$.  
It is easy to see that $(M_T, \times, \textbf{1})$ is a monoid and we give a
proof in Appendix~\ref{app:sst-basics}. We say that the transition monoid $M_T$ of an \sst{} $T$ is  $1$-bounded if in
all entries $(j, (x_1, \dots,x_n))$  of the matrices of $M_T$, $j \leq 1$.
 A streaming string transducer is \emph{aperiodic} if its transition monoid is
aperiodic. 

\section{FOTs $\equiv$ Aperiodic 2WST$_{sf}$ }
\label{sec:fot-2wst}
\begin{theorem}
  A transformation $f: \Sigma^\omega \to \Gamma^\omega$ is FOT-definable if and
  only if it is definable using  an aperiodic two way transducer with star-free look-around.
\end{theorem}
\begin{proof}[Proof (Sketch)]
  This proof is in two parts.
  \begin{itemize}
  \item {\bf Aperiodic \twst$_{sf}$ $\subseteq$  FOT.}
We first show that given an aperiodic \twst$_{sf}$ $\Aa$, we can effectively construct an \fot{} that captures the same transduction as $\Aa$ over infinite words. Let $\Aa=(Q, \Sigma, \Gamma, q_0, \delta, F)$ be an aperiodic  \twst$_{sf}$, where each transition outputs at most one letter. 
Note that this is without loss of generality, since we can output any longer string by having some extra states. 
 Given $\Aa$, we construct    
the \fot{} $T=(\Sigma, \Gamma, \phi_{dom}, C, \phi_{pos}, \phi_{\prec})$ that realizes the transduction of $\Aa$. 
The formula $\phi_{dom}$ specifies that the input graph is linear.
The formula $\phi_{dom}$ is the conjunction of formulae $\isstring$ and $\varphi$ 
where $\varphi$ is a FO formula  that captures the set of accepted strings of
$\Aa$ (obtained by proving $L(\Aa)$ is aperiodic, lemma\ref{cor-lem}, Appendix \ref{app:fot-2wst}) and $\isstring$ is a FO formula that specifies that the input graph is a
string (see Appendix \ref{app:fo}). 
The copies of the \fot{} are the states of $\Aa$. For any two positions $x,y$ of
the input string, and any two copies $q,q'$, we need to define
$\phi_{\prec}^{q,q'}$.
This is simply describing the behaviour of $\Aa$ on the substring from position
$x$ to position $y$ of the input string $u$, assuming at position $x$, we are in
state $q$, and reach state $q'$ at position $y$.
The following lemma (proof in Appendix~\ref{app:fo-behav}) gives an FO formula
$\psi_{q,q'}(x,y)$ describing this.

\begin{lemma}
  \label{fo-behav}
  Let $\Aa$ be an aperiodic \twst$_{sf}$ with the Muller acceptance condition. Then for all pairs of states $q,q'$, there exists an FO formula 
  $\psi_{q,q'}(x,y)$ such that for all strings $s \in \Sigma^{\omega}$ and a pair of positions $x,y$ of $s$, 
  $s \models \psi_{q,q'}(x,y)$ iff there is a run from state $q$ starting at position $x$ of $s$ that
  reaches position $y$ of $s$ in state $q'$.      
\end{lemma}

An edge exists between position $x$ of copy
$q$ and position $y$ of copy $q'$ iff  the input string $u \models\psi_{q,q'}(x,y)$.   
The formulae  $\phi_\gamma^q(x)$ for each copy $q$ specifies the output at position $x$ in state $q$. 
We have to capture that position $x$ is reached from the initial position in state $q$, and also the possible outputs
produced while in state $q$ at $x$. The transition function $\delta$ gives us these symbols. 
The formula  $\bigvee_{\delta(q,a)=(q',dir,\gamma)}L_a(x)$ captures the possible output symbols. 
To state that we reach $q$ at position $x$, we say $\exists y[\ifirst(y) \wedge \psi_{q_0,q}(y,x)]$. The conjunction of these two formulae 
gives $\phi_\gamma^q(x)$. 
This completes the  \fot{} $T$. 

\item
  {\bf FOT $\subseteq$ Aperiodic \twst$_{sf}${}.}
  Given an \fot{}, we show that we can construct an aperiodic \twst{} with star-free look-around capturing the same transduction over $\omega$-words. For this, we first show that
given an \fot{}, we can construct \twst{} enriched with FO instructions that captures the same transduction as the \fot{}. 
The idea of the proof follows \cite{EH01}, where one first defines an intermediate model of aperiodic \twst{} with FO instructions 
instead of look-around. Then we show $\fot{} \subseteq \twst{}_{fo}\subseteq \twst{}_{\la}$, to complete the proof. 
  \end{itemize}
The omitted details can be found in Appendix \ref{app:fo-wst}. 
\end{proof}

\section{Aperiodic SST $\subset$ FOT}
\label{sec:sst-fot}
\label{sec:2wst-sst}
\begin{lemma}
\label{lem:sst-fot}
A transformation is FO-definable if it is aperiodic-SST definable.
\end{lemma}
We show that every aperiodic $1$-bounded \sst{} definable transformation is  definable using FO-transducers.
A crucial component in the proof of this lemma is to show that the
variable flow in the aperiodic 1-bounded SST is FO-definable (see Appendix \ref{app:varflow}).
To construct the \fot{}, we make use of the 
 output structure for \sst{}.   
It is an intermediate representation of the output, and the transformation of any input string into its   
\sst-output structure will be shown to be FO-definable.
For any \sst{} $T$ and string $s\in\dom(T)$,  the \sst-output structure of $s$ is a relational structure 
$G_T(s)$ obtained by taking, for each variable $X\in\varsst$, two copies of $\dom(s)$, respectively denoted by $X^{in}$ and $X^{out}$.  
For notational convenience we assume that these structures are labeled on the edges.
A pair $(X,i)$ is \emph{useful} if the content of
variable $X$ before reading $s[i]$ will be part of the output after
reading the whole string $s$. 
This structure satisfies the following \emph{invariants}: for all $i\in
dom(s)$, $(1)$ the nodes $(X^{in}, i)$ and $(X^{out}, i)$ exist only if
$(X,i)$ is useful, and $(2)$ there is a directed path from $(X^{in}, i)$ to $(X^{out}, i)$ whose
labels are same as variable $X$ computed by $T$ after reading $s[i]$. 

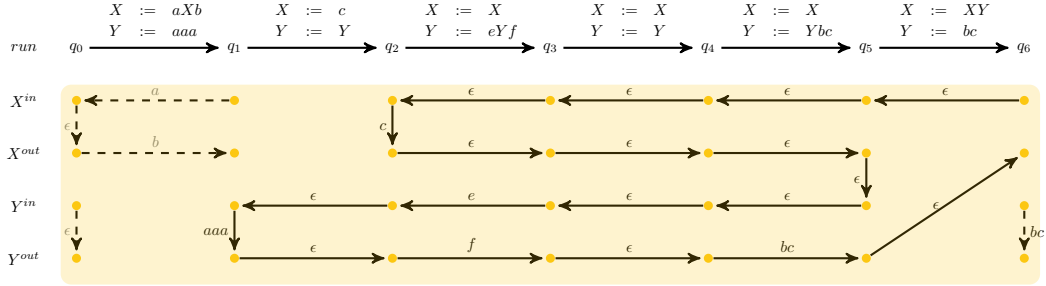
\begin{center}
\begin{tikzpicture}[->,>=stealth',shorten >=1pt,auto,node distance=1.5cm,
                    thick,scale=0.7, every node/.style={scale=0.6}]

\tikzstyle{graphnode}=[circle,fill=gold,thick,inner
sep=0pt,minimum size=2mm]

\tikzstyle{graphnodeblack}=[circle,fill=gold,thick,inner
sep=0pt,minimum size=2mm]

\node [graphnode] (x0in) at (0,5) {} ;
\node [graphnode] (x1in) at (3,5) {} ;
\node [graphnodeblack] (x2in) at (6,5) {} ;
\node [graphnodeblack] (x3in) at (9,5) {} ;
\node [graphnodeblack] (x4in) at (12,5) {} ;
\node [graphnodeblack] (x5in) at (15,5) {} ;
\node [graphnodeblack] (x6in) at (18,5) {} ;

\node [graphnode] (x0out) at (0,4) {} ;
\node [graphnode] (x1out) at (3,4) {} ;
\node [graphnodeblack] (x2out) at (6,4) {} ;
\node [graphnodeblack] (x3out) at (9,4) {} ;
\node [graphnodeblack] (x4out) at (12,4) {} ;
\node [graphnodeblack] (x5out) at (15,4) {} ;
\node [graphnodeblack] (x6out) at (18,4) {} ;

\node [graphnode] (y0in) at (0,3) {} ;
\node [graphnodeblack] (y1in) at (3,3) {} ;
\node [graphnodeblack] (y2in) at (6,3) {} ;
\node [graphnodeblack] (y3in) at (9,3) {} ;
\node [graphnodeblack] (y4in) at (12,3) {} ;
\node [graphnodeblack] (y5in) at (15,3) {} ;
\node [graphnodeblack] (y6in) at (18,3) {} ;

\node [graphnode] (y0out) at (0,2) {} ;
\node [graphnodeblack] (y1out) at (3,2) {} ;
\node [graphnodeblack] (y2out) at (6,2) {} ;
\node [graphnodeblack] (y3out) at (9,2) {} ;
\node [graphnodeblack] (y4out) at (12,2) {} ;
\node [graphnodeblack] (y5out) at (15,2) {} ;
\node [graphnodeblack] (y6out) at (18,2) {} ;

%
%
%

\node (xin) at (-1,5) {$X^{in}$};
\node (xout) at (-1,4) {$X^{out}$};
\node (yin) at (-1,3) {$Y^{in}$};
\node (yout) at (-1,2) {$Y^{out}$};

  \draw[->,dashed]   (x0in) -- node[left] {\textcolor{gray}{$\epsilon$}} (x0out) ;
  \draw[->,dashed]   (y0in) -- node[left] {\textcolor{gray}{$\epsilon$}} (y0out);


  \draw[->,dashed] (x1in) -- node[above] {\textcolor{gray}{$a$}} (x0in) ;
  \draw[->,dashed] (x0out) -- node[above] {\textcolor{gray}{$b$}} (x1out) ;



  \draw[->] (y1in) -- node[left] {$aaa$} (y1out) ;


  \draw[->] (y2in) -- node[above] {$\epsilon$} (y1in) ;
  \draw[->] (y1out) -- node[above] {$\epsilon$} (y2out) ;



  \draw[->] (x2in) -- node[left] {$c$} (x2out) ;
           

  \draw[->] (y3in) -- node[above] {$e$} (y2in) ;
  \draw[->] (y2out) -- node[above] {$f$} (y3out) ;

  \draw[->] (x3in) -- node[above] {$\epsilon$} (x2in) ;
  \draw[->] (x2out) -- node[above] {$\epsilon$} (x3out) ;



  

  \draw[->] (x4in) -- node[above] {$\epsilon$} (x3in) ;
  \draw[->] (x3out) -- node[above] {$\epsilon$} (x4out) ;

  \draw[->] (y4in) -- node[above] {$\epsilon$} (y3in) ;
  \draw[->] (y3out) -- node[above] {$\epsilon$} (y4out) ;
 
%


\draw[->] (x5in) -- node[above] {$\epsilon$} (x4in) ;
  \draw[->] (x4out) -- node[above] {$\epsilon$} (x5out) ;

\draw[->] (y5in) -- node[above] {$\epsilon$} (y4in) ;
  \draw[->] (y4out) -- node[above] {$bc$} (y5out) ;

\draw[->] (x6in) -- node[above] {$\epsilon$} (x5in) ;
  \draw[->] (x5out) -- node[left] {$\epsilon$} (y5in) ;

\draw[->,dashed] (y6in) -- node[right] {$bc$} (y6out) ;
 \draw[->] (y5out) -- node[left] {$\epsilon$} (x6out) ;

\node (run) at (-1,6) {$run$};

\node (q) at (0,6) {$q_0$} ;
\node (q1) at (3,6) {$q_1$} ;
\node (r) at (6,6) {$q_2$} ;
\node (q2) at (9,6) {$q_3$} ;
\node (r1) at (12,6) {$q_4$} ;
\node (r2) at (15,6) {$q_5$} ;
\node (r3) at (18,6) {$q_6$} ;


 \draw[->] (q) -- node[above] {$\begin{array}{lll} X & := & aXb
       \\ Y & := & aaa  \end{array}$}(q1) ;

 \draw[->] (q1) -- node[above] {$\begin{array}{lll} X & := & c
       \\ Y & := & Y  \end{array}$}(r) ;

 \draw[->] (r) -- node[above] {$\begin{array}{lll} X & := & X
       \\ Y & := & eYf  \end{array}$}(q2) ;

 \draw[->] (q2) -- node[above] {$\begin{array}{lll} X & := & X
       \\ Y & := & Y  \end{array}$}(r1) ;

 \draw[->] (r1) -- node[above] {$\begin{array}{lll} X & := & X
       \\ Y & := & Ybc  \end{array}$}(r2) ;

 \draw[->] (r2) -- node[above] {$\begin{array}{lll} X & := & XY
       \\ Y & := & bc  \end{array}$}(r3) ;



\fill[gold,fill opacity=0.2,rounded corners] (-0.3,1.5) rectangle (18.3,5.3);







\end{tikzpicture}
\end{center}


\noindent We define \sst-output structures formally in Appendix \ref{app:sst-output},
however, the illustration above shows an \sst-output structure. 
We show only the variable updates. 
Dashed arrows represent variable updates for useless variables, and therefore does not belong the \sst-output
structure. The path from $(X^{in},6)$ to $(X^{out},6)$ gives the contents of $X$ ($ceaaafbc$) 
after 6 steps. We write $O_T$ for the set of strings appearing in right-hand side of variable updates.

We next show that the transformation that maps an $\omega$-string $s$ into its output structure 
is FO-definable, whenever the SST is 1-bounded and aperiodic. Using the fact that variable flow is FO-definable,
we show that for any two variables $X,Y$, we can capture in FO, a path from $(X^d, i)$ to $(Y^e, j)$ 
  for $d, e\in \{in, out\}$ in $G_T(s)$ and all positions $i, j$.

 \begin{lemma}\label{lem:fopath}
    Let $T$ be an \textbf{aperiodic,1-bounded} \sst{} $T$. 
    For all $X,Y\in \varsst$ and all $d,d'\in
    \{in,out\}$, there exists an FO[$\Sigma$]-formula 
    $\text{path}_{X,Y,d,d'}(x,y)$ with two free variables
    such that for all strings $s\in\dom(T)$ and all positions
    $i,j\in dom(s)$, $s\models \text{path}_{X,Y,d,d'}(i,j)$ iff 
    there exists a path from $(X^d,i)$ to $(Y^{d'},j)$ in
    $G_T(s)$.
\end{lemma}
The proof of Lemma \ref{lem:fopath} is in Appendix \ref{app:fopath}. 
As seen in Appendix (in Proposition \ref{prop:fostates}) one can write a formula 
$\phi_q(x)$ (to capture the state $q$ reached) and formula $\psi^{Rec}_P$ (to capture the recurrence of a Muller set $P$) 
in an accepting run after reading a prefix. 
For each variable $X \in \mathcal{X}$, we have two copies $X^{in}$ and $X^{out}$ that serve as the copy set 
of the \fot{}. As given by the \sst{} output-structure, for each step $i$, state $q$ and symbol $a$, 
a copy is connected to copies in the previous step based on the updates $\rho(q,a)$. 
The full details of the \fot{} construction handling the Muller acceptance condition of the \sst{} are in Appendix \ref{app-fot-cons}.

 \section{Aperiodic $\twst_{\la}$ $\subset$  Aperiodic SST}
\label{sec:2wst-sst}
\label{sec:2wst-sst}
We show that given an aperiodic 2WST $\Aa=(\Sigma, \Gamma, Q,
q_0, \delta, F)$ with star-free look around over $\omega$-words, we can
construct an aperiodic \sst{} $\Tt$ that realizes the same transformation.
\begin{lemma}
  For every transformation  definable with an  aperiodic \twst{} with star-free
  look around, there exists an equivalent aperiodic 1-bounded \sst{}. 
\end{lemma}
\begin{proof} 
While the idea of the construction is  similar to ~\cite{FiliotTrivediLics12}, 
the main challenge is to eliminate the star-free look-around for infinite strings from the \sst{}, preserving aperiodicity.  
As an intermediate model we introduce streaming $\omega$-string transducers with
star-free look-around  $\mathsf{SST}_{\la}$  that can make transitions based on
some star-free property of the input string.
We first show that for every  aperiodic $\twst_{\la}$ one can obtain an aperiodic
$\mathsf{SST}_{\la}$, and then prove that the star-free look arounds can be
eliminated from the $\mathsf{SST}_{\la}$.   
\begin{itemize}

\item ($\twst_{\la} \subset \mathsf{SST}_{\la}$).
  One of the key observations in the construction is that a $\twst_{\la}$
  can move in either direction, while $\sst_{\la}$  cannot.  
  Since we start with a deterministic $\twst_{\la}$ that reads the entire input
  string,  it is clear that if a cell $i$ is visited in a state $q$, then we
  never come back to that cell in the same state.
  We keep track in each cell $i$, with current state $q$, the state  $f(q)$ the
  $\twst_{\la}$  will be in, when it moves into  cell $i+1$ for the first time. 
  The $\sst_{\la}$ will move from state $q$ in cell $i$ to state $f(q)$ in cell
  $i+1$, keeping track of the output produced in the interim time; that is, the
  output produced between $q$ in cell $i$ and $f(q)$ in cell $i+1$ must be
  produced by the $\sst_{\la}$ during the move. This output is stored in a variable $X_q$. 
  The state of the $\sst_{\la}$  at each point of time thus comprises   of a
  pair $(q, f)$ where $q$ is the current state of the $\twst_{\la}$, and $f$ is
  the function  which computes the state that $q$ will evolve into, when moving
  to the right, the first time.   
  In each cell $i$, the state of the $\mathsf{SST}$ will coincide with the state the
  $\twst_{\la}$ is in, when reading cell $i$ for the first time.  
 In particular, in the \sst$_{sf}$, we define $\delta'((q,f),r,a,p)=(f'(q),f')$ where 
 $f'(q)=f'(f(t))$ if in the $\twst_{\la}$ we have $\delta(q,r,a,p)=(t,\gamma,-1)$.
 $f'(q)$ gives the state in which the $\twst_{\la}$ will move to the right of the
 current cell,  but clearly this depends on $f(t)$, the state in which the $\twst_{\la}$
 will move to the right  
from the previous cell. The variables of the \sst$_{sf}$ are of the form $X_q$, where $q$ is the current state 
of the \sst$_{sf}$. Update of $X_q$ depends on whether the $\twst_{\la}$ moves left, right or stays 
in state $q$.  For example, $X_q$ is updated as 
$X_t \rho(X_{f(t)})$ if in the 2WST, $\delta(q,r,a,p)=(t, \gamma,-1)$ and $f(t)$ is defined. 
The definition is recursive, and $X_t$ handles the output produced from state $t$ in cell $i-1$.  
We allow all subsets of $Q$ as Muller sets of the $\mathsf{SST}_{\la}$, and 
keep any checks on these, as part of the look-ahead. 

A special variable $O$ is used to define the output of the Muller sets, by
simply  updating it as  $O:=O\rho(X_q)$ corresponding to the current state $q$
of the $\twst_{\la}$ (and $(q,f)$ is the state of the \sst$_{sf}$).  
The details of the correctness of  construction are in Appendix \ref{app-2wst-sst}.

  \item ($\mathsf{SST}_{\la} \subset \mathsf{SST}$).
An aperiodic $\mathsf{SST}$ with star-free lookaround is a tuple $(T,B, A)$
where $A = (P_A, \Sigma, \delta_A, P_f)$ is an aperiodic, deterministic
Muller automaton called a look-ahead automaton,  $B= (P_B,\Sigma, \delta_B)$ is an aperiodic
automaton called the  look-behind automaton, and  $T$ is a tuple 
$(\Sigma, \Gamma, Q, q_0, \delta, \varsst, \rho, F)$ where 
$\Sigma$, $\Gamma$, $Q$, $q_0$, $\varsst$, $\rho$, and $F$  are defined in the
same fashion as for $\omega$-\sst{}s, and $\delta: Q \times
P_B \times \Sigma \times P_A \to Q$ is the transition function. 
On a string $a_1a_2 \dots$, while processing symbol $a_i$,
we have in the $\mathsf{SST}_{\la}$,  $\delta((q, p_B, p_A),a_i)=q'$, 
(and the next transition is $\delta((q', p'_B, p'_A),a_{i+1})$) if (i) the prefix $a_1a_2 \dots a_i \in L(p_A)$,
(ii) the suffix $a_{i+1} a_{i+2} \dots \in L(p_B)$, where 
$L(p_A)$ ($L(p_B)$) denotes the language accepted starting in state
$p_A$ ($p_B$). We further assume that the look-aheads are mutually exclusive, i.e. for
all symbols $a\in\Sigma$, all states $q\in Q$, and all transitions
$q' = \delta (q,r, a,p)$ and $q'' = (q,r',a,p')$, we have that  
$L(A_p)\cap L(A_{p'})=\varnothing$ and  $L(B_r)\cap L(B_{r'})=\varnothing$.
 In Appendix \ref{app:unique}, we show that for any input string, there is atmost
one useful, accepting run in the $\mathsf{SST}_{\la}$, while  
in Lemma \ref{lem:aperiodicSSTLA} in Appendix \ref{app:unique}, we show that
adding (aperiodic) look-arounds  to $\mathsf{SST}$ does not increase their
expressiveness.

\end{itemize}
The proof sketch is now complete. \end{proof}

\vspace{-0.5em}
\section{Conclusion}
\vspace{-0.5em}
We extended the notion of aperiodicity from finite string transformations to
that on infinite strings.
We have shown a way to generalize transition monoids for deterministic Muller
automata to streaming string transducers and two-way finite state transducers
that capture the FO definable global transformations.  
A interesting and natural next step is to investigate LTL-definable transformations, their
connection with FO-definable transformations, and their practical applications in
verification and synthesis. 

\bibliography{papers}

\begin{thebibliography}{10}

\bibitem{ADT13}
R.~Alur, A.~Durand-Gasselin, and A.~Trivedi.
\newblock From monadic second-order definable string transformations to
  transducers.
\newblock In {\em LICS}, pages 458--467, 2013.

\bibitem{AD12}
Rajeev Alur and Loris D'Antoni.
\newblock {\em Automata, Languages, and Programming: 39th International
  Colloquium, ICALP 2012, Warwick, UK, July 9-13, 2012, Proceedings, Part II},
  chapter Streaming Tree Transducers, pages 42--53.
\newblock Springer Berlin Heidelberg, Berlin, Heidelberg, 2012.

\bibitem{FiliotTrivediLics12}
Rajeev Alur, Emmanuel Filiot, and Ashutosh Trivedi.
\newblock Regular transformations of infinite strings.
\newblock In {\em Proceedings of the 2012 27th Annual IEEE/ACM Symposium on
  Logic in Computer Science}, LICS '12, pages 65--74, Washington, DC, USA,
  2012. IEEE Computer Society.

\bibitem{AM09}
Rajeev Alur and P.~Madhusudan.
\newblock Adding nesting structure to words.
\newblock {\em J. ACM}, 56(3):16:1--16:43, May 2009.

\bibitem{AC11}
Rajeev Alur and Pavol \v{C}ern\'{y}.
\newblock Streaming transducers for algorithmic verification of single-pass
  list-processing programs.
\newblock In {\em Proceedings of the 38th Annual ACM SIGPLAN-SIGACT Symposium
  on Principles of Programming Languages}, POPL '11, pages 599--610. ACM, 2011.

\bibitem{AC10}
Rajeev Alur and Pavol Čern{\'y}.
\newblock {Expressiveness of streaming string transducers}.
\newblock In Kamal Lodaya and Meena Mahajan, editors, {\em IARCS Annual
  Conference on Foundations of Software Technology and Theoretical Computer
  Science (FSTTCS 2010)}, volume~8 of {\em Leibniz International Proceedings in
  Informatics (LIPIcs)}, pages 1--12, Dagstuhl, Germany, 2010. Schloss
  Dagstuhl--Leibniz-Zentrum fuer Informatik.

\bibitem{Bu60}
J.~R. B\"uchi.
\newblock Weak second-order arithmetic and finite automata.
\newblock {\em Zeitschrift f\"ur Mathematische Logik und Grundlagen der
  Mathematik}, 6(1--6):66--92, 1960.

\bibitem{Bu62}
J.~R. B\"uchi.
\newblock On a decision method in restricted second-order arithmetic.
\newblock In {\em Int. Congr. for Logic Methodology and Philosophy of Science},
  pages 1--11. Standford University Press, Stanford, 1962.

\bibitem{CD15}
Olivier Carton and Luc Dartois.
\newblock Aperiodic two-way transducers and fo-transductions.
\newblock In {\em 24th {EACSL} Annual Conference on Computer Science Logic,
  {CSL} 2015, September 7-10, 2015, Berlin, Germany}, pages 160--174, 2015.

\bibitem{Cour94}
B.~Courcelle.
\newblock Monadic second-order definable graph transductions: a survey.
\newblock {\em Theoretical Computer Science}, 126(1):53--75, 1994.

\bibitem{dg08SIWT}
V.~Diekert and P.~Gastin.
\newblock First-order definable languages.
\newblock pages 261--306, 2008.

\bibitem{EH01}
J.~Engelfriet and H.~J. Hoogeboom.
\newblock {MSO} definable string transductions and two-way finite-state
  transducers.
\newblock {\em ACM Trans. Comput. Logic}, 2:216--254, 2001.

\bibitem{Fil15}
Emmanuel Filiot.
\newblock {\em Logic and Its Applications: 6th Indian Conference, ICLA 2015,
  Mumbai, India, January 8-10, 2015. Proceedings}, chapter Logic-Automata
  Connections for Transformations, pages 30--57.
\newblock Springer Berlin Heidelberg, Berlin, Heidelberg, 2015.

\bibitem{FKT14}
Emmanuel Filiot, Shankara~Narayanan Krishna, and Ashutosh Trivedi.
\newblock First-order definable string transformations.
\newblock In {\em 34th International Conference on Foundation of Software
  Technology and Theoretical Computer Science, {FSTTCS} 2014, December 15-17,
  2014, New Delhi, India}, pages 147--159, 2014.

\bibitem{Ladner77}
{Richard E.} Ladner.
\newblock Application of model theoretic games to discrete linear orders and
  finite automata.
\newblock {\em Information and Control}, 33(4):281--303, 1977.

\bibitem{McN66}
R.~McNaughton.
\newblock Testing and generating infinite sequences by a finite automaton.
\newblock {\em Inform. Contr.}, 9:521--530, 1966.

\bibitem{McP71}
R.~McNaughton and S.~Papert.
\newblock Counter-free automata.
\newblock {\em M.I.T. Research Monograph}, 65, 1971.
\newblock With an appendix by William Henneman.

\bibitem{Perrin84}
Dominique Perrin.
\newblock {\em Mathematical Foundations of Computer Science 1984: Proceedings,
  11th Symposium Praha, Czechoslovakia September 3--7, 1984}, chapter Recent
  results on automata and infinite words, pages 134--148.
\newblock Springer Berlin Heidelberg, 1984.

\bibitem{Sch65}
M.P. Schuetzenberger.
\newblock On finite monoids having only trivial subgroups.
\newblock {\em Information and Control}, 8(2):190 -- 194, 1965.

\bibitem{Strau94}
H.~Straubing.
\newblock {\em Finite Automata, Formal Logic, and Circuit Complexity}.
\newblock Birkh{\"a}user, Boston, 1994.

\bibitem{Tho96}
W.~Thomas.
\newblock Languages, automata, and logic.
\newblock In {\em Handbook of Formal Languages}, pages 389--455. Springer,
  1996.

\bibitem{Tho97}
W.~Thomas.
\newblock Ehrenfeucht games, the composition method, and the monadic theory of
  ordinal words.
\newblock In {\em In Structures in Logic and Computer Science: A Selection of
  Essays in Honor of A. Ehrenfeucht, Lecture}, pages 118--143. Springer-Verlag,
  1997.

\bibitem{thomas79}
Wolfgang Thomas.
\newblock Star-free regular sets of $\omega$-sequences.
\newblock {\em Information and Control}, 42(2):148--156, 1979.

\end{thebibliography}

\newpage
\appendix
\centerline{\Large Appendix}
\section{Example of an Aperiodic Monoid Recognizing a Language}
\label{app:example-aperiodic} 
The language $L=(ab)^{\omega}$ is aperiodic, since it is recognized by the morphism $h: \{a,b\}^* \rightarrow M$ 
where $M$ is an aperiodic monoid.  
 $M=(\{1,2,3\}, ., 1)$ with 2.3=1, 3.2=3, 2.2=3, 3.3=3 and $x.1=1.x=x$ for all $x \in M$. 
 Define $h(\epsilon)=1, h(a)=2, h(b)=3$. Then $M$ is aperiodic, as $x^n=x^{n+1}$ for all $x \in M$ and $n \geq 3$. 
 It is clear that $h$ recognizes $L$.

\section{First-Order Logic: Examples}
\label{app:fo}
We define the following useful FO-shorthands. 
\begin{itemize}
\item
  $x \succ y \rmdef \neg (x \preceq y)$ and $x \prec y \rmdef (x \preceq y) \wedge
  \neg (x = y)$,
  \item $first(x)=\neg \exists y(y \prec x)$
  \item \isstring{} is defined as 
  \[
\forall x(~\uniquesucc(x) \wedge \uniquepred(x))\wedge \exists y[\ifirst(y) \wedge \forall z[\ifirst(z) \rightarrow z=y]] \wedge \forall x \exists y 
        [x \prec y \wedge x \neq y]
        \]
        where
 $\uniquesucc(x)=\{\exists y[x \prec y \wedge \neg \exists z[x \prec z \wedge z \prec y]] \wedge \exists y'[x \prec y' \wedge \neg \exists z[x \prec z \wedge z \prec y']] \rightarrow (y=y')\} \wedge  \exists y[x \prec y \wedge \neg \exists z[x \prec z \wedge z \prec y]]$
 and \\
 $\uniquepred(x)=\{\exists y[y \prec x \wedge \neg \exists z[y \prec z \wedge z \prec x]] \wedge \exists y'[y' \prec x \wedge \neg \exists z[y' \prec z \wedge z \prec x]] \rightarrow (y=y')\} \wedge  \exists y[y \prec x \wedge \neg \exists z[y \prec z \wedge z \prec x]]$ characterize that the position $x$ has unique predecessor and successor.
\end{itemize}
It is easy to see that a structure satisfying $\isstring$  characterizes a string.


\section{Proofs from Section \ref{ap-omega}}

\subsection{Transition Monoid of Muller Automata}
\label{app:basic-muller}

 We start with an example for a transition monoid. The Muller automaton given in  
figure \ref{fig:muller-ex1} has two Muller acceptance sets $\{q\}, \{r\}$.
Consider the strings $ab$ and $bb$. The transition monoids are
\[
M_{ab}=\begin{blockarray}{cccc}
& \matindex{$q$} & \matindex{$r$} & \matindex{$t$}  \\
\begin{block}{c(ccc)}
\matindex{$q$} & \bot & (0,0) & \bot \\
\matindex{$r$} & \bot & \bot & (0,0)  \\
\matindex{$t$} & (0,0) & \bot & \bot  \\
\end{block}
\end{blockarray} 
~~M_{bb}=\begin{blockarray}{cccc}
& \matindex{$q$} & \matindex{$r$} & \matindex{$t$}  \\
\begin{block}{c(ccc)}
\matindex{$q$} & (1,0) & \bot & \bot  \\
\matindex{$r$} & \bot & (0,0) & \bot  \\
\matindex{$t$} & \bot & \bot & (0,0)  \\
\end{block}
\end{blockarray} 
\]
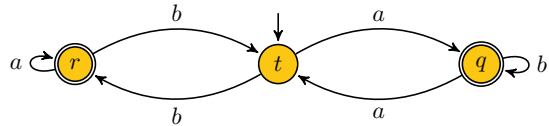
\begin{figure}[h]
  \begin{center}
    \begin{tikzpicture}[->,>=stealth',shorten >=1pt,auto,node distance=1cm,
	semithick,scale=0.9,every node/.style={scale=0.9}]
      \tikzstyle{every state}=[fill=gold,minimum size=1em]
      \node[state,fill=gold,accepting] at (8,-3) (C) {$q$} ;
      \node[state,fill=gold,accepting] at (2,-3) (B) {$r$} ;
      \node[state,initial,  initial where=above,initial text={},fill=gold] at (5,-3) (D) {$t$} ;
      \path(B) edge[loop left] node {$a $} (B);
      \path(C) edge[loop right] node {$b$} (C);
      \path(B) edge[bend left] node {$b$} (D);
      \path(D) edge[bend left] node {$b$} (B);
      \path(C) edge[bend left] node {$a$} (D);
      \path(D) edge[bend left] node {$a$} (C);
    \end{tikzpicture}
  \end{center}
  \caption{Muller accepting set = $\{\{q\},\{r\}\}$}
  \label{fig:muller-ex1}
\end{figure}

\begin{lemma}
	$(M_T,\times,\textbf{1})$ is a monoid, where $\times$ is defined as matrix
	multiplication and the identity element $\textbf{1}$ is the matrix with diagonal elements 
	$(\emptyset,\emptyset,\dots,\emptyset)$ and  all non-diagonal elements being 
	$\bot$. 
\end{lemma}
\begin{proof}
	
	Consider any matrix $M_s$ where $s \in \Sigma^*$. 
	Let there be $m$ states $\{p_1, \dots, p_m\}$ in the Muller automaton. 
	Consider a row corresponding to some $p_i$. Only one entry can be  different from $\bot$. 
	 Let this entry be $[p_i][p_j] = (\kappa_1, \dots, \kappa_n)$, where each $\kappa_h \in \{0,1,2^Q\}$, $1 \leq h \leq n$.  
	
	Consider $M_s \times \textbf{1}$. The $[p_i][p_j]$ entry 
	of the product are obtained from the $p_i$th row of $M_s$ and the $p_j$th column of $\textbf{1}$. 
	The $p_j$th column of $\textbf{1}$ has exactly one entry  $[p_j][p_j] = (\emptyset, \dots, \emptyset)$, while 
	all other elements are $\bot$.  
	Then the $[p_i][p_j]$ entry for the product matrix $M_s \times \textbf{1}$ is of the form 
	$\bot \oplus \dots \oplus \bot \oplus (\kappa_1, \dots, \kappa_n) \odot (\emptyset, \dots, \emptyset) \oplus \bot \oplus \dots \oplus \bot$. Clearly, this is equal to 
	$(\kappa_1, \dots, \kappa_n)$, since $\kappa \odot \emptyset=\kappa \cup \emptyset=\kappa$ and 
	$\bot \oplus \kappa=\kappa$. 
	
	 Similarly, it can be shown that the 
	$[p_i][p_j]$th entry of $M_s$ is preserved in   $\textbf{1} \times M_s$ as well.  Associativity of matrix multiplication 
	follows easily.
\end{proof}



We exploit the following lemma, proved in \cite{dg08SIWT}, in our proofs.
\begin{lemma}
	\label{lem-gd}
	Let $h: \Sigma^* \rightarrow M$ be a morphism to a finite monoid $M$ and let $w=u_0u_1\dots$ be an infinite word 
	with $u_i \in \Sigma^+$ for $i \geq 0$. Then there exist $s, e \in M$ and an increasing sequence $0< p_1 < p_2 < \dots$ 
	such that 
	\begin{enumerate}
		\item $se=s$ and $e^2=e$
		\item $h(u_0\dots u_{p_1-1})=s$ and $h(u_{p_i} \dots u_{p_j})=e$ for all $0 < i<j$. 
	\end{enumerate}
\end{lemma}

Using Lemma \ref{lem-gd}, we prove the following lemma, which is used in the proof of Lemma \ref{lem-conv}.

\begin{lemma}
	\label{lem:ap-rec}
	Let $\mathcal{A}$ be a DMA. 
	The mapping $h$ which maps any string $s$ to its
	transition matrix $M_s$,  is a morphism from $(\Sigma^*, ., \epsilon)$
	to $(M_T, \times, \textbf{1})$. Hence,  $h$ recognizes $L(\mathcal{A})$.
\end{lemma}
\begin{proof}
It is easy to see that $h$ is a morphism : $h(s_1s_2)$ is by definition, $M_{s_1s_2}$. This is clearly 
equal to $M_{s_1}M_{s_2}=h(s_1)h(s_2)$.  

	Let $w \in L(\Aa)$ and let 
	$w = w_1 w_2 w_3 w_4 \dots$ be a factorization of $w$, with $w_i \in \Sigma^+$ for all $i$. Consider another string $w' \in \Sigma^{\omega}$ with a factorization $w' = w'_1 w'_2 w'_3 w'_4 \dots$, such that
	$h(w_i) = h(w'_i)$ forall $i$. We have to show that $w' \in L(\Aa)$.  
	
	Since  $w \in L(\Aa)$ we know that after some position $i$, the run $r$ of $w$ will contain only and all states from some muller set $F_k$. 
	Assume that the factorization of $w$ is such that, 
	 after the factor $w_i$ of $w$, the run visits only states from $F_k$.  
	 	 We can write $w$ as $u v_1 v_2 v_3 \dots$ such that
	$u = w_1 w_2 \dots w_i$, 
	$v_1 = w_{i+1} w_{i+2} \dots w_{i+j_1}$, 
	$v_2 = w_{i+j_1+1} w_{i+j_1+2} \dots w_{i+j_1+j_2}$ and so on, such that 
	each of $v_1, v_2, \dots $  witness all states of $F_k$ in the run.
	Since $h(w_i) = h(w'_i)$  forall $i$, we know that $M_{w_i}=M_{w'_i}$ for all $i$. 
	So we can factorize $w'$ as $u'v'_1v'_2 \dots$ such that $h(u)=h(u'), h(v_i)=h(v'_i)$ for all $i$. 
	Hence, $w'$ also has a run that also witnesses the muller set $F_k$ from some point onwards. Hence, $w' \in L(\Aa)$.
	Thus, $h$ recognizes $L(\Aa)$ since $w \in L(\Aa) \rightarrow [w]_h \subseteq  L(\Aa)$. 
\end{proof}

\subsection{Aperiodicty of DMA $\Aa$ $\equiv$ Aperiodicity of $L(\Aa)$}
\label{app:aper-equiv}
Note that a result similar to Theorem \ref{muller-ap} (whose proof is below) has been proved for
B\"uchi automata in \cite{dg08SIWT}, where  aperiodicity of B\"uchi automata was
defined using transition monoids. 

\begin{theorem}
  \label{muller-ap}
  A language $L \subseteq \Sigma^{\omega}$ is aperiodic iff there exists an
  aperiodic Muller automaton $\Aa$ such that $L = L(\Aa)$. 
\end{theorem}
\begin{proof}

We obtain the proof of Theorem \ref{muller-ap} by proving the following two lemmas. 

%
%
\begin{lemma}
	\label{lem-forward}
	Let $L \subseteq \Sigma^{\omega}$ be an aperiodic language. Then there is an aperiodic Muller automaton 
	accepting $L$.
\end{lemma}

\begin{proof}
	Let $L \subseteq \Sigma^{\omega}$ be an aperiodic language. Then by definition, $L$ is recognized by a morphism 
	$h: \Sigma^* \rightarrow M$ where $M$ is a finite aperiodic monoid. We first show that 
	we can construct a counter-free Muller automaton $\Aa$ such that $L=L(\Aa)$.  
	A counter-free automaton is one having the property that $u^m \in L_{pp} \rightarrow u \in L_{pp}$ for all $u \in \Sigma^*$, states $p$  and 
	$m \geq 1$.

	It can be shown \cite{dg08SIWT} that the aperiodic language $L$ can be written as the finite union of languages 
	$UV^{\omega}$ where $U, V$ are aperiodic languages of finite words. Moreover,
	for any $u_0u_1u_2 \dots \in \Sigma^{\omega}$, there is an increasing sequence  $0<p_1<p_2 \dots$ of natural numbers such that
	for the morphism $h: \Sigma^* \rightarrow M$ recognizing $L$, we have $h(u_0\dots u_{p_1})=s \in M$ 
	and $h(u_{p_i}\dots u_{p_j})=e \in M$ for all $0<i<j$.  
	The $e \in M$ is an idempotent element in $M$. Then we have 
	$h^{-1}(e)=V$ and $h^{-1}(s)=U$.

	Since $V$ is an aperiodic language $\subseteq \Sigma^*$, there is a minimal DFA 
	$\Dd$ that accepts $V$.  The initial state  of this DFA is $[\epsilon]$, and 
	the states are of the form $[x], x \in \Sigma^*$, such that 
	$xw \in L$ iff there is a run from $[x]$ on $w$ to an accepting state. 
	Accepting states have the form $[w]$ with $w \in V$ and the transition function 
	is $\delta([x],a)=[xa]$. 	 For all $x,y,v \in \Sigma^*$, 
	($xv \in V$ iff $yv \in V$) $\Rightarrow [x]=[y]$. We first show that $\Dd$ is counter-free. 
	
	Since $V$ is aperiodic, there is a morphism $g: V \rightarrow M_V$  recognizing $V$, for an aperiodic monoid $M_V$. 
		As $M_V$ is aperiodic, there exists some $m \geq 1$ such that for all 
	$x \in M_V$, $x^m=x^{m+1}$. If $[u]=[uv^m]$ for some $u, v \in \Sigma^*$, then $g(u)=g(uv^m)=g(uv^{m+1})=g(uv^mv)=g(uv)$.
	If $[u] \neq [uv]$, then we can find some string $w$ such that $uw \in V$ but $uvw \notin V$ or viceversa.
	This contradicts the hypothesis that $V$ is recognized by a morphism 
	$g: \Sigma^* \rightarrow M_V$, since $uw \in V$ and $g(uw)=g(uvw)$ implies either both $uw, uvw$ belong to  $V$ or neither.  
	Thus,  $[u]=[uv^m]$ implies $[u]=[uv]$ for all $u, v \in \Sigma^*$. That is, whenever
	$\delta^*([u],v^m)=[uv^m]=[u]$, we have $\delta^*([u],v)=[uv]=[u]$ as well, which shows 
	that the minimal DFA $\Dd$ for $V$ is counter-free.  
	
	If we intepret $\Dd$ as a Buchi automaton,  
	and consider $\alpha \in L(\Dd)$,  then $\alpha$ has infinitely many prefixes $v_1 < v_2 < \dots$ such that 
	each $v_i \in V$. We have to show that $\alpha \in V^{\omega}$; that is 
	we have to show that $\alpha=\alpha_1\alpha_2 \dots$ with $\alpha_i \in V$ for all $i$. 
	We know $v_2=v_1v'_1$, $v_3=v_2v'_2 \dots$ with $v_1, v_2, \dots \in V$. 
	Then $\alpha=v_1v'_1v'_2v'_3 \dots$. If $v'_i \in V$ for all $i$, we are done, since 
	in that case, $\Dd$ will be a counter-free Buchi automaton for $V^{\omega}$.

	To obtain the infinitely many prefixes $v_1<v_2<v_3<\dots$ such that
	$v_{i+1}=v_iv'_i$ with $v_i, v'_i \in V$, we consider the language
	$W=V.Prefree(V)$ where $Prefree(V)$ is the set of all strings in $V$ 
	which do not contain a proper prefix also lying in $V$. Since $V$ is aperiodic, 
	$W$ is also aperiodic. Let $\Ee$ be the minimal counter free DFA accepting $W$.
	We intrepret $\Ee$ as a Buchi automaton, as we did for the case of $\Dd$, 
	and show that $\Ee$ is a counter-free Buchi automaton accepting $V^{\omega}$.
	Clearly, $L(\Ee)$ is the set of strings which has infinitely many prefixes from $W$.

	If $w \in V^{\omega}$, then $w \in L(\Ee)$ since $w$ has infinitely many prefixes from $W$. 
	Conversely, let $w \in \Sigma^{\omega}$ be such that infinitely many prefixes 
	$w_1<w_2<w_3< \dots$ of $w$ are in $W$. Then we have to show that $w \in V^{\omega}$.  
	Let $w_i=x_iy_i$ where $x_i \in V, y_i \in Prefree(V)$ for each $i$. Then we have 
	$$x_1 < x_1y_1 < x_2 < x_2y_2 < x_3 < \dots$$ 
	Note that if $w_1\neq w_2$ that is, $x_1y_1 \neq x_2y_2$ and 
	$x_1=x_2$, then $y_1$ is a prefix of $y_2$. 
	Since $y_1, y_2 \in Prefree(V)$, this is not possible. 
	Thus, for any two $w_i \neq w_j$, we have $x_i <  x_j$. 
	Let $x_{i+1}=x_iy_iy'_i$ for some $y'_i$. Recall that,  using the morphism 
	$h: \Sigma^* \rightarrow M$ recognizing $L$, we have 
	$h^{-1}(e)=V$ for some idempotent $e \in M$. 
	$h(x_{i+1})= h(x_iy_iy'_i)=e.e.h(y'_i)=e.h(y'_i)=h(y_iy'_i)$. 
	Hence we get $w=x_1y_1y'_1y_2y'_2y_3y'_3\dots=x_1x_2x_3\dots \in V^{\omega}$.  
	Thus, $\Ee$ is a counter-free Buchi automaton accepting $V^{\omega}$. 
	Let $\Ee'$ be the counter-free Muller automaton obtained from $\Ee$.
	
	Since $U \subseteq \Sigma^*$ is aperiodic, we can construct as for $V$, the minimal counter-free DFA $\Uu$ for $U$. The concatenation 
	$\Uu \Ee'$ is then counter-free.  The finite union of such automata are also counter-free.   
	   
	Finally we show that counter-free automata are  aperiodic. Let $x^{n} \in L_{pq}$ for any two states 
	$p,q$,  for a large $n$. We can decompose $x^n$ as 
	$x^{k+l+m}$ such that $x^k \in L_{ps}, x^l \in L_{ss}$ and $x^m \in L_{sq}$, with $l \geq 2$. 
	Then we have $x \in L_{ss}$ by the counter-freeness. Then we obtain $x^{n-1}\in L_{pq}$. Similarly, we can show that $x^{n-1} \in L_{pq} \Rightarrow x^n \in L_{pq}$.  This shows that we have an aperiodic Muller automata accepting the finite union $UV^{\omega}$ that represents $L$. 
	Thus, starting from the assumption that $L$ is an aperiodic language, we have obtained an aperiodic Muller automaton that accepts $L$.   
\end{proof}

\begin{lemma}
	\label{lem-conv}
	Let $\Aa$ be Muller automaton whose transition monoid is aperiodic. Then $L(\Aa)$ 
	is aperiodic.
\end{lemma}
\begin{proof}
	From Lemma \ref{lem:ap-rec}, we know that we can construct a morphism mapping $(\Sigma^*, ., \epsilon)$ to 
	the transition monoid of $\Aa$ which recognizes $L(\Aa)$. Hence, $L(\Aa)$ is aperiodic.
\end{proof}

\end{proof}

\section{Example of FOT}
\label{app:fot-example}

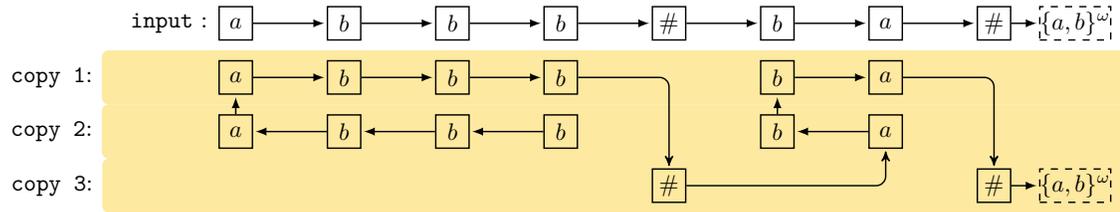
\begin{figure}[h]
  \tikzstyle{trans}=[-latex, rounded corners]
  \begin{center}
    \scalebox{0.9}{
      
      \begin{tikzpicture}[->,>=stealth',shorten >=1pt,auto, semithick,scale=.8]
        \tikzstyle{every state}=[fill=golden]

        \node[loc] at (-2, 0) (B) {$a$} ;
        \node[loc] at (0,0) (B0) {$b$} ;
        \node[loc] at (2,0) (B1) {$b$} ;
        \node[loc] at (4,0) (B2) {$b$} ;
        \node[loc] at (6,0) (B3) {$\#$} ;
        \node[loc] at (8,0) (B4) {$b$} ;
        \node[loc] at (10,0) (B5) {$a$} ;
        \node[loc] at (12,0) (B6) {$\#$} ;
        \node[loc,dashed] at (13.5,0) (B7) {$\{a,b\}^\omega$} ;
        
        \draw[trans] (B) -- (B0);
        \draw[trans] (B0) -- (B1); 
        \draw[trans] (B1) -- (B2); 
        \draw[trans] (B2) -- (B3); 
        \draw[trans] (B3) -- (B4); 
        \draw[trans] (B4) -- (B5); 
        \draw[trans] (B5) -- (B6); 
        \draw[trans] (B6) -- (B7); 
    
        \node[loc] at (-2, -1) (C) {$a$} ;
        \node[loc] at (0,-1) (C0) {$b$} ;
        \node[loc] at (2,-1) (C1) {$b$} ;
        \node[loc] at (4,-1) (C2) {$b$} ;
        \node[loc] at (8,-1) (C4) {$b$} ;
        \node[loc] at (10,-1) (C5) {$a$} ;
        \node at (13.5,-1) (C7) {$~~~~~~$} ;
        
        \draw[trans] (C) -- (C0);
        \draw[trans] (C0) -- (C1); 
        \draw[trans] (C1) -- (C2); 
        \draw[trans] (C4) -- (C5); 
        
        \node[loc] at (-2, -2) (D) {$a$} ;
        \node[loc] at (0,-2) (D0) {$b$} ;
        \node[loc] at (2,-2) (D1) {$b$} ;
        \node[loc] at (4,-2) (D2) {$b$} ;
        \node[loc] at (8,-2) (D4) {$b$} ;
        \node[loc] at (10,-2) (D5) {$a$} ;
        \node at (13.5,-2) (D7) {$~~~~~~$} ;
        
        \draw[trans] (D) -- (C);
        \draw[trans] (D2) -- (D1);
        \draw[trans] (D1) -- (D0); 
        \draw[trans] (D0) -- (D); 
        \draw[trans] (D5) -- (D4);
        
        \node at (-2,-3) (E) {$~~$} ;
        \node[loc] at (6,-3) (E3) {$\#$} ;
        \node[loc] at (12,-3) (E6) {$\#$} ;
        \node[loc,dashed] at (13.5,-3) (E7) {$\{a,b\}^\omega$} ;
        
        \draw[->,rounded corners] (C2) --  (6, -1) --  (E3);
        \draw[trans] (D4) -- (C4);
        \draw[->, rounded corners] (E3) -- (10, -3) -- (D5);
        \draw[->, rounded corners] (C5) -- (12, -1) -- (E6);
        \draw[trans] (E6) -- (E7);

        \begin{pgfonlayer}{background}
          \node at (-2.2, 0) [label=left:\texttt{input} :] {};
          \node [background, fit=(C) (C7), label=left:\texttt{copy 1}:] {};
          \node [background, fit=(D) (D7), label=left:\texttt{copy 2}:] {};
          \node [background, fit=(E) (E7), label=left:\texttt{copy 3}:] {};
        \end{pgfonlayer}
        
      \end{tikzpicture}
    }
  \end{center}
  \vspace{-1em}
  \caption{Transformation $f_1$ given as 
   FO-definable transformation for the string
    $abbb\#ba\#\{a,b\}^\omega$. \label{fig:app-fo-example}}  
  \vspace{-1em}
\end{figure}

  We give the full list of $\phi^{c,d}$ here. Let
  $btw(x,y,z)=(y \prec z \prec x) \vee (x \prec z \prec y)$ be a shorthand that says that $z$ lies between $x, y$.  
 Let $btw(x,y,\gamma)=\exists z(L_{\gamma}(z) \wedge btw(x,y,z))$ for $\gamma \in \Gamma$ and 
  let $reach_{\#}(x)=\exists y(x \prec y \wedge L_{\#}(y))$ be a shorthand which says there is a $\#$ that is ahead of $x$. 
  
  \begin{enumerate}
  \item   $\phi_{dom}=\istring$,  
  \item $\phi^1_{\gamma}(x) = \phi^2_{\gamma}(x) = L_{\gamma}(x) \wedge \neg
  L_{\#}(x) \wedge \reach(x)$, since 
  we only keep the non $\#$ symbols that can ``reach'' a $\#$ in the input
  string in the first two copies.
 \item $\phi^3_{\gamma}(x) = L_{\#}(x) \vee (\neg L_{\#}(x) \wedge \neg
  \reach(x))$, since we only keep the $\#$'s, and  the infinite suffix from
  where there are no $\#$'s.
  \end{enumerate}
  The transitive closure of the output successor relation is given by : 
  \begin{enumerate}
  \item  $\phi^{1,1}_{\prec}(x,y)=(x \prec y)=\phi^{3,3}_{\prec}(x,y)$,
 \item   $\phi^{2,2}_{\prec}(x,y)= [\neg btw(x,y,\#) \rightarrow (y \prec x) ]$
 $\wedge [btw(x,y,\#) \rightarrow (x \prec y)]$
   since we are reversing the arrows within a $\#$-free block from which a $\#$
  is reachable.     
   \item $\phi^{1,3}_{\prec}(x,y)=L_{\#}(y) \wedge (x \prec y)=\phi^{2,3}_{\prec}(x,y)$
 \item $\phi^{1,2}_{\prec}(x,y)=x \prec y \wedge btw(x,y,\#)$ 
 since  each position in a $\#$-free block is related to  each position 
  in a $\#$-free block  that comes later.   
 \item $\phi^{3,2}_{\prec}(x,y)$ and $\phi^{3,1}_{\prec}(x,y)$ are given by 
  $L_{\#}(x) \wedge (x \prec y)$
  \item  $\phi^{2,1}_{\prec}(x,y)$ expresses
  that each position $x$ in the $i$-th $\#$-free block is related to  
  position $y$ appearing in the $j$-th $\#$-free block, $j > i$. Within a \#-free block, the arrows are reversed. This translates simply to \\
  $[x \prec y \wedge btw(x,y,\#)] \vee \{\neg btw(x,y,\#) \wedge [(y \prec x) \vee y=x]\}$. 

  \end{enumerate}

\section{More details for section \ref{sec:2wst}}
\label{app:2way}

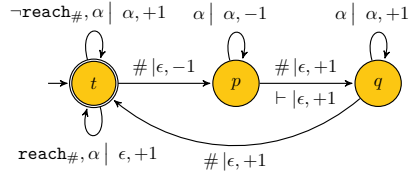
\begin{figure}[h]
  \tikzstyle{trans}=[-latex, rounded corners]
  \begin{center}
    \scalebox{0.7}{
      
      \begin{tikzpicture}[->,>=stealth',shorten >=1pt,auto, semithick,scale=.9]
        \tikzstyle{every state}=[fill=gold]
        
        \node[state] at (0,4) (A) {$p$} ;
        \node[state] at (3,4) (B) {$q$} ;
        \node[state,initial,initial where=left,initial text={}, accepting] at (-3,4) (C) {$t$} ;
        
        \path(A) edge[loop above] node
             {$\alpha\left|\begin{array}{lll} \alpha, -1 \end{array}\right.$}  (A);

             \path(A)  edge node[above] {$\# \left|
               \epsilon, +1 \right.$} node[below] {$\vdash \left| \epsilon,
               +1 \right.$} 
             (B);

             \path(B) edge[loop above] node
                  {$\alpha\left|\begin{array}{llllllll} \alpha, +1 \end{array}\right.$}
                  (B);
                  
                  \path(B) edge [bend left=40] node {$\# \left| \epsilon, +1 \right.$} (C);
                
                \path(C) edge[loop above] node
                     {$\neg \reach, \alpha\left|\begin{array}{llllllll} \alpha, +1 \end{array}\right.$} (C);
                     \path(C) edge[loop below] node
                          {$\reach, \alpha\left|\begin{array}{llllllll} \epsilon, +1 \end{array}\right.$} (C);
                          \path(C) edge node  {$\# \left| \epsilon, -1 \right.$} (A);

      \end{tikzpicture}
    }
  \end{center}
  \vspace{-1em}
  \caption{Transformation $f_1$ given as two-way
    transducers with look-ahead.   Here symbol $\alpha$ stands for
  both symbols $a$ and $b$, and the predicate $\reach$ is the lookahead
  that checks whether string contains a $\#$ in
  future. $\{t\}$ is the only Muller set. \label{fig:2wst}}  
  \vspace{-1em}
\end{figure}

First we give some examples of transition monoids for the \twst{} in Figure \ref{fig:2wst}. 
Consider the string $ab\#$. The transition monoid is obtained by using all 4 behaviours as shown below. 
Note that on reading $ab \#$, on state $t$, when we reach symbol $\#$, the look-ahead $\neg reach_{\#}$ evaluates to true. 
\[
M^{\ell r}_{ab\#}=\begin{blockarray}{cccc}
 &\matindex{t} & \matindex{q} & \matindex{p}  \\
\begin{block}{c(ccc)}
\matindex{t} & (0) & \bot & \bot\\
\matindex{q} & (0) &\bot  & \bot \\
\matindex{p} & \bot & \bot &  \bot \\
\end{block}
\end{blockarray},~ 
M^{r \ell}_{ab\#}=\begin{blockarray}{cccc}
 &\matindex{t} & \matindex{q} & \matindex{p}  \\
\begin{block}{c(ccc)}
\matindex{t} & \bot & \bot & (0)\\
\matindex{q} & \bot &\bot  & \bot \\
\matindex{p} & \bot & \bot &  \bot \\
\end{block}
\end{blockarray} 
\]

\[
M^{\ell \ell}_{ab\#}=\begin{blockarray}{cccc}
 &\matindex{t} & \matindex{q} & \matindex{p}  \\
\begin{block}{c(ccc)}
\matindex{t} & \bot & \bot & (0)\\
\matindex{q} & \bot &\bot  & \bot \\
\matindex{p} & \bot & \bot &  (0) \\
\end{block}
\end{blockarray},~ 
M^{r r}_{ab\#}=\begin{blockarray}{cccc}
 &\matindex{t} & \matindex{q} & \matindex{p}  \\
\begin{block}{c(ccc)}
\matindex{t} & (0) & \bot & \bot\\
\matindex{q} & (0) &\bot  & \bot \\
\matindex{p} & \bot & (0) &  \bot \\
\end{block}
\end{blockarray} 
\]

If we consider the string $ab\#ab\#$, 
then we can compute for instance, 
$M^{\ell r}_{ab\#ab\#}$ using 
$M^{\ell r}_{ab\#}, M^{\ell \ell}_{ab\#}$ and $M^{rr}_{ab\#}$. 
It can be checked that we obtain $M^{\ell r}_{ab\#ab\#}$
same as $M^{\ell r}_{ab \#}$.  


The transition monoid of the two-way automaton is obtained by using all the four matrices $M^{xy}_s$ for a string $s$. 
In particular, given a string $s$, we consider the matrix 
\[
M_s=\left (
 {\begin{array}{cc}
M^{\ell \ell}(s) & M^{\ell r}(s)\\
M^{r \ell}(s) & M^{ r r}(s)\\
\end{array}}
\right ) \]
as the transition matrix of $s$ in the two-way automaton. The identity element is 
\small{$\left (
 {\begin{array}{cc}
$\textbf{1}$ & $\textbf{1}$\\
$\textbf{1}$ & $\textbf{1}$\\
\end{array}}
\right )$}  
where 
$\textbf{1}$ is the $n \times n$ matrix whose diagonal entries are
$(\emptyset, \emptyset, \ldots, \emptyset)$ and non-diagonal entries are all
$\bot$'s. The matrix corresponding to the empty string $\epsilon$ is  
$\left ( {\begin{array}{cc}
$\textbf{1}$ & $\textbf{1}$\\
$\textbf{1}$ & $\textbf{1}$\\
\end{array}}
\right )$. 

Given a word $w \in \Sigma^*$, we can find a decomposition of $w$ into $w_1$ and $w_2$ such that we can write all behaviours of $w$ in terms of behaviours of $w_1$ and $w_2$, denoted by $M_{w_1 w_2}$. We enumerate the possibilities in each kind of traversal, for a successful decomposition.
\begin{enumerate}
\item For $w=w_1w_2$, and a left-left traversal, we can have 
$M^{\ell \ell}(w)=M^{\ell \ell}(w_1)$ or \\
$M^{\ell \ell}(w)=M^{\ell r}(w_1) \times (M^{\ell \ell}(w_2) \times M^{r r}(w_1))^* \times M^{\ell \ell}(w_2) \times M^{r \ell}(w_1)$. We denote these cases as  
$LL_1, LL_2$ respectively.  
\item For $w=w_1w_2$ and a left-right traversal, we have \\
$M^{\ell r}(w)=M^{\ell r}(w_1) \times (M^{\ell \ell}(w_2) \times M^{r r}(w_1))^* \times M^{\ell r}(w_2)$. We denote this case as 
$LR$.
\item For $w=w_1w_2$ and a right-left traversal, we have \\
$M^{r \ell}(w)=
M^{r \ell}(w_2) \times (M^{r r}(w_1) \times M^{\ell \ell}(w_2))^* \times M^{r \ell}(w_1)$. We denote this case as $RL$.
\item For $w=w_1w_2$ and a right-right traversal, we have\\
$M^{r r}(w)=M^{r r}(w_2)$ or \\
$M^{r r}(w)=M^{r \ell}(w_2) \times (M^{r r}(w_1) \times M^{\ell \ell}(w_2))^* \times M^{r r}(w_1) \times M^{\ell r}(w_2)$. We denote these cases as $RR_1, RR_2$ 
respectively.
\end{enumerate}

With these, for a correct decomposition of $w$ as $w_1w_2$, $M_{w_1w_2}$ is one of the four matrices  as given below for $i, j \in \{1,2\}$.
\[M_{w_1w_2}=\left (
{\begin{array}{cc}
	LL_i & LR\\
	RL & RR_j\end{array}}
\right ) \]
The multiplication of matrices $M^{xy}(w)$ using the $\times$ operator is as defined for DMA using the multiplication 
$\odot$ and addition $\oplus$ of elements in $(\{0,1\} \cup 2^Q)^n \cup \bot$.  
We define a new operation $\otimes$, which takes $M_{w_1}, M_{w_2}$, and for a ``correct'' decomposition 
of $w$ using  $w_1,w_2$, denoted $w_1 \otimes w_2$, we obtain $M_{w_1} \otimes M_{w_2}=M_{w_1 \otimes w_2}$.

It can be seen that with the $\otimes$ operation, the transition matrix 
of a string decomposed as $w_1.w_2$ correctly follows from the left-right behaviours of the 
strings $w_1, w_2$. Let $\Tt(\Aa)$ be the transition monoid of the two-way automaton $\Aa$. Let $|Q|=n$, the number of states of $\Aa$, and let there be $m$ muller sets. 
For ease of notation, we do not include the states of the look-behind automaton and set of states of the look-ahead automaton 
in the transition monoid. Thus, $\Tt(\Aa)$ contains matrices of the above form, with identity 
and the binary operation as defined above. It can be seen that recursively we can find decompositions 
of $w$ as $w_1$ followed by $w_2w_3$, and $w_1w_2$ followed by $w_3$ 
such that $M_{w}=M_{w_1} \otimes  M_{w_2w_3}= M_{w_1w_2} \otimes  M_{w_3}$. 
\begin{lemma}
	\label{app:lem1-conv}
	Let $\mathcal{A}$ be a  two-way (aperiodic) Muller automaton. 
	The mapping $h$ which maps any string $s$ to its
	transition monoid $\Tt(\Aa)$ is a morphism that recognizes $L(\Aa)$. The language $L(\Aa)$ is then aperiodic. 
\end{lemma}
\begin{proof}
The proof of Lemma \ref{app:lem1-conv} is similar to Lemma \ref{lem:ap-rec}. 
As seen in the case of Lemma \ref{lem:ap-rec}, it can then be proved that  $\Tt(\Aa)$ 
is a monoid. To define a morphism from $\Sigma^*$ to $\Tt(\Aa)$, we first define a 
mapping   $h: (\Sigma^*, \cdot, \epsilon) \rightarrow (\Sigma^* \times \Sigma^*, R, (\epsilon, \epsilon))$ as $h(s) = (s_1, s_2)$ iff $s = s_1 \otimes  s_2$
is a ``correct breakup''; that is,  $M_s = M_{s_1} \otimes M_{s_2}$. The operator $R$ is defined as $(a_1, a_2) R (b_1, b_2) = (a_1 a_2 ,b_1 b_2)$, and $h(\epsilon) = (\epsilon, \epsilon).$ It is easy to see that $h$ is a morphism, and  $(\Sigma^* \times \Sigma^*, R, (\epsilon, \epsilon))$ 
is a monoid.

Next, we define another map $g$ which works on the range of $h$. 
 $g: (\Sigma^* \times \Sigma^*, R, (\epsilon, \epsilon)) \rightarrow (T_\Aa, \otimes, \textbf{1})$ as $g((s_1, s_2))=M_{s_1 s_2}$, 
$g[(s_1,s_2) R (s_3,s_4)] = M_{s_1s_2} \otimes M_{s_3s_4}$, which by the correct breakup condition, will 
equal $M_{s_1s_2s_3s_4}$. Also, $g((\epsilon,\epsilon))=\textbf{1}$.  Again, $(T_\Aa, \otimes, \textbf{1})$  
    is a monoid and the composition of $g$ and $h$ is a morphism from $(\Sigma^*, \cdot, \epsilon)$ to 
     $\Tt(\Aa)$.

%
 Further, one can show using a similar argument to Lemma \ref{lem:ap-rec} that
 $L(\Aa)$ is recognized by the morphism $g\circ h$. 
 
 It remains to show that $L(\Aa)$ is aperiodic. By definition, we know that 
 the languages $L^{xy}_{pq}$ are aperiodic. This means that there is some $m$ such that 
  $w^m \in L^{xy}_{pq}$ iff $w^{m+1} \in L^{xy}_{pq}$ for all strings $w$. 
  This implies that  $M^{xy}(w^m)=M^{xy}(w^{m+1})$ for all strings $w$. This would in turn imply (since we have the result for all four quadrants) that 
  we obtain for the matrices of $\Tt(\Aa)$, 
    $M^m_{w}=M_{w^m}=M_{w^{m+1}}=M^{m+1}_w$, where $M^m$ stands for the iterated product using $\otimes$ of matrix $M$ $m$ times. 
  This implies that there is an $m$ such that for all matrices $M$ in $\Tt(\Aa)$,  
  $M^m=M^{m+1}$. Hence, $\Tt(\Aa)$ is aperiodic. Since $h$ is a morphism to an aperiodic monoid recognizing 
  $L(\Aa)$, we obtain that $L(\Aa)$ is aperiodic. 
  \end{proof}

\section{Proofs from Section \ref{subsec:sst}}
\label{app:sst-basics}

  \subsection{Example of SST}
\label{eg:sst}
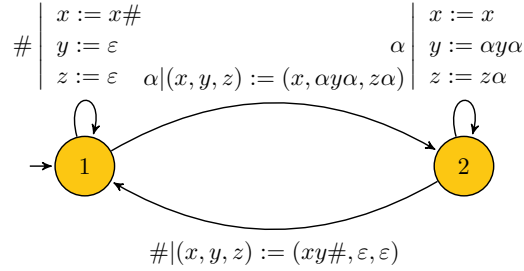
\begin{figure}[h]
  \tikzstyle{trans}=[-latex, rounded corners]
  \begin{center}
    \scalebox{0.9}{
      
      \begin{tikzpicture}[->,>=stealth',shorten >=1pt,auto, semithick,scale=.8]
        \tikzstyle{every state}=[fill=gold]
      \node[initial,state, initial text={}] at (6, 4) (A1) {$1$} ;
      \node[state] at (13,4) (B1) {$2$} ;

  \path (A1) edge [loop above] node {$\#\left|\begin{array}{llllllll}x:=x\#\\y:=\varepsilon\\z:=\varepsilon\end{array}\right.$} (A1);
  \path (A1) edge [bend left] node [above]
        {$\begin{array}{lllll} \vspace{-2mm}\\ \alpha|(x,y,z):=(x,\alpha
            y \alpha, z \alpha)\end{array}$} (B1);

  \path (B1) edge [loop above] node {$\alpha\left|\begin{array}{lllll} x:=x\\
   y:=\alpha y \alpha \\ z := z\alpha \end{array}\right.$} (A1);

  \path (B1) edge [bend left] node [below]
        {$\begin{array}{l}
            \#|(x,y,z):=(xy\#,\varepsilon,\varepsilon) \\ \vspace{-2mm}\end{array}$} (A1);
  
      \end{tikzpicture}
    }
  \end{center}
  \vspace{-1em}
  \caption{Transformation $f_1$ given as streaming string transducers with
    $F(\set{2}) = x z$ is the output associated with Muller set $\set{2}$. $\alpha$ stands for $a,b$. \label{appfig:fo-example}}  
  \vspace{-1em}
\end{figure}

  \begin{example} The transformation $f_1$ introduced  is definable by the
\sst{} in Figure~\ref{appfig:fo-example}. 
Consider the successive valuations of $x, y,$ and $z$ upon reading
the string $ab\#a^\omega$. 
\[
\begin{array}{r|cccccccccccccccccccccccccccccccccccccc}
  & & \!\!\!\!a\!\!\!\! & & \!\!\!\!b\!\!\!\! & & \!\!\!\!\#\!\!\!\! & & \!\!\!\!a\!\!\!\! & & \!\!\!\!a\!\!\!\! &\dots\\
  \hline 
  x & \varepsilon &  & \varepsilon & & \varepsilon & & baab\# & &
  baab\# & & baab\# \\
  y & \varepsilon & & aa & & baab & & \varepsilon & & aa & & aaaa \\
  z & \varepsilon & & a & & ab & & \varepsilon & & a & & aa \\
\end{array}
\]
Notice that the limit of $x z$ exists and equals
$baab\#a^\omega$. 
\end{example}

\subsection{Transition Monoid for SSTs}
 \begin{lemma}
\label{sst-monoid}
$(M_T,\times,\textbf{1})$ is a monoid, where $\times$ is defined as matrix
	multiplication and the identity element $\textbf{1}$ is the matrix with diagonal elements 
	$(\emptyset,\emptyset,\dots,\emptyset)$ and  all non-diagonal elements being 
	$\bot$.
\end{lemma}

\begin{proof}
  Consider any matrix $M_s$ where $s \in \Sigma^*$. 
Assume the \sst{} has  $g$ variables $X_1, \dots, X_g$, and 
$m$ states $\{p_1, \dots, p_m\}$. Let there be an ordering  $(p_i, X_j) <  (p_i, X_l)$ for $j < l$
and $(p_i, X_j) <  (p_k, X_j)$ for $i < k$ used in the transition monoid.
Consider a row corresponding to some $(p,X_i)$. The only entries that are not $\perp$ in this row
are of the form $[(p,X_i)(r,X_1)], \dots, [(p,X_i)(r,X_g)]$, for some state $r$, such that there is a run 
of $s$ from $p$ to $r$, $p, r \in \{p_1, \dots, p_m\}$.  
 The $(x_1, \dots, x_n)$ component of all entries $[(p,X_i)(r,X_1)], \dots, [(p,X_i)(r,X_g)]$ 
are same. Let it be $(\kappa_1, \dots, \kappa_n)$. Let    
$[(p,X_i)(r,X_j)]=k_{ij}(\kappa_1, \dots, \kappa_n)$. 

Consider $M_s.\textbf{1}$. The $[(p,X)(r,-)]$ entries 
of the product are obtained from the $(p,X)$th row of $M_s$ and the $(r,-)$th column of $\textbf{1}$. 
The $(r,-)$th column of $\textbf{1}$ has exactly one $1(\emptyset, \dots, \emptyset)$, while 
all other elements are $\bot$. 
Let $p=p_c$ and $r=p_h$, with $c<h$ ($h<c$ is similar).
Then the $[(p,X)(r,Y)]$ entry for $X=X_i, Y=X_j$ is of the form 
$\bot+ \dots + \bot + k_{ij}(\kappa_1, \dots, \kappa_n).1(\emptyset, \dots, \emptyset)+k_{i+1~j}(\kappa_1, \dots, \kappa_n).\bot$
$+ \dots +k_{gj}(\kappa_1, \dots, \kappa_n).\bot + \dots +\bot.$
Clearly, this is equal to 
$k_{ij}(\kappa_1, \dots, \kappa_n)$. Similarly, it can be shown that the 
$[(p,X_i)(r,X_j)]$th entry of $M_s$ is preserved in   $\textbf{1}.M_s$.  
Associativity can also be checked easily. 
\end{proof}

  The mapping $M_{\bullet}$, which maps any string $s$ to its transition matrix
$M_s$,  is a morphism from $(\Sigma^*, ., \epsilon)$ to $(M_T, \times, \textbf{1})$. 
We say that the transition monoid $M_T$ of an \sst{} $T$ is  $n$-bounded if in
all entries $(j, (x_1, \dots,x_n))$  of the matrices of $M_T$, $j \leq n$.
Clearly, any $n$-bounded transition monoid is finite. 
A streaming string transducer is \emph{aperiodic} if its transition monoid is
aperiodic. 
An $\omega$-streaming string transducer with $n$ Muller acceptance sets is
1-bounded  if its transition monoid is 1-bounded.
That is, for all strings $s$, and all pairs $(p,Y)$, $(q,X)$, 
$M_s[p,Y][q,X] \in \{\bot\} \cup \{ (i, (x_1, \dots,x_n)) \mid i \leq 1\}$.

\subsection{Example of Transition Monoid}
\label{app-example}

  \begin{figure}[h]
    \begin{center}
      \begin{tikzpicture}[->,>=stealth',shorten >=1pt,auto,node distance=1cm,
          semithick,scale=0.9,every node/.style={scale=0.9}]
        \tikzstyle{every state}=[fill=gold,minimum size=1em]
        \node[state,fill=gold,accepting] at (8,-3) (C) {$q$} ;
        \node[state,fill=gold,accepting] at (2,-3) (B) {$r$} ;
        \node[state,initial,  initial where=above,initial text={},fill=gold] at (5,-3) (D) {$t$} ;
        \path(B) edge[loop left] node {$a \mid X:=Xb$} (B);
        \path(C) edge[loop right] node {$b\mid Y:=YX$} (C);
        \path(B) edge[bend left] node {$b$} (D);
        \path(D) edge[bend left] node {$b\mid X:=bY$} (B);
        \path(C) edge[bend left] node {$a\mid X:=bX$} (D);
        \path(D) edge[bend left] node {$a\mid Y:=aX$} (C);
        
        \node[state,initial,initial text={},initial where=above,fill=gold,accepting] at (12,-3) (A1) {$u$} ;
        \node[state,fill=gold,accepting] at (15,-3) (B1) {$v$} ;
        \path(A1) edge[bend left] node {$a$} (B1);
        \path(B1) edge[bend left] node {$b$} (A1);
        \path(B1) edge[loop below] node {$a$} (B1);
        \path(A1) edge[loop below] node {$b$} (A1);
        
      \end{tikzpicture}
    \end{center}
    \caption{Muller accepting set of \sst{} on the left= $\{\{q\},\{r\}\}$. Also, $F(q)=XY$, $F(r)=X$. 
      Muller accepting set for \sst{} on the right=$\{u,v\}$;$X:=X$ and $Y:=Y$ on all edges, $F(\{u,v\})=X$.}
    \label{fig:tm-ex1}
  \end{figure}
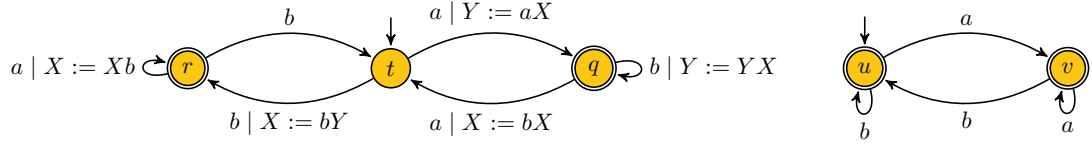
In figure \ref{fig:tm-ex1} consider the strings $ab$ and $bb$ for the automaton on the left. The transition monoids are
 
\[
  M_{ab}=\begin{blockarray}{ccccccc}
    & \matindex{(t,X)} & \matindex{(t,Y)} & \matindex{(q,X)} & \matindex{(q,Y)} & \matindex{(r,X)} & \matindex{(r,Y)} \\
    \begin{block}{c(cccccc)}
      \matindex{(t,X)} & \bot & \bot & 1,(0,0) & 2,(0,0) &\bot  &\bot \\
      \matindex{(t,Y)} & \bot & \bot  & 0, (0,0) & 0, (0,0) & \bot & \bot \\
      \matindex{(q,X)} & \bot & \bot   & \bot  & \bot & 0, (0,0) & 0, (0,0)\\
\matindex{(q,Y)} & \bot & \bot  & \bot  & \bot & 1, (0,0) &1, (0,0) \\
      \matindex{(r,X)} & 1,(0,0) & 0, (0,0)  & \bot & \bot & \bot & \bot\\
      \matindex{(r,Y)} & 0,(0,0) & 1,(0,0)  & \bot & \bot & \bot &\bot \\
                   \end{block}
  \end{blockarray} 
    \]

\[
  M_{bb}=\begin{blockarray}{ccccccc}
    & \matindex{(t,X)} & \matindex{(t,Y)} & \matindex{(q,X)} & \matindex{(q,Y)} & \matindex{(r,X)} & \matindex{(r,Y)} \\
    \begin{block}{c(cccccc)}
      \matindex{(t,X)} & 0,(0,0) & 0,(0,0) & \bot & \bot &\bot  &\bot \\
      \matindex{(t,Y)} & 1,(0,0) & 1,(0,0)  & \bot & \bot & \bot & \bot \\
      \matindex{(q,X)} & \bot & \bot   & 1,(1,0)  & 2, (1,0) &\bot  & \bot\\
\matindex{(q,Y)} & \bot & \bot  & 0,(1,0)  & 1,(1,0) & \bot & \bot  \\
      \matindex{(r,X)} & \bot & \bot  & \bot & \bot & 0,(0,0) & 0,(0,0)\\
      \matindex{(r,Y)} & \bot & \bot  & \bot & \bot & 1,(0,0) &1,(0,0) \\
                   \end{block}
  \end{blockarray} 
    \]

It can be checked that $M_{abbb}=M_{ab}M_{bb}$. Likewise, the transition monoid for $a,b$
for the automaton on the right is 

\[
  M_{a}=\begin{blockarray}{ccccc}
    & \matindex{(u,X)} & \matindex{(u,Y)} & \matindex{(v,X)} & \matindex{(v,Y)}  \\
    \begin{block}{c(cccc)}
      \matindex{(u,X)} & \bot & \bot & 1(1) & 0(1)  \\
      \matindex{(u,Y)} & \bot & \bot  & 0(1) & 1(1)  \\
      \matindex{(v,X)} & \bot & \bot   & 1(v)  & 0(v) \\
\matindex{(v,Y)} & \bot & \bot  & 0(v)  & 1(v)   \\
                   \end{block}
  \end{blockarray} 
    \]
\[
  M_{b}=\begin{blockarray}{ccccc}
    & \matindex{(u,X)} & \matindex{(u,Y)} & \matindex{(v,X)} & \matindex{(v,Y)}  \\
    \begin{block}{c(cccc)}
      \matindex{(u,X)} & 1(u) & 0(u) & \bot & \bot  \\
      \matindex{(u,Y)} & 0(u) & 1(u)  & \bot & \bot \\
      \matindex{(v,X)} & 1(1) & 0(1)   & \bot  & \bot \\
\matindex{(v,Y)} & 0(1) & 1(1)  & \bot  & \bot   \\
                   \end{block}
  \end{blockarray} 
    \]
  \begin{proposition}\label{prop:fodomain}
  The domain of an aperiodic \sst{} is FO-definable. 
\end{proposition}
 \begin{proof}
    Let $T = (\Sigma, \Gamma, Q, q_0, Q_f, \delta, \varsst, \rho, F)$
    be an aperiodic SST and $M_T$ its (aperiodic) transition monoid. Let
    us define a function $\varphi$ which associates with each matrix $M\in
    M_T$, the $|Q|\times |Q|$ Boolean matrix $\varphi(M)$ defined by 
    $\varphi(M)[p][q] =(x_1, \dots, x_n)$ iff there exist $X,Y\in\varsst$ such that 
    $M_T[p,X][q,Y]=(k,(x_1, \dots,x_n))$, with $k \geq 0$.  Clearly, $\varphi(M_T)$ is the transition
    monoid of the underlying input automaton of $T$ (ignoring the variable updates).  The result
    follows, since the homomorphic image of an aperiodic monoid is 
    aperiodic. 
 \end{proof}

\section{Proofs from Section \ref{sec:fot-2wst}}
\label{app:fot-2wst}

\begin{lemma}
\label{cor-lem}
The proof of Lemma \ref{app:lem1-conv} works when $\Aa$ is a two-way, aperiodic  Muller automaton with  star-free look-around. That is, 
$L(\Aa)$ is aperiodic. 
\end{lemma}
\begin{proof}
When we allow star-free look-around, we also have aperiodic Muller look-ahead automaton $A$ and aperiodic look-behind automata $B$
along with $\Aa$. 
Consider a string in $w \in \Sigma^{\omega}$ with a factorization $w=w_1w_2w_3 \dots$. 
Whenever we consider $M_{w_i}$ for some $w_i$ in $\Tt(\Aa)$ (recall this is the monoid for $\Aa$ without the look-around), we also
consider the transition matrix of $w_1 \dots w_i$ with respect to $B$, and 
the matrices $M_{w_{j}}, j > i$ with respect to $A$. A string $w'=w'_1w'_2 \dots$ is then equivalent to $w$ if 
the respective matrices $M_{w_i}, M_{w'_i}$ match in $\Aa$, and so do the others (prefixes upto $w_i, w'_i$ for look-behind $B$ and 
$M_{w_j}, M_{w'_j}$ for look-ahead $A$). We know that the transition monoids of $A, B$ are aperiodic. That is, there exists $m_A, m_B \in \mathbb{N}$ such that
 for all strings $y$, and all pairs of states $p,q$ in ($A$ or $B$), $y^{m_x} \in L_{pq}$ iff  
$y^{m_x+1} \in L_{pq}$ for $x \in \{A,B\}$. 

In the presence of look-around, we keep track of the transition monoids 
in $A, \Aa$ and $B$ at the same time. It can be seen that for two infinite strings $w,w'$ as above, 
the map $h$ from $\Sigma^*$ to the transition monoids with respect to $B, \Aa, A$ 
will be a morphism : this is seen by considering the respective morphisms $h_1, h_2, h_3$ where 
$h_1$ is the morphism from $\Sigma^*$ to the transition monoid of $B$ (a DFA), 
$h_2$ is the morphism from $\Sigma^*$ to the transition monoid of $\Aa$ (the \twst{}), and 
$h_3$ is the morphism from $\Sigma^*$ to the transition monoid of $A$ (a DMA). 
Clearly, equivalent strings $w, w'$ (equivalent with respect to $h$) are either both accepted or rejected 
by the $\twst_{sf}$. The aperiodicity of the combined monoid follows from the aperiodicity of the respective monoids of $B, \Aa$ and  $A$. 
Thus, $L(\Aa)$ is recognized by a morphism to an aperiodic monoid, and hence is an aperiodic language.
%
%
%
%
%
\end{proof}

\subsection{Lemma \ref{fo-behav}}
\label{app:fo-behav}
\begin{proof}
	Lets look at the underlying two-way Muller automaton of the \twst$_{sf}$ $\Aa$. Since $\Aa$ is aperiodic, so is the underlying automaton. 
	Let $s=s[1 \dots x'-1]s[x' \dots y']s[y' \dots]$ be a decomposition of $s$.  
	Let  $x,y$ be any two positions in $s$.  
	Depending on $x,y$, we have 4 cases. 
	If $x=y$ then there is a substring $s_1$ of $s$  
	such that $s_1 \in L^{rr}_{qq'}$ or $s_1 \in L^{ll}_{qq'}$. Similarly, if $x<y$, then there is a substring $s_1$ of $s$ such that 
	$s_1 \in L^{lr}_{qq'}$. Likewise if $x>y$,  
	then there is a substring $s_2$ of $s$ such that 
	$s_2 \in L^{rl}_{qq'}$. If we can characterize $L_{qq'}^{xy}$ in FO for all cases, we are done, since 
	$\psi_{q,q'}(x,y)$ will then be the disjunction of all these formulae.  
	
	Using lemma \ref{cor-lem},  
	the underlying input language  $L(\Aa)\subseteq \Sigma^{\omega}$  can be shown to be aperiodic, by constructing the morphism 
	$h: \Sigma^* \rightarrow M$ recognizing $L(\Aa)$, where $M$ is the aperiodic monoid described in lemma\ref{cor-lem}. 
	It is known 	 \cite{dg08SIWT} that every aperiodic language  $\subseteq \Sigma^{\omega}$ is FO-definable. 
			 	
			 We now show that $h$ recognizes $L_{qq'}^{xy}$. 			 	Consider $w \in L_{qq'}^{xy}$. For the morphism $h: \Sigma^* \rightarrow M$ as above, 
	let $w' \in \Sigma^+$ be such that  $h(w)=h(w')$. Then it is easy  to see that 
	$w' \in L_{qq'}^{xy}$.  	Then $h$ is a morphism to a finite aperiodic monoid 
	that recognizes $L_{qq'}^{xy}$. Hence, $L_{qq'}^{xy}$ is aperiodic and hence FO-definable. For each case, 
	$x <y, x > y, x=y$, 	let $\exists x \exists y \psi_{q,q'}(x,y)$ be the FO formula that captures $L_{qq'}^{xy}$. 
	
		For a particular assignment of positions $x, y$, it can be seen that 
	all words $u$ which have a run starting at position $x$ in state $q$, 
	to state $q'$ in position $y$ will satisfy $\psi_{q,q'}(x,y)$.  
	\end{proof}
\subsection{FOT $\subseteq$ Aperiodic \twst$_{sf}$}
\label{app:fo-wst}
\begin{definition}
A \twst{} with FO instructions ($\twst{}_{fo}$) is a tuple 
$\Aa=(\Sigma, \Gamma, Q, q_0, \delta, F)$ such that $\Sigma, \Gamma, Q, q_0$ and $F$ are as defined in section \ref{sec:2wst}
and $\delta: Q \times \Sigma \times \phi_1 \rightarrow Q \times \Gamma \times \phi_2$ is the transition function where
$\phi_1$ is a set of FO formulae  over $\Sigma$ with one free variable defining the guard of the transition, 
and $\phi_2$ is a set of FO formulae over $\Sigma$ with two free variables defining the jump of the input head. 

Given a state $q$ and position $x$ of the input string $s$, the transition $\delta(q,\phi_1)=(q',b,\phi_2)$ is enabled if 
$s \models \phi_1(x)$, and as a result $b$ is written on the output, and the input head moves to position 
$y$ such that $s \models \phi_2(x,y)$.  
 Note that the jump is deterministic, at each state $q$ and each position $x$,  the formula $\phi_2(x,y)$ is such that 
there is a unique position $y$ to which the reading head will jump.   

A $\twst{}_{fo}$ is  aperiodic if the underlying input language accepted is aperiodic. 
\end{definition}

As in the case of \twst{}, we assume that the entire input is read by 
the $\twst{}_{fo}$, failing which the output is not defined.

\begin{lemma}(\fot{} $\subseteq$ \twst$_{fo}$)
\label{foi-2wst}
Any $\omega$-transformation captured by an \fot{} is also captured by a \twst$_{fo}$. 
\end{lemma}
\begin{proof}
Let $T=(\Sigma, \Gamma, \phi_{dom}, C, \phi_{pos}, \phi_{\prec})$ be a 
\fot{}. We define the \twst$_{fo}$ $\Aa=(\Sigma, \Gamma, Q, q_0, \delta, F)$ such that 
$\inter{T}=\inter{\Aa}$. The states $Q$ of $\Aa$ correspond to the copies in $T$. 
So, $Q=C$. Given a state $q$ and a position $x$ of the input string,  the transition that 
checks a guard at $x$, and 
decides 
the jump to position $y$, after writing a symbol $b$ on the output is obtained 
from the formulae $is\_string$, $\phi_{b}^q(x)$ and $\phi^{q,q'}(x,y)$. 
 $is\_string$ describes the input string, $\phi_{b}^q(x)$  is an FO formula that captures the position $x$ in copy $q$ (and also asserts that 
 the output is  $b \in \Gamma$), $\phi^{q,q'}(x,y)$ is an FO formula that enables the $\prec$ relation between positions 
 $x,y$ of copies $q,q'$ respectively. In $\Aa$, this amounts to evaluating the guard $\phi_b^q(x)$ at position $x$ in state $q$, outputting $b$, and 
 jumping to position $y$ of the input in state $q'$ if  the input string satisfies $\phi^{q,q'}(x,y)$. Thus, we write the transition as 
 $\delta(q, a, \phi_b^q(x))=(q',b,\phi^{q,q'}(x,y))$. The initial state $q_0$ of $\Aa$ is the copy $c$ which has its first position $y$ 
 such that $x \nprec y$ for all positions $y \neq x$. Since the \fot{} is string-to-string, we will have such a unique copy. 
 Thus, $q_0=d$ where $d$ is a copy satisfying the formula $\exists y[first^d(y)]$. The set $F$ of Muller states of the constructed \twst$_{fo}$ is all possible subsets 
 of $Q$, since we have captured the transitions between copies (now states) correctly.

 Now we have to show that $\Aa$ is aperiodic.   This amounts to showing that there exists some integer $n$ such that for all pairs of states $p,q,$
  $v^n \in L_{pq}^{xy}$ iff 
$v^{n+1} \in L_{pq}^{xy}$ for all strings $v \in \Sigma^*$, and 
$x,y \in \{l,r\}$. Recall that the states of the automaton correspond to the copies of the \fot{}, and 
 $v^n \in L_{pq}^{xy}$ means that $v^n \models \phi^{p,q}(x,y)$. Thus we have to show that 
  $v^n \models \phi^{p,q}(x,y)$ iff $v^{n+1} \models \phi^{p,q}(x,y)$ for all 
 strings $v \in \Sigma^*$.

 Note that since the domain of an \fot{} is aperiodic, 
 for any strings $u,v,w$, there exists an integer $n$ such that 
 $uv^nw$ is in the domain of $T$ iff $uv^{n+1}w$ is.  
In particular, for any pair of positions $x,y$ in $v,w$ respectively, 
and copies $p, q$, the formula $\varphi^{p,q}(x,y)$ 
which asserts the existence of a path from position $x$ of copy $p$ 
to position $y$ of copy $q$ is such that 
 $uv^nw \models \varphi^{p,q}(x,y)$ iff 
 $uv^{n+1}w \models \varphi^{p,q}(x,y)$ (note that this is true since 
 these formulae evaluate in the same way for all strings in the domain, and the domain is aperiodic).
   This means that 
 $uv^nw \in  L_{pq}^{xy}$ iff 
$uv^{n+1}w \in  L_{pq}^{xy}$ for all $u,v,w$. In particular, for $u=w=\epsilon$, we 
obtain  $v^n  \in L_{pq}^{xy}$ iff $v^{n+1} \in  L_{pq}^{xy}$ for all 
 strings $v \in \Sigma^*$, showing that $\Aa$ is aperiodic. 
 \end{proof}

Before we show \twst$_{fo} \subseteq \twst_{sf}$, we need the next two lemmas.
\begin{lemma}
\label{lem-int-1}
Let $\Delta, \Delta'$ be disjoint subsets of $\Sigma$, and let $L \subseteq \Sigma^{\omega}$
be an aperiodic language such that each string in $L$ contains exactly one occurrence 
of a symbol from $\Delta$ and one occurrence of a symbol from $\Delta'$. Then $L$ can be written as the finite union of disjoint languages
    $R_{\ell}.a.R_{m}.b.R_{r}$ where $(a,b) \in (\Delta \times \Delta') \cup (\Delta' \times \Delta)$, $R_{\ell}, R_m \subseteq (\Sigma -\Delta -\Delta')^*$, 
    and $R_r \subseteq (\Sigma -\Delta -\Delta')^{\omega}$. Moreover, $R_{\ell}, R_m$ and $R_r$ are aperiodic. 
\end{lemma}
\begin{proof}
Let $\Aa$ be a deterministic, aperiodic  Muller automaton accepting $L$. From our assumption, it follows that 
each path in $\Aa$ from the initial state $q_0$ passes through exactly one transition labeled 
with a symbol from $\Delta$ and one transition labeled with a symbol from $\Delta'$. 
Let $(q,a,q') \in Q \times \Delta \times Q, (p, b, p')  \in Q \times \Delta' \times Q$ be two such transitions. 
Let $R_{\ell}$ consist  of all finite strings from $q_0$ to $q$, and let $R_m$ consist of all finite strings 
from $q'$ to $p$, and let $R_r$ be the set of all strings from $p'$ which continously witnesses some Muller set from some point onwards. 
Since the underlying automaton is aperiodic, it is easy to see that 
the restricted automata for $R_{\ell}, R_m, R_r$ are also aperiodic. 
Hence, $R_{\ell}, R_m$ are aperiodic languages $\subseteq \Sigma^*$ while 
$R_r$ is an aperiodic $\omega$-language. This breakup using $R_{\ell}, R_m$ and $R_r$ accounts   
for strings where the symbol from $\Delta$ occurs first, and the symbol from $\Delta'$ occurs later.

Symmetrically, we can define $R'_{\ell}, R'_m$ and $R'_r$ to be the breakup for 
strings of $L$ where a symbol of $\Delta'$ is seen first, followed by a symbol of $\Delta$. Clearly, 
 $L$ is the finite union of languages 
$R_{\ell}.a.R_{m}.b.R_{r}$ and $R'_{\ell}.a.R'_{m}.b.R'_{r}$, where 
$R_{\ell}, R'_{\ell}, R_{m}, R'_{m}, R_{r}$ and $R'_{r}$ are all aperiodic. 
\end{proof}

\begin{lemma}
\label{lem-int-2}
Let $\Delta \subseteq \Sigma$, and let $L \subseteq \Sigma^{\omega}$
be an aperiodic language such that each string in $L$ contains exactly one occurrence 
of a symbol from $\Delta$.  Then $L$ can be written as the finite union of disjoint languages
    $R_{\ell}.a.R_{r}$ where $a \in \Delta$, $R_{\ell} \subseteq (\Sigma -\Delta)^*$, 
    and $R_r \subseteq (\Sigma -\Delta)^{\omega}$. Moreover, $R_{\ell}$ and $R_r$ are aperiodic. 
\end{lemma}
The proof of Lemma \ref{lem-int-2} is similar to that of Lemma \ref{lem-int-1}.

\begin{lemma}($\twst_{fo} \subseteq \twst_{sf}$)
\label{2wstfo-2wst}
An $\omega$-transformation captured by an aperiodic  $\twst_{fo}$ is also captured by an aperiodic $\twst_{sf}$. 
\end{lemma}
\begin{proof}
To prove this, we show that Aperiodic $\twst{}_{fo} \subseteq$ Aperiodic  $\twst{}_\la$. Let $\Tt=(\Sigma, \Gamma, Q, q_0, \delta, F)$  be an aperiodic  $\twst{}_{fo}$. We define an aperiodic  $\twst{}_{sf}$ (with star free look around)  $(T, A, B)$ where 
$T=(\Sigma, \Gamma, Q, q_0, \delta, F)$ capturing the same transformation. A transition of the 
$\twst{}_{fo}$ is of the form $\delta(q,a, \varphi_1)=(q', z, \varphi_2)$. $\varphi_1(x)$ is an FO formula
that acts as the guard of the transition, while $\varphi_2(x,y)$ is an FO formula that deterministically decides the jump 
to a position $y$ from the current position $x$. The languages of these formulae, $L(\varphi_1)$ and $L(\varphi_2)$ are aperiodic, and 
one can construct aperiodic Muller automata accepting $L(\varphi_1)$ and $L(\varphi_2)$. From Lemma \ref{lem-int-2}, we can write 
$L(\varphi_1)$  as a finite union of disjoint languages $R'_{\ell}.(a,1).R'_{r}$ with $R'_{\ell} \subseteq [\Sigma \times \{0\}]^*$ and 
$R'_{r} \subseteq [\Sigma \times \{0\}]^{\omega}$. It is easy to see that one can construct aperiodic automata 
accepting $R_{\ell}$ and $R_{m}$, the projections of $R'_{\ell}$ and $R'_{m}$ on $\Sigma$.  

We have to now show how to simulate the jumps of $\varphi_2$ by moving one cell at a time. First we augment 
$\varphi_2(x,y)$ in such a way that we know whether $y< x$ or $y >x$ or $y=x$. This is done by rewriting 
$\varphi_2(x,y)$ as $\exists y (y \sim x\wedge \varphi_2(x,y))$ where $\sim \in \{<, =, >\}$.  
Clearly, the new formula is also FO, and the language of the formula is aperiodic. If  $y < x$, by Lemma \ref{lem-int-1},  
 the language of $\varphi_2(x,y)$ can be written as 
 $R'_{\ell}(a,1,0) R'_m (b,0,1) R'_r$, where $R'_{\ell}, R'_m \subseteq [\Sigma \times \{0\} \times \{0\}]^*$ and $R'_r \subseteq [\Sigma \times \{0\} \times \{0\}]^{\omega}$. Let $R_{\ell}, R_m$ and $R_r$ be the projections of $R'_{\ell}, R'_m$ and $R'_r$ to $\Sigma$. We can construct an aperiodic look-behind automata 
 for $R_{\ell} a R_m b$ and an aperiodic look ahead Muller automaton for $R_r$. This is possible since 
 $R_{\ell},R_m $ are aperiodic.  To walk cell by cell, instead of jumping, 
 the $\twst_{sf}$ does the following. (1) construct the aperiodic automaton that accepts the reverse    
of $R_m$ (this is possible since the reverse of an aperiodic language is aperiodic), (2) simulate this reverse automaton on each transition, and remember 
the state reached in this automaton in the finite control, while  moving left cell by cell each time, (3) when an accepting state is reached in the reverse automaton, check the look-behind $R_{\ell}$, and look-ahead $R_m b R_r$. Note that there is an aperiodic look-ahead Muller automaton 
for $R_m b R_r$. If indeed at the position where we are at an accepting state of the reverse of $R_m$ automaton, 
the look-ahead and look-behind are satisfied, then we can stop moving left. 

A similar construction works when $y>x$. Clearly, each transition of the aperiodic $\twst_{fo}$ can be simulated by 
a $\twst_{sf}$ with star-free look around. The aperiodicity of $\twst_{sf}$ follows from the fact 
the input language accepted by the $\twst_{sf}$ is same as that of the $\twst_{fo}$, and 
is aperiodic.

\end{proof}

\section{Proofs from Section \ref{sec:sst-fot}}
\label{app:varflow}
\begin{proposition}(FO-definability of variable flow)
\label{prop:foflow}
    Let $T$ be an aperiodic,1-bounded \sst{} $T$ with set of variables $\varsst$. 
    For all variables $X,Y\in \varsst$, there exists an FO-formula
    $\phi_{X\flows Y}(x,y)$ with two free variables such that,
    for all strings $s\in dom(T)$ and any two positions $i\leq j\in dom(s)$, 
    $s\models \phi_{X\flows Y}(i,j)$ iff $(q_i,X)\flows^{s[i{+}1{:}j]}_1
    (q_j,Y)$, where $q_0\dots q_n \dots$ is the accepting run of $T$ on $s$.
\end{proposition}
Let $X\in\varsst$, $s\in\dom(T)$, $i\in \dom(s)$. Let $P \in 2^Q$ be a Muller accepting set with output $F(P)=X_1 \dots X_n$. Let 
$r = q_0\dots q_{n}\dots$ be an accepting run of $T$ on $s$. Then 
there is a position $j$ such that $\forall k >j$, only states in $P$ are visited. 
On all transitions beyond $j$, $X_1, \dots, X_{n-1}$ remain unchanged, 
while $X_n$ is updated as $X_nu$ for some $u \in (\mathcal{X} \cup \Gamma)^*$.
We say that the pair $(X,i)$ is \emph{useful} if the content of
variable $X$ before reading $s[i]$ will be part of the output after
reading the whole string $s$. 
Formally, for an accepting run $r = q_0\dots q_{j}q_{j+1}\dots$, such that 
only states from the Muller set $P$ are seen after position $j$, we say that
  $(X,i)$ is useful for $s$ if $(q_{i-1}, X)\flows^{s[i{:}j] }_1(q,X_k)$ for some variable $X_k \in F(P)$, $k \leq n$, or 
  if $(q_{i-1}, X)\flows^{s[i{:}\ell] }_1(q,X_n)$, with $q \in P$ and $\ell >j$.  
  Thanks to Proposition \ref{prop:foflow}, this property is FO-definable.
  
  \begin{proposition}\label{prop:contribution}
    For all $X\in\varsst$, there exists an FO-formula
    $\contribute_X(i)$ s.t. for all strings $s\in dom(T)$ and all
    positions $i\in \dom(s)$, 
    $s\models \contribute_X(i)$ iff $(X,i)$ is useful for string $s$.
\end{proposition}

 Proofs of propositions \ref{prop:foflow} and \ref{prop:contribution} are below.

\subsection{Proof of Proposition \ref{prop:foflow}}

First, we show that reachable states in  accepting runs of aperiodic \sst{} are
FO-definable:

\begin{proposition}
	\label{prop:fostates}
	Let $T$ be an aperiodic \sst{} $T$. For all states $q$, there exists an
	FO-formula $\phi_q(x)$ such that for all strings $s\in \Sigma^{\omega}$, for
	all positions $i$, $s\models \phi_q(i)$ iff $s\in dom(T)$ and the state of the (unique)
	accepting run of $T$ before reading the $i$-th symbol of $s$ is $q$. 
\end{proposition}

\begin{proof}
	Let $A$ be the underlying (deterministic) aperiodic, Muller automaton of the SST $T$. Since $T$
	is aperiodic, so is $A$. For all states 
	$q$, let $L_q$ be the set of strings $s \in \Sigma^+$ such that there exists a run of $T$ on
	string $s$ that ends in state $q$. Clearly, $L_q$ can be defined by some aperiodic
	finite state automaton $A_q$ obtained by setting the set of final states of $A$ to
	$\{q\}$. Here, $A_q$ is interpreted as a DFA. 
	Therefore $L_q$ is definable by some FO-formula $\psi^L_q$. 
	Let $L_{q,p}$ be the set of strings that have a run in $T$ from state $q$ to state $p$.
	This also is obtained from $T$ by considering $q$ as the start state and $p$ as 
	the accepting state. Let $\psi^L_{q,p}$ be the FO formula which captures this.

	For $P \in 2^Q$, let $R_P$ be the set of strings $s \in \Sigma^{\omega}$ such that there exists a run of
	$T$ on $s$ from some state $p \in P$, which stays in the set of states $P$, and all states of $P$ are witnessed
	infinitely often.  
	Clearly, $u\in dom(T)$
	iff there exists $p\in P$ with $v_1\in L_q, v_2 \in L_{q,p}$ and $v_3\in R_P$  such that $u=v_1v_2v_3$.
	To capture $v_3$, we have the FO-formula $\psi^{Rec}_P$ 
	
	$\bigvee_{p \in P} L_p(x) \wedge 
	\forall y(y \geq x \rightarrow \bigvee_{q \in P} L_q(y)) 
	\wedge \forall y \exists k_1 \dots \exists k_{|P|}(k_i \geq y \wedge \bigwedge_{i \neq j}(k_i \neq k_j) \wedge  
	\bigwedge_{q_i \in P} L_{q_i}(k_i))$.    
	
	Then, $\phi_q(x)$ is defined as
	$$
	\phi_q(x) = [\psi^L_q]_{\prec x} \wedge [\psi^L_{q,p}]_{x \prec y} \wedge [L_p(y) \wedge \psi^{Rec}_P]_{y\preceq}
	$$
	where $[\psi^L_q]_{\prec x}$ is the formula $\psi^L_q$ in which 
	all quantifications of any variable $z$ is guarded by $z\prec x$,
	$[\psi^L_{q,p}]_{x \prec y}$ is the formula where all variables 
	lie in between $x,y$, and finally, 
	$[\psi^{Rec}_P]_{y\preceq}$ is the formula $\psi^{Rec}_P$ in which
	all quantifications of any variable $z$ is guarded by $y\preceq z$.
	Therefore, $s\models
	\phi_q(i)$ iff $s[1{:}i)\in L_q$, 
	$s[i{:}j)\in L_{q,p}$
	and $s[j{:}\infty]\in R_P$.  %
	
\end{proof}
Now we start the proof of Proposition \ref{prop:foflow}.    
\begin{proof}
	For all states $p,q\in Q$, 
	let $L_{(p,X)\flows (q,Y)}$ be the language of strings $u$ such that $(p,X)\flows_1^{u}
	(q,Y)$. We show that $L_{(p,X)\flows (q,Y)}$ is an aperiodic
	language. It is indeed
	definable by an aperiodic non-deterministic automaton $A$ that keeps track of flow information
	when reading $u$. It is constructed from $T$ as follows.
	Its state set $Q'$ are pairs $(r,Z)\in 2^{Q\times \mathcal{X}}$. Its initial state is
	$\{(p,X)\}$ and final states are all states $P$ such that
	$(q,Y)\in P$. There exists a transition $P\xrightarrow{a} P'$ in
	$A$ iff for all $(p_2,X_2)\in P'$, there exists $(p_1,X_1)\in P$
	and a transition $p_1\xrightarrow{a|\rho} p_2$ in $T$ such that
	$\rho(X_2)$ contains an occurrence of $X_1$. Note that by
	definition of $A$, there exists a run from a state $P$ to a state
	$P'$ on some $s\in\Sigma^*$ iff for all $(p_2,X_2)\in P'$, there
	exists $(p_1,X_1)\in P$ such that $(p_1,X_1)\flows^{s}_1
	(p_2,X_2)$  (Remark $\star$). 
	
	Clearly, $L(A) = L_{(p,X)\flows (q,Y)}$. It remains to show that $A$ is
	aperiodic, i.e. its transition monoid $M_A$ is aperiodic.
	Since $T$ is aperiodic, there exists $m\geq 0$ such
	that for all matrices $M\in M_T$, $M^m = M^{m+1}$. 
	For $s\in\Sigma^*$, let $\Phi_A(s) \in M_A$ (resp. $\Phi_T(s)$) 
	the square matrix of dimension
	$|Q'|$ (resp. $|Q|$) associated with $s$ in $M_A$ (resp. in $M_T$). We show that $\Phi_A(s^m) =
	\Phi_A(s^{m+1})$, i.e. $(P,P')\in \Phi_A(s^m)$ iff
	$(P,P')\in\Phi_A(s^{m+1})$, for all $P,P'\in Q'$.

	First, suppose that $(P,P')\in \Phi_A(s^m)$, and let
	$(p_2,X_2)\in P'$.  By definition of $A$, there exists 
	$(p_1,X_1)\in P$ such that $(p_1,X_1)\flows^{s^m}_1 (p_2,X_2)$, and 
	by aperiodicity of $T$, it implies that
	$(p_1,X_1)\flows^{s^{m+1}}_1 (p_2,X_2)$. Since it is true for all
	$(p_2,X_2)\in P'$, it implies by Remark $(\star)$ that there exists a run 
	of $A$ from $P$ to $P'$ on $s^{m+1}$, i.e. $(P,P')\in
	\Phi_A(s^{m+1})$. The converse is proved similarly.

	We have just proved that $L_{(p,X)\flows (q,Y)}$ is aperiodic. Therefore it is
	definable by some FO-formula $\phi_{(p,X)\flows (q,Y)}$. To capture 
	the variable flow in an accepting run, we also need to have some conditions on the states $p$ and $q$.  
	$\phi_{X\flows Y}(x,y)$ is defined by
	$$
	\phi_{X\flows Y}(x,y)\equiv x\preceq y \wedge
	\bigvee_{q,p\in Q}\{ [\phi_{(q,X)\flows (p,Y)}]^{x\preceq
		\cdot\preceq y} \wedge 
	[\psi^L_q]_{\prec x} \wedge [\psi^L_{q,p}]_{x \prec y} \wedge \exists z.\{[\psi^L_{p,r}]_{y \prec z} \wedge 
	[L_r(z) \wedge \psi^{Rec}_P]_{z\preceq}\} \}$$
	where 
	$[\psi^L_q]_{\prec x}, [\psi^L_{q,p}]_{x \prec y}, [\psi^R_p]_{y\preceq}$
	were defined in Proposition
	\ref{prop:fostates} and $[\phi_{(p,X)\flows (q,Y)}]^{x\preceq \cdot\preceq y}$ is
	obtained from $\phi_{(p,X)\flows (q,Y)}$ by guarding all the
	quantifications of any variable  by $x\preceq z'\preceq y$. 
The Muller set $P$ starts from some position $z$ ahead of $y$, in some state $r \in P$. 
	In the case when $p \in P$, consider $r=p$ in the above formula.  
\end{proof}

\subsection{Proof of Proposition \ref{prop:contribution}}
\label{app:useful}
\begin{proof}
	The formula $\contribute_X(x)$ is defined by
	$$
	\begin{array}{llllllll}
	\contribute_X(x) & = & \exists y\cdot[ 
	\bigvee_{P \in 2^Q}\psi_P^{Rec}(y) \wedge \bigvee_{p \in P}\psi^L_p(y) \wedge \bigvee_{\{X_1, \dots, X_n \mid F(P)=X_1 \dots X_n\}}
	\Phi_{X\flows X_i}(x,y)]
	\end{array}
	$$
	where $\psi_P^{Rec}(y)$ defines 
	a position $y$ from where the Muller set $P$ is visited continously. 
	$\psi_P^{Rec}(y)$ is defined in proposition \ref{prop:fostates} and 
	$\Phi_{X\flows Y}(x,y)$ in proposition \ref{prop:foflow}.
\end{proof}

\subsection{Definition of \sst-output graphs}
\label{app:sst-output}
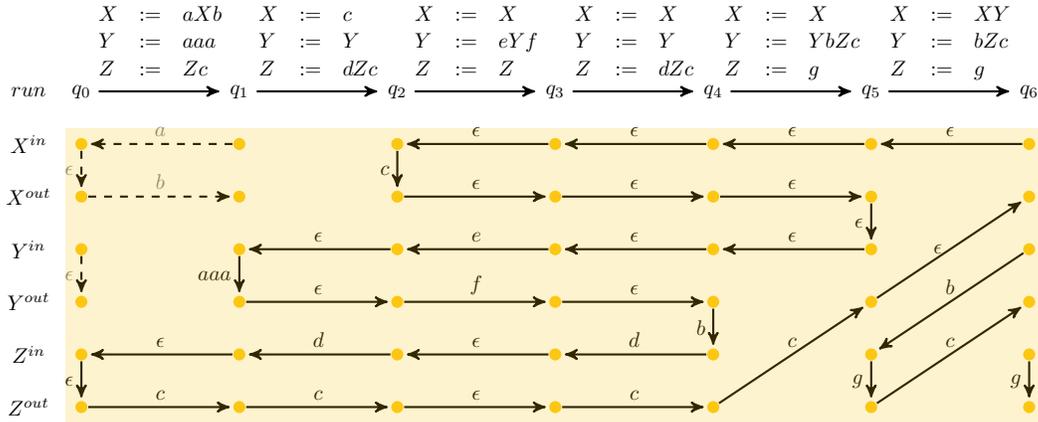
\begin{figure*}[h]
\begin{center}
\begin{tikzpicture}[->,>=stealth',shorten >=1pt,auto,node distance=1.5cm,
                    thick,scale=0.7, every node/.style={scale=0.8}]

\tikzstyle{graphnode}=[circle,fill=gold,thick,inner
sep=0pt,minimum size=2mm]

\tikzstyle{graphnodeblack}=[circle,fill=gold,thick,inner
sep=0pt,minimum size=2mm]

\node [graphnode] (x0in) at (0,5) {} ;
\node [graphnode] (x1in) at (3,5) {} ;
\node [graphnodeblack] (x2in) at (6,5) {} ;
\node [graphnodeblack] (x3in) at (9,5) {} ;
\node [graphnodeblack] (x4in) at (12,5) {} ;
\node [graphnodeblack] (x5in) at (15,5) {} ;
\node [graphnodeblack] (x6in) at (18,5) {} ;

\node [graphnode] (x0out) at (0,4) {} ;
\node [graphnode] (x1out) at (3,4) {} ;
\node [graphnodeblack] (x2out) at (6,4) {} ;
\node [graphnodeblack] (x3out) at (9,4) {} ;
\node [graphnodeblack] (x4out) at (12,4) {} ;
\node [graphnodeblack] (x5out) at (15,4) {} ;
\node [graphnodeblack] (x6out) at (18,4) {} ;

\node [graphnode] (y0in) at (0,3) {} ;
\node [graphnodeblack] (y1in) at (3,3) {} ;
\node [graphnodeblack] (y2in) at (6,3) {} ;
\node [graphnodeblack] (y3in) at (9,3) {} ;
\node [graphnodeblack] (y4in) at (12,3) {} ;
\node [graphnodeblack] (y5in) at (15,3) {} ;
\node [graphnodeblack] (y6in) at (18,3) {} ;

\node [graphnode] (y0out) at (0,2) {} ;
\node [graphnodeblack] (y1out) at (3,2) {} ;
\node [graphnodeblack] (y2out) at (6,2) {} ;
\node [graphnodeblack] (y3out) at (9,2) {} ;
\node [graphnodeblack] (y4out) at (12,2) {} ;
\node [graphnodeblack] (y5out) at (15,2) {} ;
\node [graphnodeblack] (y6out) at (18,2) {} ;

\node [graphnodeblack] (z0in) at (0,1) {} ;
\node [graphnodeblack] (z1in) at (3,1) {} ;
\node [graphnodeblack] (z2in) at (6,1) {} ;
\node [graphnodeblack] (z3in) at (9,1) {} ;
\node [graphnodeblack] (z4in) at (12,1) {} ;
\node [graphnodeblack] (z5in) at (15,1) {} ;
\node [graphnodeblack] (z6in) at (18,1) {} ;

\node [graphnodeblack] (z0out) at (0,0) {} ;
\node [graphnodeblack] (z1out) at (3,0) {} ;
\node [graphnodeblack] (z2out) at (6,0) {} ;
\node [graphnodeblack] (z3out) at (9,0) {} ;
\node [graphnodeblack] (z4out) at (12,0) {} ;
\node [graphnodeblack] (z5out) at (15,0) {} ;
\node [graphnodeblack] (z6out) at (18,0) {} ;

\node (xin) at (-1,5) {$X^{in}$};
\node (xout) at (-1,4) {$X^{out}$};
\node (yin) at (-1,3) {$Y^{in}$};
\node (yout) at (-1,2) {$Y^{out}$};
\node (zin) at (-1,1) {$Z^{in}$};
\node (zout) at (-1,0) {$Z^{out}$};

  \draw[->,dashed]   (x0in) -- node[left] {\textcolor{gray}{$\epsilon$}} (x0out) ;
  \draw[->,dashed]   (y0in) -- node[left] {\textcolor{gray}{$\epsilon$}} (y0out);
  \draw[->]   (z0in) -- node[left] {$\epsilon$} (z0out);


  \draw[->,dashed] (x1in) -- node[above] {\textcolor{gray}{$a$}} (x0in) ;
  \draw[->,dashed] (x0out) -- node[above] {\textcolor{gray}{$b$}} (x1out) ;

  \draw[->] (z1in) -- node[above] {$\epsilon$} (z0in) ;
  \draw[->] (z0out) -- node[above] {$c$} (z1out) ;


  \draw[->] (y1in) -- node[left] {$aaa$} (y1out) ;


  \draw[->] (y2in) -- node[above] {$\epsilon$} (y1in) ;
  \draw[->] (y1out) -- node[above] {$\epsilon$} (y2out) ;

  \draw[->] (z2in) -- node[above] {$d$} (z1in) ;
  \draw[->] (z1out) -- node[above] {$c$} (z2out) ;


  \draw[->] (x2in) -- node[left] {$c$} (x2out) ;
           

  \draw[->] (y3in) -- node[above] {$e$} (y2in) ;
  \draw[->] (y2out) -- node[above] {$f$} (y3out) ;

  \draw[->] (x3in) -- node[above] {$\epsilon$} (x2in) ;
  \draw[->] (x2out) -- node[above] {$\epsilon$} (x3out) ;

  \draw[->] (z3in) -- node[above] {$\epsilon$} (z2in) ;
  \draw[->] (z2out) -- node[above] {$\epsilon$} (z3out) ;


  

  \draw[->] (x4in) -- node[above] {$\epsilon$} (x3in) ;
  \draw[->] (x3out) -- node[above] {$\epsilon$} (x4out) ;

  \draw[->] (y4in) -- node[above] {$\epsilon$} (y3in) ;
  \draw[->] (y3out) -- node[above] {$\epsilon$} (y4out) ;
 
  \draw[->] (z4in) -- node[above] {$d$} (z3in) ;
  \draw[->] (z3out) -- node[above] {$c$} (z4out) ;
  


\draw[->] (x5in) -- node[above] {$\epsilon$} (x4in) ;
  \draw[->] (x4out) -- node[above] {$\epsilon$} (x5out) ;

\draw[->] (y5in) -- node[above] {$\epsilon$} (y4in) ;
  \draw[->] (y4out) -- node[left] {$b$} (z4in) ;
 \draw[->] (z4out) -- node[above] {$c$} (y5out) ;
 \draw[->] (z5in) -- node[left] {$g$} (z5out) ;

\draw[->] (x6in) -- node[above] {$\epsilon$} (x5in) ;
  \draw[->] (x5out) -- node[left] {$\epsilon$} (y5in) ;

\draw[->] (y6in) -- node[above] {$b$} (z5in) ;
 \draw[->] (y5out) -- node[left] {$\epsilon$} (x6out) ;
 \draw[->] (z5out) -- node[above] {$c$} (y6out) ;
 \draw[->] (z6in) -- node[left] {$g$} (z6out) ;

\node (run) at (-1,6) {$run$};

\node (q) at (0,6) {$q_0$} ;
\node (q1) at (3,6) {$q_1$} ;
\node (r) at (6,6) {$q_2$} ;
\node (q2) at (9,6) {$q_3$} ;
\node (r1) at (12,6) {$q_4$} ;
\node (r2) at (15,6) {$q_5$} ;
\node (r3) at (18,6) {$q_6$} ;


 \draw[->] (q) -- node[above] {$\begin{array}{lll} X & := & aXb
       \\ Y & := & aaa \\ Z & := & Zc \end{array}$}(q1) ;

 \draw[->] (q1) -- node[above] {$\begin{array}{lll} X & := & c
       \\ Y & := & Y \\ Z & := & dZc \end{array}$}(r) ;

 \draw[->] (r) -- node[above] {$\begin{array}{lll} X & := & X
       \\ Y & := & eYf \\ Z & := & Z \end{array}$}(q2) ;

 \draw[->] (q2) -- node[above] {$\begin{array}{lll} X & := & X
       \\ Y & := & Y \\ Z & := & dZc \end{array}$}(r1) ;

 \draw[->] (r1) -- node[above] {$\begin{array}{lll} X & := & X
       \\ Y & := & YbZc \\ Z & := & g \end{array}$}(r2) ;

 \draw[->] (r2) -- node[above] {$\begin{array}{lll} X & := & XY
       \\ Y & := & bZc \\ Z & := & g \end{array}$}(r3) ;



\fill[gold,fill opacity=0.2] (-0.3,-0.3) rectangle (18.3,5.3);







\end{tikzpicture}
\end{center}
\vspace{-5mm}
\caption{\label{fig:outputgraph} \sst-output structure depicting a partial run. }
\vspace{-5mm}
\end{figure*}


Figure \ref{fig:outputgraph} gives an example of \sst-output structure. 
We show only the variable updates. 
 Dashed arrows represent variable updates for useless variables, and therefore does not belong to the \sst-output
structure. Initially the variable content of $Z$ is equal to $\epsilon$. 
It is represented by the $\epsilon$-edge from $(Z^{in}, 0)$ to $(Z^{out}, 0)$ in the first column. 
Then, variable $Z$ is updated to $Zc$. 
Therefore, the new content of $Z$ starts with $\epsilon$ (represented by the $\epsilon$-edge from $(Z^{in},1)$ to $(Z^{in},0)$,
which is concatenated with the previous content of $Z$, and then concatenated with $c$ (it is represented by the $c$-edge from
$(Z^{out},0)$ to $(Z^{out},1)$). 
Note that the invariant is satisfied. The content of variable $X$ at position 5 is given by  
the label of the path from $(X^{in},5)$ to $(X^{out},5)$, which is 
$c$. Also note that some edges are labelled by strings with several letters, but there are finitely many possible such strings. 
In particular, we denote by $O_T$ the set of all strings that appear in right-hand side of variable updates.

Let $T = (Q, q_0, \Sigma, \Gamma, \varsst, \delta, \rho, Q_f)$ be an
\sst{}.  Let $u\in (\Gamma\cup X)^*$ and $s\in
\Gamma^*$. The string $s$ is said to \emph{occur} in $u$ if $s$ is a
factor of $u$. In particular, $\epsilon$ occurs in $u$ for all $u$. 
Let $O_T$ be the set of constant strings
occurring in variable updates, i.e. $O_T = \{ s\in\Gamma^*\ |\ \exists
t\in \delta,\ s\text{ occurs in } \rho(t)\}$. Note that $O_T$ is
finite since $\delta$ is finite.

Let $w \in dom(T)$. The \emph{\sst-output graph} of $w$ by $T$, denoted by
$G_T(w)$, is defined as an infinite directed graph whose edges are labelled by
elements of $O_T$. Formally, it is the graph $G_T(w) =
(V,(E_\gamma)_{\gamma\in O_T})$ where 
$V = \{0,1,\dots,\}\times \mathcal{X}\times \{in,out\}$ is the set of
vertices, $E := \bigcup_{\gamma\in O_T} E_\gamma \subseteq V\times V$
is the set of labelled edges defined as follows.
Vertices $(i,X,d)\in V$ are denoted by $(X^d,i)$. 
Let $r = q_0\dots q_n \dots $ be an 
accepting run of $T$ on $w$. The set $E$ is defined as the smallest
set such that for all $X\in \varsst$, 
\begin{enumerate}
	\item $((X^{in},0),(X^{out},0))\in E_\epsilon$ if $(X,0)$ is useful,
	
	\item for all $i$ and $X\in X$, if $(X,i)$ is useful and
	if $\rho(q_{i},w[i+1],q_{i+1})(X) = \gamma$, then 
	$((X^{in},i+1),(X^{out},i+1))\in E_\gamma$,

	\item for all $i$ and $X\in X$, if $(X,i)$ is useful 
	and if $\rho(q_{i},w[i+1],q_{i+1})(X) = \gamma_1X_1\dots
	\gamma_kX_{k}\gamma_{k+1}$ (with $k>1$), then 
	\begin{itemize}
		
		\item $((X^{in},i+1), (X_1^{in},i))\in E_{\gamma_1}$
		\item $((X_k^{out},i), (X^{out},i+1))\in E_{\gamma_{k+1}}$
		\item for all $1\leq j< k$, $((X_j^{out},i), (X_{j+1}^{in},i))\in E_{\gamma_{j+1}}$
		
	\end{itemize}      
	
\end{enumerate}

Note that since the transition monoid of $T$ is $1$-bounded, it is
never the case that two copies of any variable (say $X$) flows into 
a variable (say $Y$), therefore this graph is
well-defined and there are \textbf{no} multiple edges between two
nodes.

We next show that the transformation that maps an $\omega$-string $s$ into its output structure 
is FO-definable, whenever the SST is 1-bounded and aperiodic. Using the fact that variable flow is FO-definable,
we show that for any two variables $X,Y$, we can capture in FO, a path from $(X^d, i)$ to $(Y^e, j)$ 
  for $d, e\in \{in, out\}$ in $G_T(s)$ and all positions $i, j$.
   We are in state $q_i$ at position $i$. 
   
    For example, 
   \begin{enumerate}
  \item  There is a path from $(Z^{d},1)$ to $(Y^{out},5)$ for $d \in \{in,out\}$.  
  This is because $Z$ at position 1 flows into $Z$ at position 4 (path $(Z^{in},4)$ to $(Z^{in},1)$, edge from 
    and $(Z^{in},1)$ to $(Z^{out},1)$, path from $(Z^{out},1)$ to $(Z^{out},4)$) 
      this value of $Z$ is used in updating $Y$ at position 5 as $Y:=bZc$. 
(edge from $(Z^{out},4)$  to $(Y^{out},5)$). 
  \item There is a path  from $(Y^{in},5)$ to $(Z^{d},2)$. This is because $Z$ at position 2 
  flows into $Y$ at position 5 by the update $Y:=YbZc$. (path from $(Z^{in},4)$ to $(Z^{in},2)$; path from 
  $(Z^{in},2)$ to $(Z^{out},2)$ and path from $(Z^{out},2)$ to $(Z^{out},4)$; lastly, edge from $(Z^{out},4)$ 
  to $(Y^{out},5)$. Also, note the edge from $(Y^{in},5)$ to $(Y^{in},4)$, and the path from 
  $(Y^{in},4)$ to $(Y^{out},4)$, edge from $(Y^{out},4)$ to $(Z^{in},4)$). 
   \item There is a path
   from $(X^{in},3)$ to $(Z^{in},1)$ in Figure \ref{fig:outputgraph}. However, $Z$ at position 1 does not flow into 
   $X$ at position 3. Note that this is because of the update $X:=XY$ at position 6, and 
   $Z$ at position 1 flows into $Y$ at position 5, (note the path from $(Z^{out},1)$ to $(Z^{out},4)$, and 
   the edge from $(Z^{out},4)$ to $(Y^{out},5)$) 
  and $X$ at position 3 flows into $X$ at position 5 (note the path from $(X^{out},3)$ to $(X^{out},5)$) and $X$ and $Y$ are catenated in order 
  at position 5 (the edge from $(X^{out},5)$ to $(Y^{in},5)$) to define $X$ at position 6 (edge from $(Y^{out},5)$ to $(X^{out},6)$).
  
    \end{enumerate}
 Thus, the SST output graphs have a nice property,
  which connects a path from $(X^d,i)$ to $(Y^{d'},j)$ based on the variable flow, and the catenation of variables 
  in updates. 
  Formally, let $T$ be an \textbf{aperiodic,1-bounded} \sst{} $T$. 
    Let $s\in dom(T)$, $G_T(s)$ its \sst-output structure and $r=q_0\dots q_n\dots$ the accepting run of $T$ on $s$. 
    For all variables $X,Y\in \varsst$, all positions $i,j\in \dom(s)\cup\{0\}$, all $d,d'\in\{in,out\}$, there exists a path from node 
    $(X^{d},i)$ to node $(Y^{d'},j)$ in $G_T(s)$ iff $(X,i)$ and
    $(Y,j)$ are both useful and one of the following conditions hold: either
           \begin{enumerate}
    \item 
$(q_i,X)\flows^{s[i{+}1{:}j]}_1 (q_j,Y)$ and $d' = out$, or
      \item
$(q_j,Y)\flows^{s[j{+}1{:}i]}_1(q_i,X)$ and $d = in$, or

    \item 
there exists $k\geq max(i,j)$ and two variables $X',Y'$ such
      $(q_i,X)\flows^{s[i{{+}1:}k]}_1 (q_k,X')$, $(q_j,Y)\flows^{s[j{+}1{:}k]}_1 (q_k,Y')$
      and $X'$ and $Y'$ are concatenated in this order\footnote{by concatenated we
        mean that there exists a variable update whose rhs is of the
        form $\dots X'\dots Y'\dots$} by $r$ when reading $s[k+1]$. 
    \end{enumerate}

\subsection{Proof of Lemma \ref{lem:fopath}}
\label{app:fopath}

 Let $s\in dom(T)$, $G_T(s)$ its \sst-output structure and $r=q_0\dots q_n\dots$ the accepting run of $T$ on $s$. 
    For all variables $X,Y\in \varsst$, all positions $i,j\in \dom(s)\cup\{0\}$, all $d,d'\in\{in,out\}$, there exists a path from node 
    $(X^{d},i)$ to node $(Y^{d'},j)$ in $G_T(s)$ iff $(X,i)$ and
    $(Y,j)$ are both useful and one of the following conditions hold: either
    \begin{enumerate}
      \item 
$(q_i,X)\flows^{s[i{+}1{:}j]}_1 (q_j,Y)$ and $d' = out$, or
      \item
$(q_j,Y)\flows^{s[j{+}1{:}i]}_1(q_i,X)$ and $d = in$, or
      \item 
there exists $k\geq max(i,j)$ and two variables $X',Y'$ such
      $(q_i,X)\flows^{s[i{{+}1:}k]}_1 (q_k,X')$, $(q_j,Y)\flows^{s[j{+}1{:}k]}_1 (q_k,Y')$
      and $X'$ and $Y'$ are concatenated in this order\footnote{by concatenated we
        mean that there exists a variable update whose rhs is of the
        form $\dots X'\dots Y'\dots$} by $r$ when reading $s[k+1]$. 
    \end{enumerate}

For all variables $X,Y\in \varsst$, we denote by $Cat_{X,Y}$ the set of 
    pairs $(p,q,a)\in Q^2\times \Sigma$ such that there exists a transition from $p$
    to $q$ on $a$ whose variable update concatenates $X$ and $Y$ (in
    this order). Define a formula for condition $(3)$:
$$
\Psi_3^{X,Y}(x,y) \ \equiv \ \exists z\cdot x\preceq z
\wedge y\preceq z 
\wedge \bigvee_{X',Y'\in\varsst, (p,q,a)\in Cat_{X',Y'}}[ \\
L_a(z)\wedge 
\phi_{X\flows X'}(x,z) \wedge 
\phi_{Y\flows Y'}(y,z)  \wedge \phi_p(z)\wedge \phi_q(z+1)]
$$

Then, formula $\text{path}_{X,Y,d,d'}(x,y)$ is defined by
$$
\begin{array}{llllllllll}
\text{path}_{X,Y,in,in}(x,y) & \equiv & 
    \phi_{Y\flows X}(y,x) \vee \Psi_3^{X,Y} \\
\text{path}_{X,Y,in,out}(x,y) & \equiv & 
    \phi_{Y\flows X}(y,x) \vee \phi_{X\flows Y}(x,y) \vee \Psi_3^{X,Y} \\
\text{path}_{X,Y,out,in}(x,y) & \equiv & 
    \text{false} \\
\text{path}_{X,Y,out,out}(x,y) & \equiv & 
    \phi_{X\flows Y}(x,y) \vee \Psi_3^{X,Y} \\
\end{array}
$$

\begin{enumerate}
\item $\text{path}_{X,Y,in,in}(x,y)$ : Recall that a path from $(X^{in},x)$ to $(Y^{in},y)$ always passes through $(X^{out},y)$. Hence, $y \leq x$, (the $X^{in}$ arrows move on the left) 
and 
it must be that there is an edge from $(X^{out},y)$ to  $(Y^{in},y)$, which happens when $X$ and $Y$ are concatenated in order. This is handled by 
 $\Psi_3^{X,Y}$. The other possibility is when $Y$ occurs in the right side of $X$ at $x$; in this case, we have 
 an edge from $(X^{in},x)$ to some $(Z^{in},x)$ ($Z$ could be $Y$), leading into a path to $(Y^{in},y)$, $y \leq x$. Clearly, here 
 $Y$ flows into $X$ between $y$ and $x$.   
 
 \item $\text{path}_{X,Y,in,out}(x,y)$ : One possibility is that there is a path from $(X^{in},x)$ to $(X^{in},z)$, with $x > z$, and $(X^{in},z)$ has an edge to 
 $(X^{out},z)$.  An edge from $(X^{out},z)$ to $(Y^{in}, z)$ happens if $X, Y$ are concatenated in order  at $z$. 
A path from $(Y^{in}, z)$ to  $(Y^{in}, y)$ and then
 $(Y^{out}, y)$, $z \geq y$ can happen. 
 This  is handled by $\Psi_3^{X,Y}$. 
 A second case is similar to 1, where we have a path from $(X^{in},x)$ to $(Y^{in},y)$, where $Y$ flows into $X$ between $y$ and $x$, and 
 then there is a path from $(Y^{in},y)$ to $(Y^{out},y)$. 
  A third  case is when $X$ occurs in the right of $Y$ at $y$; 
  in this case, we have an edge from $(X^{out},y)$ to 
  $(Y^{out},y)$. Also, there is a path from $(X^{in},y)$
to $(X^{out},y)$  which goes through $(X^{in},x)$, $x < y$. Clearly, $X$ flows into $Y$  
 between $x$ and $y$.

\item Note that $\text{path}_{X,Y,out,in}(x,y)$ is false, since there is no path or edge from $X^{out}$ to $Y^{in}$ capturing flow; 
the edge from $X^{out}$ to $Y^{in}$  occurs only when a catenation of $X, Y$ happens in order.

\item $\text{path}_{X,Y,out,out}(x,y)$ :
One case is when $X,Y$ are catenated in order on the rightside of some variable $Z$. Then there is a 
path from $(X^{out},x)$ to $(Y^{out},y)$ (through $(Y^{in},y)$), as handled by $\Psi_3^{X,Y}$. 
The other case is when $X$ occurs on the right of $Y$ at $y$. Then there is an edge from $(X^{out},x)$ to
 $(Y^{out},y)$ ($x=y-1$) (may be through some $(Z^{in},x)$).
\end{enumerate}

\subsection{Formal Construction of FOT from Aperiodic SST}
\label{app-fot-cons}
We describe the construction of the \fot{} from the 
1-bounded, aperiodic \sst{}.

 In the case of finite strings, if the output function of the accepting state 
is $X_1 \dots X_n$, then $(X_1^{in}, |s|)$ is the start node, on reading string $s$. 
The path from $(X_1^{in}, |s|)$ to $(X_1^{out}, |s|)$, followed 
by the path from $(X_2^{in}, |s|)$ to $(X_2^{out}, |s|)$ and so on gives the output. 

Unlike the finite string case, one of the  difficulties here, is to specify the start node 
of the \fot{} since the strings are infinite. 
The first thing we do, in order to specify the start node, is to identify the position where 
some Muller set starts, and the rest of the run stays in that set. This is done 
for instance, in Proposition \ref{prop:fostates} by the formula $\psi^{Rec}_P$. We can easily catch the first such position where 
$\psi^{Rec}_P$ holds.   This position will be labeled in the \fot{} with a unique symbol $\perp$. Let $O \subseteq \Gamma^*$ be the finite set 
of output strings $\gamma_i$ that appear in the variable updates of the \sst{}. That is, 
$O=\{\gamma_i \mid \rho(q,a)(X)=\gamma_0X_1 \gamma_1 \dots \gamma_{n-1} X_n \gamma_n\}$. 
We build an  \fot{}  that marks the first position where $\psi^{Rec}_P$ evaluates to true as 
as $\perp$, outputs the contents of $X_1, \dots, X_{n-1}$ first till
$\perp$, and then of $X_n$, where 
$F(P)=X_1\dots X_{n-1}X_n$ is the output function of the \sst{}. 

The \fot{} is defined as $(\Sigma, O \cup \{\perp\}, \phi_{dom}, C, \phi_{pos}, \phi_{\prec})$ where:
\begin{itemize}
\item $\phi_{dom}=is\_string \wedge \exists i\psi^{Rec}_P(i)$ where the FO formula $is\_string$ is a simple FO formula that says every position has a 
unique successor and predecessor. The conjunction with $\exists i\psi^{Rec}_P(i)$ says that there is a  position  $i$
of the string from where   $\psi^{Rec}_P$ is true. 
\item $C=\mathcal{X} \times \{in, out\}$, 
\item $\phi_{pos}=\{\phi_{\gamma}^c(i) \mid c \in C, \gamma \in O \cup \{\perp\}\}$ is such that 
\begin{itemize}
\item 
$\phi_{\perp}^{(X_1,in)}(i)=\bigvee_{\{P \in F \mid F(P)=X_1 \dots X_n\}}[\psi^{Rec}_P(i) \wedge \forall j(j \prec i \rightarrow \neg\psi^{Rec}_P(j))]$. 
$i$ is the first position from where a Muller set starts to hold continuosuly. 
 \item for $\gamma \in O$, we define $\phi_{\gamma}^c(i)=\neg \phi_{\perp}^c(i) \wedge \psi_{\gamma}^c(i)$, where 
$\psi_{\gamma}^c(i)$ is 
  
\end{itemize}

\begin{itemize}
\item true, if  $(i=0$), $c=(X,in)$ and $\gamma=\epsilon$,
\item false, if  $(i=0$), $c=(X,in)$ and $\gamma \neq \epsilon$,
\item $\bigvee_{q \in Q, a \in \Sigma}\phi_q(i-1) \wedge L_a(i-1)$ if $c=(X,in)$, $i > 0$
and $\rho(q,a)(X)=\gamma Y_1\gamma_1 Y_2 \dots \gamma_n$. The formula $\phi_q(x)$ is defined in Proposition \ref{prop:fostates}.
\item $\bigvee_{q \in Q, a \in \Sigma}\phi_q(i) \wedge L_a(i)$ if $c=(X,out)$, $i \geq 0$
and $\rho(q,a)(Y)=\gamma_0 Y_1\gamma_1 Y_2 \dots Y_k\gamma_k \dots \gamma_n$, for a unique $Y \in \mathcal{X}$, and 
$Y_k=X$ and $\gamma=\gamma_k$ for some $1 \leq k \leq n$. Note that $Y \in \mathcal{X}$ is unique 
since the \sst{} is 1-bounded (the \sst{} output-structure does not contain useless variables, and 
if $X$ appears in the rightside of two variables $Y,Z$ then one of them will be useless for the output) 
\item $\bigvee_{q \in Q, a \in \Sigma}\phi_q(i) \wedge L_a(i)$ if $c=(X,out)$, $i \geq 0$
and  $\gamma=\epsilon$ if $X$ does not appear in the update of any variable in $\rho(q,a)$.  
\end{itemize}
\item We next define $\phi^{c,d}(i,j)=\bigvee_{P \in F}[\psi^{Rec}_P \wedge \psi^{c,d}_P(i,j)]
$, where 
$\psi^{c,d}_P(i,j)$ is defined as follows. Let $F(P)=X_1X_2 \dots X_{n-1}X_n$. 
Let $k$ be the earliest position where $\psi^{Rec}_P$ holds.  Note that all 
these edges are defined only when a copy is useful, that is, it contributes to the 
output. The notion of usefulness has been defined in section \ref{sec:sst-fot}. 
\begin{itemize}
\item $true \wedge useful_X(i)$, if $c=(X,in)$, $d=(X,out)$ for some variable $X$, 
and $i=j=0$. The formula $useful_X(i)$ is defined in Appendix \ref{app:useful}.  
\item $\phi_q(i-1) \wedge L_a(i-1) \wedge useful_X(i)$ if $c=(X,in), d=(X,out)$ 
$i=j>0$ and $\rho(q,a)(X)=\gamma$.  
\item $\phi_q(i) \wedge L_a(i) \wedge useful_X(i)$ if $c=(X,out)$, $d=(Y,in)$ 
for some variables $X,Y$, $i=j\geq 0$, and for some $Z$ we have 
$\rho(q,a)(Z)=\gamma_0 Z_1 \gamma_1 \dots Z_n \gamma_n$ with $Z_{\ell}=X$ and 
$Z_{\ell+1}=Y$ for some $1 \leq \ell \leq n$
\item $\phi_q(i-1) \wedge L_a(i-1) \wedge useful_X(i)$ if $C=(X,in)$, $d=(Y, in)$  
for some variables $X,Y$, $i >0$, $j=i-1$, $\rho(q,a)(X)=\gamma Y \gamma_1 Y_2 \dots Y_n \gamma_n$ 
\item $\phi_q(i) \wedge L_a(i) \wedge useful_X(i)$ if $c=(X,out)$, $d=(Y,out)$ for some variables $X,Y$,  
$i+1=j$ and $\rho(q,a)(Y)=\gamma X_1 \dots X \gamma_n$.
\item We do the above  between copies upto position $k$, which influences 
the values of variables $X_1, \dots, X_{n-1}$, which 
 produces the variable flow upto position $k$. 
Beyond position $k$, the contents of variables $X_1, \dots, X_{n-1}$ remain unchanged.  
Hence, for all $i \geq k$, and $1 \leq j \leq n-1$, we dont have any edges from $(X_j,out)$ 
at position $i$ to 
  $(X_j,out)$ at position $i+1$, and from $(X_j,in)$ at position $i+1$ 
to $(X_j,in)$ at position $i$.  
We simply add a transition from $(X_j,out)$ to $(X_{j+1},in)$ at position $k$ to 
seamlessly catenate the contents of $X_1, \dots, X_{n-1}, X_n$ 
at position $k$.  
\item From $(X_n,in)$ at position $k$ onwards, we simply follow the edges as defined above, to obtain 
at each position, the correct output $X_1X_2 \dots X_n$. If at position $k$, 
$X_n:=X_n  \gamma_0 Y_1 \dots \gamma_{n-1} Y_n \gamma_n$ then we have the connections 
from $(X_n, in)$ at position $k$ to $(X_n, in)$ at position $k-1$, which will 
eventually reach $(X_n, out)$ at position $k-1$. This is then connected to 
$(Y_1, in)$ at position $k-1$, and so on till we reach 
$(Y_n,out)$ at position $k-1$. This  is then connected to 
$(X_n, out)$ at position $k$, rendering the correct output at position $k$. The same thing repeats 
for position $k+1$ and so on.  
\end{itemize}
\end{itemize}
Thanks to Lemma \ref{lem:fopath}, the transitive closure between some copy $(X,d)$ and  $(Y,d')$ is FO-definable 
as $\phi^{(X,d),(Y,d')}_{\prec}=path_{X,Y,d,d'}(x,y)$. This completes the construction 
of the \fot{}.

\section{Proofs from Section \ref{sec:2wst-sst} : $\twst_{\la} \subset \mathsf{SST}_{\la}$}
\label{app-2wst-sst}
In this section, we show that $\twst_{\la} \subset \mathsf{SST}_{\la}$. 
Let the \twst$_{\la}$ be $(T, A, B)$ with $T = (\Sigma, \Gamma, Q, q_0, \delta, F)$. 
We assume wlg that $L(A_p) \cap L(A_{p'})=\emptyset$ for all $p, p' \in Q_A$ and 
$L(A_r) \cap L(A_{r'})=\emptyset$ for all $r, r' \in Q_B$. Also, we assume that 
$F=2^Q$, and check 
 particular accepting conditions in  the look-ahead automaton $A$.

The \sst$_{sf}$ we construct is $\Tt=((\Sigma, \Gamma, Q', q_0', \delta', X, \rho, F'), A, B)$ where 
$Q'=Q \times [Q \rightharpoonup Q]$ where $\rightharpoonup$ represents  a partial function. $q'_0=(q_0, I)$ where $I$ is the identity. The set of variables
$\mathcal{X}=\{X_q \mid q \in Q\} \cup \{O\}$, 
and $F': 2^{Q'} \rightarrow X^*$ is defined for all $P' \in F'$ as $F'(P')=O$. 
$\delta': Q' \times Q_B \times \Sigma \times Q_A \rightarrow Q'$ is defined as follows:
$\delta'((q,f),r,a,p)=(f'(q),f')$ where $f'$ is defined as follows:
\[
f'(q)=\left \{
  \begin{array}{lll}
  & t & \mbox{if}~ \delta(q,r,a,p)=(t, \gamma, 1),\\
  & f'(t) & \mbox{if}~ \delta(q,r,a,p)=(t, \gamma, 0),\\
  & f'(f(t)) & \mbox{if}~ \delta(q,r,a,p)=(t, \gamma, -1)
  \end{array}\right.
\]
  Thanks to the determinism of the \twst$_{sf}$ and the restriction that the entire input be read, 
  $f'$ is a well-defined, partial function. In each cell $i$, the state of the \sst$_{sf}$ will coincide with the state 
  the \twst$_{sf}$ is in, when reading cell $i$ for the first time.

  The variable update $\rho$ on each transition is defined as follows:
  for $\delta'((q,f),a)=(f'(q),f')$, the update $\rho((q,f),a)$ is defined as follows:
  \[
  \rho((q,f),a)(Y)= \left \{
  \begin{array}{ll}
  O\rho(X_q) & \mbox{if}~Y=O,\\
  \gamma & \mbox{if}~Y=X_q~\mbox{and}~\delta(q,r,a,p)=(t, \gamma,1),\\
  \gamma\rho(X_t) & \mbox{if}~Y=X_q~\mbox{and}~\delta(q,r,a,p)=(t, \gamma,0),\\
  \gamma X_t \rho(X_{f(t)}) & \mbox{if}~Y=X_q~\mbox{and}~\delta(q,r,a,p)=(t, \gamma,-1)~\mbox{and}~f(t)~\mbox{is defined},\\
  \epsilon & \mbox{otherwise}.
    \end{array}\right.
    \]
At the moment, the \sst$_{sf}$ is not 1-bounded. 
One place where 1-boundedness is violated happens 
because it is possible to have $f(s_1)=f(s_2)=t$ for $s_1 \neq s_2$. In particular, 
assume $\delta(s_1,r,a,p)=\delta(s_2,r,a,p)=(t,\gamma,-1)$, and let $f(t)=q$. Then we have 
$X_{s_2}, X_{s_1}:=\gamma X_t \rho(X_q)$. This would mean that we visit the same cell in the same state $t$ 
twice, which would lead to non-determinism or not fully reading the input by the \twst$_{sf}$. 
Thus, this cannot happen in an accepting run of the \twst$_{sf}$. Hence, if we have 
$\delta'((q,f),r,a,p)=(f'(q), f')$ where $f'$ is a partial function which is not 1-1, we can safely replace 
it with a subset $f''$ of $f'$ which is 1-1. Since $f'(s_1)=f'(s_2)$ for $s_1 \neq s_2$ 
does not contribute to an accepting run, we can define $f''$ on those 
states which preserve the 1-1 property. 
The second place where 1-boundedness is violated is when the contents of $\rho(X_q)$ flows into $O$ as well as $X_q$. 
This can be fixed by composition with another \sst{} where on each transition, all variables other than $X_q$ 
are unchanged ($X:=X$ for $X \neq X_q$) and resetting $X_q$ ($X_q:=\epsilon$). Note that this does not affect the ouput since 
$O$ is the only variable that is responsible for the output, and $\rho(X_q)$ has faithfully been reflected into it. 
With these two fixes, we regain 
the 1-boundedness property. 

\subsection{Connecting runs of the \twst$_{sf}$ and \sst$_{sf}$}

\begin{lemma}
Let $a_1a_2 \dots \in \Sigma^{\omega}$, and let 
$r=(q_0, i_0=0) \stackrel{r_1,s(i_0),p_1|z_1}{\rightarrow}(q_1,i_1)\stackrel{r_2,s(i_1),p_2|z_2}{\rightarrow}(q_2,i_2)\dots$  be an accepting run of the \twst$_{sf}$. 
Let $j_n=min\{m \mid i_m=n\}$ be the index when cell $n$ is read for the first time. Let 
$r'=(q'_0, f_0) \stackrel{r'_1,a_1,p'_1}{\rightarrow}  (q'_1, f_1) \stackrel{r'_2,a_2,p'_2}{\rightarrow}(q'_2, f_2)  \dots$ be the 
corresponding run in the constructed \sst$_{sf}$. Then we have 
\begin{itemize}
\item $q'_i=q_{j_i}, a_i=s(j_i), r'_i=r_{j_i}, p'_i=p_{j_i}$,
\item the output $\sigma'(O)$ till position $i$ of the input is $z_1z_2 \dots z_{j_i}$
\item In short, if the $\twst_{sf}$ reads $a_i$ with $r_i \in Q_B, p_i \in Q_A$, and state $q$, then the \twst$_{sf}$ will be in state 
  $f_i(q)$ when it next reads $a_{i+1}$, and produces $\sigma_{r', i}(X_q)$ as output.  
\end{itemize}
\end{lemma}
\begin{proof}
The proof is by induction on the length of the prefixes of the run $r'$. Consider a prefix of length 1.
In this case, we are at the first position of the input, in initial state 
$q'_0=q_0$, since the  \twst$_{sf}$  moves right. We also have the states $r'_1=r_1, p'_1=p_1$, and 
the output obtained as $O:=O \rho(X_{q_0})$, and $\rho(X_{q_0})=z_1$. Thus, for a prefix of length 1, the states coincide, and the output
in $O$ coincides with the output produced so far.  

Assume that upto a prefix of length $k$, we obtain the conditions of the lemma, and consider a prefix of length $k+1$. \\

Let $r=(q_0, i_0=0) \stackrel{r_1,s(i_0),p_1|z_1}{\rightarrow}(q_1,i_1)\dots
(q_{k-1},i_{k-1})\stackrel{r_{k},s(i_{k-1}),p_{k}|z_{k}}{\rightarrow}(q_k,i_k)\stackrel{r_{k+1},s(i_k),p_{k+1}|z_{k+1}}{\rightarrow}(q_{k+1},i_{k+1})\dots
$ be a prefix of length $k+1$. The corresponding run $r'$ in the \sst$_{sf}$ is
$r'=(q'_0, f_0) \stackrel{r'_1,a_1,p'_1}{\rightarrow}  (q'_1, f_1) \dots (q'_{k-1}, f_{k-1}) \stackrel{r'_k,a_k,p'_k}{\rightarrow}(q'_k, f_k) \stackrel{r'_{k+1},a_{k+1},p'_{k+1}}{\rightarrow}(q'_{k+1},f_{k+1})\dots $, where by inductive hypothesis, we know that 
$q'_k=q_{j_k}, a_k=s_{j_k}, r'_k=r_{j_k}, p'_k=p_{j_k}$, and the output $\sigma'(O)$ till position $k$ 
of the input is $z_1z_2 \dots z_{j_k}$. $q'_k=q_{j_k}$ is the state the \twst$_{sf}$ is in, when it moves to the right of cell $k$ containing $a_k$, for the first time, and reads 
$a_{k+1}$. 

By construction of the \sst$_{sf}$, $\delta'((q_{j_k}, f), r_{j_k}, p_{j_k})=(f'(q_{j_k}), f')$, where $f'$ is defined based on 
the transition in the \twst$_{sf}$ from state  $q_{j_k}$ on reading $a_{k+1}$. 
\begin{enumerate}
\item If $\delta(q_{j_k},a,r_{j_k}, p_{j_k})=(t, \gamma,1)$, then by construction of the \sst$_{sf}$,  we obtain 
$f'(q_{j_k})=t$ and $\rho((q_{j_k},f),a)(O)=O\rho(X_{q_{j_k}})=O\gamma$. By inductive hypothesis, 
the contents of $O$ agrees with the output till $k$ steps. By catenating $\gamma$,  
in this case also, we obtain the same situation. 
\item If $\delta(q_{j_k},a,r_{j_k}, p_{j_k})=(t, \gamma,0)$, then 
by construction of the \sst$_{sf}$,  we obtain 
$f'(q_{j_k})=f'(t)$ and $\rho((q_{j_k},f),a)(O)=O\gamma\rho(X_t)$. Inducting on the number of steps it takes for 
the \twst$_{sf}$ to move to the right of cell $k+1$, we can show that the condition in the lemma holds. 
\begin{itemize}
\item 
If the \twst$_{sf}$ moves to the right of cell $k+1$ in the next step (one step), in state $u=q'_{k+1}$, then we obtain
$f'(t)=u$, where  $\delta(t,a,r,p)=(u,\gamma',1)$, in which case we obtain 
$\rho((q_{j_k},f),a)(O)=O\gamma\gamma'$ ($O \rho(X_t)$, where the current contents of $O$ is the old contents of $O$ followed by $\gamma$ and 
$\rho(X_t)=\gamma'$) which indeed is the output obtained 
in the \twst$_{sf}$ when it moves to the right of cell $k+1$. 
\item Assume as inductive hypothesis, that the result holds (that is, the output $O$ agrees with the output in the 
\twst$_{sf}$ so far, and that the state of the $\sst_{sf}$ agrees with the state of the \twst$_{sf}$) when it moves to the right 
of cell $k$ in $\leq m$ steps.   
Consider now the case when it moves to the right of cell $k+1$ in $m+1$ steps.

At the end of $m$ steps, the \twst$_{sf}$ is on cell $k+1$, in some state 
$q$, and in the constructed \sst$_{sf}$, by hypothesis, the contents of $O$ correctly reflect the output so far.
By assumption, we have $\delta(q, a_{k+1}, r,p)=(q'', \gamma'',1)$ for some state $q''$. Then we obtain, by 
the inductive hypothesis applied to cell $k$ at the end of $m$ steps, and the construction of the \sst$_{sf}$, 
the \sst$_{sf}$  to be in state 
$(f'(q),f')$, where $f'(q)=q''$, and 
   $\rho((q'',f),a_{k+1})(O)=O\gamma''$. Since $O$ is up-to-date 
with the output with respect to the \twst$_{sf}$, catenating $\gamma''$ indeed keeps it up-to-date, and 
the state of the \sst$_{sf}$ $q'_{k+1}$ is indeed the state the \twst$_{sf}$ is in, when it moves to the right of cell $k+1$ 
for the first time.  
  \end{itemize}
\item If $\delta(q_{j_k},a,r_{j_k}, p_{j_k})=(t, \gamma,-1)$, then 
by construction of the \sst$_{sf}$,  we obtain 
$f'(q_{j_k})=f'(f(t))$ and $\rho((q_{j_k},f),a)(O)=O\gamma X_t \rho(X_{f(t)})$. 
The \twst$_{sf}$ is in cell $k$ at this point of time.  
Inducting on the number of steps it takes for 
the \twst$_{sf}$ to move to the right of cell $k+1$ as above, we can show that the condition in the lemma holds.

\end{enumerate}

\end{proof}
\subsection{Aperiodicity of the  \sst$_{sf}$}
\label{app:aper-sst}
Having constructed the \sst$_{sf}$, the next task is to argue that the \sst$_{sf}$ is aperiodic. 
\begin{lemma}
The constructed \sst$_{sf}$  is aperiodic.
\end{lemma}
\begin{proof}
We start with an aperiodic \twst$_{sf}$ $\Aa$ and obtain the \sst$_{sf}$ $\Tt$. Since $\Aa$ is aperiodic, there exists an $m \in \N$  such that $u^m \in L_{pq}^{xy}$ iff $u^{m+1} \in L_{pq}^{xy}$ for all $p, q \in Q$ and $x, y \in \{\ell, r\}$.  
To show the aperiodicity of the \sst$_{sf}$, we have to show that the transition monoid which captures state and variable flow is aperiodic. 
Lets consider a string $u \in \Sigma^*$ and a run in the \sst$_{sf}$ $(q,f) \stackrel{u}{\rightarrow} (q',f')$. 
Based on the transitions of the \sst$_{sf}$, we make a basic observation on variable assignments as follows. 
We show a correspondence between the flow of states for each kind of traversal in the \twst$_{sf}$, and the respective state and variable flow
in the constructed \sst$_{sf}$. Using this information and the aperiodicity of the  \twst$_{sf}$, we prove the aperiodicity 
of the  \sst$_{sf}$ by showing the state, variable flow to be identical for strings $w^m, w^{m+1}$, for 
some chosen $m \in \mathbb{N}$, and any $w$.

We know already that 
if $u=a$, and if $\delta(q,r,a,p)=(\gamma, q', 1)$ in the \twst$_{sf}$ $\Aa$, then in the \sst$_{sf}$ we have the variable update 
$\rho(X_q)=\gamma$.
 If $u=a$, and if $\delta(q,r,a,p)=(\gamma, q', -1)$ in the \twst$_{sf}$ $\Aa$, then in the \sst{} we have the variable update
  $\rho(X_q)=\gamma X_{q'} \rho(X_{f(q')})$. 
We now generalize the above for any $u \in \Sigma^*$ as follows:
\begin{itemize}
\item[(a)] If the run of $u$ in the \twst$_{sf}$ $\Aa$ starts at the rightmost letter of $u$ in state $q$, and 
leaves it at the right with output $\gamma$, then we have $\rho(X_q)=\gamma$.
\item[(b)] Assume the run of $u$ in the \twst$_{sf}$ $\Aa$ starts at the rightmost letter of $u$ in state $q$, and 
leaves it at the left in state $t_1$. Since the whole input 
is read,  eventually the \twst$_{sf}$ will leave $u$  at the right in state $f'(q)$.  Then  
we have $\rho(X_q)=\gamma_0 X_{t_1} \gamma_1 X_{t_2} \dots \gamma_{n-1} X_{t_n} \gamma_n$ such that 
(i) $u \in L_{q,t_1}^{r\ell}$ with output $\gamma_0$, (ii) $u \in L_{f(t_1),t_2}^{\ell\ell}$ with output $\gamma_1$. 
In general,    $u \in L_{f(t_i),t_{i+1}}^{\ell\ell}$ with output $\gamma_i$, and (iii) since $\Aa$ completely reads 
the input, at some point, $u \in L_{f(t_n),f'(q)}^{\ell,r}$ with output $\gamma_n$. 
Also, in state $t_i$, we move one position left from some cell $h$,  the leftmost position 
of $u$, and in state $f(t_i)$ we move to the right 
of cell $h$ for the first time, and the output generated in the interim is obtained from $\rho(X_{t_i})$.  
\end{itemize}
We will show the above  (a), (b)
by inducting on $u$. For the base case $u=a$, we already 
have the proof, by the construction of the updates in the \sst$_{sf}$.
 Let $p$ be any state. 
 Consider the case when in the \twst$_{sf}$ we start at the right of $ua$ in state $p$ and leave it on the right.
 \begin{itemize}
\item If we start at state $p$ on $a$, and $\delta(p,a)=(\gamma, q, -1)$, then 
we are at the right of $u$ in state $q$. We will be back on $a$ in state 
$f(q)$. By inductive hypothesis on $u$, we have
$\rho(X_q)$ as some string over $\Gamma^*$, 
since we will leave $u$ at the right starting in state $q$
 in the right.  
Again, by inductive hypothesis for $a$, we have 
$\rho(X_p)=\gamma X_q \rho(X_{f(q)})$; by the inductive hypothesis 
on $u$, when we leave $u$ on the right starting in state $q$ on the rightmost symbol of $u$, we obtain $\rho(X_q)$ as a constant.
We are now on $a$ in state $f(q)$. 
If we leave $a$ to the right starting at $f(q)$, then we have by inductive hypothesis, 
$\rho(X_{f(q)})$ is a constant, and then we obtain 
$\rho(X_p)$ as a constant, while considering moving to the right of $ua$ starting at 
$a$ in state $p$. This agrees with (a) above.  
\item If we start at state $p$ on $a$, and $\delta(p,a)=(\gamma, q, 0)$.
By construction of the \sst$_{sf}$, we have $\rho(X_p)=\gamma \rho(X_q)$. As long as we stay in the same position, we continue 
obtaining the same situation. If we leave $a$ to the right eventually, then by inductive hypothesis, we obtain 
$\rho(X_p)$ as a string over $\Gamma^*$. If from some state $t$, while on $a$, we move one position to the left, (on the rightmost symbol of $u$),
then by inductive hypothesis on $u$, we obtain
$\rho(X_t)$ as a constant; finally, since we move to the right of $a$ starting in state $f(t)$, 
 we obtain $\rho(X_{f(t)})$ as a string over $\Gamma^*$, and hence, 
$\rho(X_p)$ as well.

Thus,  the inductive hypothesis works since $f(p)=f(q)$.
\item If  we start at state $p$ on $a$, and $\delta(p,a)=(\gamma, q, 1)$, then 
we straightaway have $\rho(X_p)$ as a constant, and are done.   
\end{itemize}
Now consider the case when we start at the right of $ua$ and leave it at the left. Let us start in state 
$p$ on $a$.
\begin{itemize}
\item If we start at state $p$ on $a$, and $\delta(p,a)=(\gamma, q, -1)$, then 
we are at the right of $u$ in state $q$. 
By inductive hypothesis for $a$ we also have 
$\rho(X_p)=\gamma X_q \rho(X_{f'(q)})$. 
If we leave $ua$ at the left 
in state $r$, then starting in the right of $u$ in state $q$, we leave it at the left 
in $r$. By inductive hypothesis on $u$, we have 
$\rho(X_q)=\gamma_0 X_{r} \dots X_{r'}\gamma'$, such that  $u \in L_{f(r'),f'(q)}^{\ell r}$.   
Leaving $a$ on the right from state $f'(q)$ 
gives by inductive hypothesis $\rho(X_{f'(q)})$ as a constant string $\beta$.

On $ua$ starting in state $q$ on the right, the content of 
$X_p$ is then $\gamma \gamma_0 X_r \dots X_{r'} \gamma' \rho(\rho(X_{f'(q)}))$
=$\gamma \gamma_0 X_r \dots X_{r'} \gamma' \beta$, which agrees with (b), since 
$r, \dots, r'$ are the states in which we leave $u$ at the left.

\item If we start at state $p$ on $a$, and $\delta(p,a)=(\gamma, q, 0)$. 
Then the inductive hypothesis works since $f(p)=f(q)$.
\item Lastly, the case $\delta(p,a)=(\gamma, q, 1)$ does not apply since 
we are leaving $ua$ at the left starting on $a$ in state $p$. 
\end{itemize}
The aperiodicity of the \sst$_{sf}$ is now proved as follows.
Since the \twst$_{sf}$ is aperiodic, we know that there is an $m \in \mathbb{N}$ such that 
$u^m \in L_{pq}^{xy}$ iff  
$u^{m+1} \in L_{pq}^{xy}$. Moreover, we also know that the matrices 
of $u^m, u^{m+1}$ are identical in the \twst$_{sf}$, which tells us that the sequence of states 
seen in the left, right traversals of $u^m, u^{m+1}$ are identical. If this is the case, then in the above argument, we 
obtain the variable substitutions and state, variable flow 
for $u^m$ and $u^{m+1}$ to be identical, since the state and variable flow 
of the \sst$_{sf}$ only depends on the state flow of the \twst$_{sf}$.
Then we obtain the transition monoids of $u^m, u^{m+1}$ to be the same in the \sst$_{sf}$. 
The states of the look-behind and look-ahead automata in the \sst$_{sf}$ also follow the same sequence of states  
of the look-behind and look-ahead automata in the \twst$_{sf}$.
 Since the look-behind and look-ahead automata are aperiodic, 
we obtain $M_{u^m}=M_{u^{m+1}}$ in the \sst$_{sf}$. Hence the \sst$_{sf}$ is aperiodic. 
\end{proof}

\section{Proofs from Section \ref{sec:2wst-sst} : $\sst_{\la} \subseteq \sst{}$}
\label{app:starfree}
We now show that we can eliminate the star-free look-around from the $\sst_{sf}$ $(T,A,B)$  without losing expressiveness.
Eliminating the look behind is easy: $B$ can be simulated by computing  
for each state $p_B \in P_B$, the state $p_B'$ of $P_B$ reached by $B$ starting in
$p_B$ on the current prefix, and whenever $p_B'$ is a final state of $B$, the transition is triggered.

In order to remove the look-ahead, we need to keep track of $P_i \in2^{P_A}$ at
every step. On processing a string  $s=a_1 a_2 a_3 \ldots \in \Sigma^{\omega}$ 
in $T$ starting with $P_0=\emptyset$, 
we obtain successively $P_1, P_2, \dots$ where 
$P_{i+1}= P_{i+1}' \cup \{p_{i+1}\}$ such that  
$\delta(q_i, r_{i+1}, a_{i+1}, p_{i+1}) =  q_{i+1}$, and for all $p\in P_i$,
$\delta_A(p,a_{i+1})\in P_{i+1}'$. Thus, starting with $P_0=\emptyset$ and 
$\delta(q_0, r_{1}, a_{1}, p_{1}) =  q_{1}$, we have
$P_1=\{p_1\}$, $P_2=\delta_A(p_1, a_2) \cup \{p_2\}$ and so on. 
A configuration of the \sst$_{sf}$ is thus a tuple $(q, (r'_1, \dots, r'_n), 2^{P_A})$
where $r'_i$ is the state reached in $B$ on reading the current prefix from 
state $r_i$, assuming $P_B=\{r_1, \dots, r_n\}$,  
and $2^{P_A}$ is a set of states in $P_A$ obtained as explained above. 
 We say that
$\rho$ is an accepting run of $s$ if $(q_0,(r_1, \dots, r_n), P_0)$ is an
 initial configuration, i.e.  
$q_0\in Q_0, P_0=\emptyset$ and after some point, we only see all elements of
 exactly one Muller set $M_i$ repeating infinitely often in $Q$ as well as in
 $P_A$ i.e. $\Omega(\rho)_{1} = M_i$ from domain of $F$ and $\Omega(\rho)_2 = M_j$
 from $P_f$. Also, $(q_i, (r'_1, \dots, r'_n), P_i) \stackrel{r_{k}, a,
   p_{i+1}}{\longrightarrow} (q'_i, (r''_1, \dots, r''_n), P_{i+1})$ iff
 $\delta(q_i,a,r_k,p_{i+1})=q'_i$,  
$P_{i+1}= \delta^*_A(P_{i},a) \cup \{p_{i+1}\}$ and
 $\delta_B(r'_j,a)=r''_j$ for $1 \leq j \leq n$ and the state $r''_k$ is an
 accepting  state of $B$.

A configuration in the \sst$_{sf}$ is said to be \emph{accessible} if
it can be reached from an initial configuration, and
\emph{co-accessible} if from it accepting configurations can be
reached. It is \emph{useful} if it is both accessible and
co-accessible. Note that from the mutual-exclusiveness of
look-arounds and the determinism of $A,B$, it follows that for any input string, there is at most one
run of the $\mathsf{SST}_{\la}$ from and to useful configurations, as shown in
Appendix \ref{app:unique}. 

The concept of substitutions induced by a run can be naturally extended from
$\mathsf{SST}$ to $\mathsf{SST}_{\la}$. 
Also, we can define the transformation implemented by an $\mathsf{SST}_{\la}$
in a straightforward manner.
The transition monoid of an $\mathsf{SST}_{\la}$ is defined by matrices indexed by
configurations $(q_i, (r_1, \dots, r_n), P_i)\in Q\times P_B^n \times 2^{P_A}$, using the notion of run defined before,
and the definition of aperiodicity of $\mathsf{SST}_{\la}$ follows that of $\mathsf{SST}$.
\subsection{Uniqueness of Accepting Runs in $\mathsf{SST}_{\la}$}
\label{app:unique}
Let $a_1a_2\dots \in \Sigma^{\omega}$ and $\rho: (q_0, (r_1, \dots, r_n), P_0) \stackrel{a_1}{\rightarrow}(q_1, (r_1^1, \dots, r_n^1), P_1)\stackrel{a_2}{\rightarrow}(q_2, (r_1^2, \dots, r_n^2), P_2) \dots$ be an accepting run in the $\mathsf{SST}_{\la}$. We show that $\rho$ as well as the sequences of transitions associated with $\rho$ are unique. 
Given a sequence of transitions of the $\mathsf{SST}_{\la}$, it is clear that there is exactly one run since both $A, B$ are deterministic. Lets assume that the sequence of transitions are not unique, that is there is another accepting run $\rho'$ for $a_1a_2 \dots$. Let $i$ be the smallest index where 
$\rho$ and $\rho'$ differ. The $(i-1)$th configuration is then some $(q_i, (r'_1, \dots, r'_n), P_i)$ in both $\rho, \rho'$. 
Let us assume that we have two transitions $\delta(q_i, r_j, a_i, p_i)=q_{i+1}$ and 
 $\delta(q_i, r_k, a_i, p'_i)=q'_{i+1}$ enabled such that $r_j \neq r_k$ or $p_i \neq p'_i$ or $q_{i+1} \neq q'_{i+1}$. 
 Assume $r_j \neq r_k$. 
 Since both are trigerred, we have the prefix upto now is in $L(B_{r_j}) \cap L(B_{r_k})$ with $r_j \neq r_k$, which contradicts mutual exclusiveness of look-behind. 
 If $p_i \neq p'_i$, then since both runs are accepting, we have the infinite suffix  in 
 $L(A_{p_i}) \cap L(A_{p'_i})$, which contradicts the mutual exclusiveness of look-ahead. 
 If $r_j=r_k$ and $p_i=p'_i$, but $q_{i+1} \neq q'_{i+1}$, then $\delta$ is not a function, which is again a contradiction. 

\subsection*{Variable Flow and Transition Monoid of $\mathsf{SST}_{\la}$}
\textbf{Variable Flow and Transition Monoid}. 
Let $P_A, P_B$ represent the states of the (deterministic) lookahead and look-behind automaton $A,B$, and $Q$ denote states of the 
 $\mathsf{SST}_{\la}$.

The transition monoid of an $\mathsf{SST}_{\la}$ depends on its configurations
and variables. It extends the notion of transition monoid for $\mathsf{SST}$ with look-behind, look-ahead states
components but is defined only on {\it useful} configurations $(q,(r'_1, \dots, r'_n), P)$. 
A configuration $(q,(r'_1, \dots, r'_n), P)$ is {\it useful} iff it is {\it accessible} and 
{\it co-accessible} : that is, $(q,(r'_1, \dots, r'_n), P)$ is reachable from the initial
configuration $(q_0,(r_1, \dots, r_n),\varnothing)$ (here, $P_B=\{r_1, \dots, r_n\}$) and will reach a 
configuration $(q',(r''_1, \dots, r''_n),P')$ 
from where on, some muller subset of $Q$ is witnessed continuously, and some muller subset of $P_A$ is witnessed continuously.

Note that given two useful configurations $(q,(r'_1, \dots, r'_n),P)$, $(q',(r''_1, \dots, r''_n), P')$ and a
string $s\in\Sigma^*$, there exists \emph{at most} one run from
$(q,(r'_1, \dots, r'_n),P)$ to $(q',(r''_1, \dots, r''_n), P')$
on $s$. 
 Indeed, since $(q,(r'_1, \dots, r'_n), P)$ and 
 $(q',(r''_1, \dots, r''_n), P')$ 
  are both useful, 
there exists $s_1,s_2\in\Sigma^*$ such that 
$(q_0, (r_1, \dots, r_n),\varnothing)\flows^{s_1} (q,(r'_1, \dots, r'_n),P)$ and $(q',(r''_1, \dots, r''_n),P') \flows^{s_2} (q'', (r'''_1, \dots, r'''_n), P'')$
such that from $(q'', (r'''_1, \dots, r'''_n), P'')$, 
 we settle in some Muller set of both $Q$ and  $P_A$ reading some $w \in \Sigma^{\omega}$. 
   If there are two runs from
$(q,(r'_1, \dots, r'_n),P)$ to 
$(q',(r''_1, \dots, r''_n),P')$ on $s$, then 
there are two accepting runs 
for $s_1ss_2w$, 
which contradicts the fact that accepting runs are
unique. 
We denote by 
$\textsf{useful}(T,A,B)$ the useful configurations of $(T,A,B)$.

Thanks to the uniqueness of the sequence of transitions associated
with the run of an $\mathsf{SST}_{\la}$ from and to useful
configurations on a given string, one can extend the notion of
variable flow naturally by considering, as for $\mathsf{SST}$, the composition
of the variable updates along the run. 

Assume that $T$ has $j$ muller sets and $B$ has $k$ muller sets. 
A string $s\in \Sigma^*$ maps to a square  matrix $M_s$ of dimension
$|Q\times (P_B \times \dots \times P_B) \times 2^{P_A}|\cdot |\varsst|$ and is defined as
\begin{itemize}
\item 
  $M_{s}[(q, (r'_1, \dots, r'_n),P), X][(q', (r''_1, \dots, r''_n), P'),
X']=n, \alpha_1, \alpha_2$ if there exists a run $\rho$ from 
\\ $(q, (r'_1, \dots, r'_n), P)$ to
$(q', (r''_1, \dots, r''_n), P')$ on $s$ such that $n$ copies of $X$ flows to $X'$ over the
run $\rho$, and $(q, (r_1, \dots, r_n), P)$ and $(q', (r''_1, \dots, r''_n), P')$ are both useful  (which implies
that the sequence of transitions of $\rho$ from $(q, (r_1, \dots, r_n), P)$ to $(q', (r''_1, \dots, r''_n), P')$ is
unique, as seen before), and 
\item $\alpha_1 \in \{0,1,\kappa\}^j$ and $\alpha_2 \in 2^{[\{0,1,\kappa\}^k]}$. 
$\alpha_1$ is a $j$-tuple keeping track for each of the $j$ muller sets of $T$, whether
the set of states seen on reading a string $s$ is a muller set, a strict subset of it, or has 
seen a state outside of the muller set. Likewise, $\alpha_2$ is a set of $k$-tuples doing the same thing
based on the transition from set $P$ to set $P'$.
 Note that since we keep  a subset 
of states of $P_A$ in the configuration, we need to keep track of this information for each state in the set.  
Thus, $\alpha_1$ keeps track of the run from $q$ and whether it witnesses 
a full muller set (1), a partial muller set ($\kappa$), or 
goes outside of a muller set (0). Likewise, 
$\alpha_2$ keeps track of the same for each $p \in P_i$, the run from $p$. 
If  a run is accepting, then $\alpha_2$ will eventually become the singleton 
$\{(0, \dots, 0, 1, 0, \dots, 0)\}$ corresponding to some muller set of $P_A$. 
\item  $M_{s}[(q, (r'_1, \dots, r'_n),P), X][(q',(r''_1, \dots, r''_n),P'), X']=\bot$, otherwise. 

 \end{itemize}

%

\begin{lemma}\label{lem:aperiodicSSTLA} 
  For all aperiodic 1-bounded $\mathsf{SST}_{\la}$ with star-free
  look-around, there exists an equivalent aperiodic 1-bounded SST.
\end{lemma}

\begin{proof}
Let $(T,A,B)$ be an $\mathsf{SST}_{\la}$, with  $A = (P_A, \Sigma, \delta_A, P_f)$  a deterministic lookahead muller automaton, 
$B=(P_B, \Sigma, \delta_B)$ be a deterministic look-behind automaton. 
Let $T = (\Sigma, \Gamma, Q, q_0,\delta, \varsst, \rho,
F)$. Without loss of generality, we make the following \emph{unique successor} assumption
  \begin{itemize}
 \item For all states $q,q',q''\in Q$, and for all states 
 $p, p' \in P_A$, and for all states $r, r' \in P_B$, and for any symbol $a \in \Sigma$, 
 whenever $p \neq p'$ or $r \neq r'$, and if  $\delta(q,r,a,p)=q'$, 
$\delta(q,r',a,p')=q''$, then $q'\neq q''$. 
 \end{itemize}
If this is not the case, say we have $p \neq p', r=r'$, and $q'=q''$. 
Then considering the state of $T$ as $(q,r)$, we obtain 
$\delta((q,r),a,p)=(q',r)=\delta((q,r),a,p')$. However, it is easy 
to define transitions $\delta((q,r),a,p)=(q',r)$, $\delta((q,r),a,p')=(q'',r)$
by duplicating the transitions of $q', q''$, without affecting anything else, especially the aperiodicity. 
The same argument can be given when $r \neq r'$ or both ($p \neq p'$ and $r \neq r'$). Thus, for our convenience, we assume the 
\emph{unique successor} assumption without loss of generality.

\subsection{Elimination of look-around from $(T,A,B)$ : Construction of $T'$} 
We construct an aperiodic and
1-bounded \sst{} $T'$ equivalent to \sst$_{sf}$ $(T,A,B)$.
As explained in definition of $\mathsf{SST}_{\la}$,  the unique  run of a string $s$ on $(T,A,B)$ is not only a sequence of $Q$-states, but also a collection of the look ahead states $2^{P_A}$, and maintains the reachable state from each state of $P_B$.
 At any time, the current state of $Q$, the $n$-tuple of reachable states of $P_B$, and the collection of look-ahead states $P \subseteq P_A$ is a configuration. 
 For brevity of notation, let $\eta$ (we also use $\zeta, \eta'$ in the sequel) denote the $n$-tuple of $P_B$ states.
 
A configuration $(q_1,\eta_1,P_1)$, on reading $a$, evolves into 
$(q_2,\eta_2,P_2 \cup\{p_2\})$, where $\delta(q_1,r_2,a,p_2)=q_2$ is a transition  in the $\mathsf{SST}_{\la}$  and  
$\delta_A(P_1,a)=P_2$, where $\delta_A$ is the transition function of the look ahead automaton $A$, and 
the state reached from $r_2$  is an accepting state of $P_B$.
Note that the transition monoid of the $\mathsf{SST}_{\la}$ is aperiodic and 1-bounded by assumption. We now show 
how to remove the look-around, resulting in an equivalent $\sst{}$ $T'$ whose transition monoid 
is aperiodic and 1-bounded. 

While defining $T'$, we put together all 
the states resulting from transitions of the form $(q,r,a,p,q')$ and $(q,r',a,p',q'')$ in the $\mathsf{SST}_{\la}$.  
We define $T'=(\Sigma, \Gamma,Q',q'_0, \delta', \varsst', \rho',Q_f')$
with: 
\begin{itemize}
\item $Q' = 2^{\textsf{useful}(T,A,B)}$ where $\textsf{useful}(T,A,B)$ are the
useful configurations of $(T,A,B)$ ;\\
(Note that we can pre-compute $\textsf{useful}(T,A,B)$ once we know $(T,A,B)$),

\item $q'_0=\{(q_0,\eta,\emptyset)\}$ ($\eta=(r_1, \dots, r_n)$ 
where $P_B=\{r_1, \dots, r_n\}$. 
  We assume wlog that, $(T,A,B)$ accepts at least one input. Therefore $(q_0,\eta,\emptyset)$ is useful),

\item $Q'_f$, the set of muller sets of $T'$, is defined by 
the set of pairs $(S,R)$ where $S$ is any of the $j$ muller sets in $F$, and 
$R$ is any of the $k$ muller sets in $P_f$. 

\item $\varsst'=\{X_{q'} \mid X \in \varsst, q' \in
  \textsf{useful}(T,A,B)\}$,

\item   The transitions are defined as follows:
$\delta'(S,a)=\bigcup_{(q,\eta,P)\in S}\Delta((q,\eta,P),a)$ where \\
$\Delta((q,\eta,P),a)=\{(q', \eta', P' \cup\{p'\}) \mid \delta(q,r,a,p') =q'$ and  
$\delta_A(P,a)=P'\}\cap \useful(T,A,B)$. Each component of $\eta'$ is obtained from $\eta$ using $\delta_B$.
\end{itemize}

Before defining the update function, we first assume a total ordering
$\preceq_{\useful(T,A,B)}$ on $\useful(T,A,B)$. 
For all $(p,\eta,P)$, we define the substitution $\sigma_{(p,\eta,P)}$ as  $X \in
\varsst \mapsto X_{(p,\eta, P)}$.  Let $(S,a,S')$ be a transition of $T'$. Given a
state $(q',\eta', P') \in S'$, there might be several predecessor states $(q_1, \eta_1, P_1),\dots,(q_k,\eta_k, P_k)$
in $S$ on reading $a$. The set $\{(q_1,\eta_1, P_1),\dots,(q_k, \eta_k, P_k)\} \subseteq S$ is denoted by 
$Pre_S((q',\eta',P'),a)$. Formally, it is defined by $\{ (q,\eta,P)\in S\ |\
(q',\eta',P')\in \Delta((q,\eta,P),a)\}$.

We consider only the variable update of the
transition from the  minimal predecessor state.
Indeed, since any string has at most one accepting run in the $\mathsf{SST}_{\la}$
$(T,A,B)$ (and at most one associated sequence of transitions), if two runs reach the same state at
some point, they will anyway define the same output and therefore we
can drop one of the variable update, as shown in
\cite{FiliotTrivediLics12}. Formally, the variable update
$\rho'(S,a,S')(X_{(q',\eta',P')})$, for all $X_{(q',\eta',P')}\in\varsst'$ is
defined by $\epsilon$ if $(q',\eta',P') \notin S'$, and by
$\sigma_{(q,\eta,P)}\circ \rho(q,r,a,p,q')(X)$, where $(q,\eta,P) = \text{min}\ \{
(t,\zeta,R)\in S\ |\ (q',\eta',P')\in \Delta((t,\zeta,R),a)\}$, and $\delta(q,r,a,p) = q'$
(by the \emph{unique successor} assumption  the look-around states $p,r$ are unique). 
It is shown in \cite{FiliotTrivediLics12} that indeed $T'$ is
equivalent to $T$. We show here that the transition monoid of $T'$ is aperiodic and
$1$-bounded.

For all $S\in Q'$, and $s \in \Sigma^*$,  define $\Delta^*(S,s)=\{ (q',\eta',P') \mid \exists
(q,\eta,P) \in S$ such that $(q,\eta,P) \rightsquigarrow^s_{T,A,B}
(q',\eta',P')\}\cap\useful(T,A,B)$.

\subsection{Connecting Transition Monoids of $T'$ and $(T,A,B)$}
\begin{lemma}
\label{lem:int}
Let $M_{T'}$ be the transition monoid of $T'$ and
$M_{(T,A,B)}$ the transition monoid of $(T,A,B)$. Let
$S_1,S_2\in Q'$, $X_{(q,\eta,P)},Y_{(q',\eta',P')}\in \varsst'$ and $s\in\Sigma^*$. \\
Then $M_{T',s}[S_1,X_{(q,\eta,P)}][S_2,Y_{(q',\eta',P')}]=i, \alpha_1, \alpha_2$
iff $S_2 = \Delta^*(S_1,s)$ and one of the 
following hold:
\begin{enumerate}
\item either $i=0$ and, $(q,\eta,P)\not\in S_1$ or $(q',\eta',P')\not\in S_2$, or
\item $(q,\eta,P) \in S_1$, $(q',\eta',P')\in S_2$, $(q,\eta,P)$ is the minimal
  ancestor in $S_1$ of $(q',\eta',P')$ (i.e. $(q,\eta,P) = \text{min}\ \{ (t,\zeta,R)\in S_1\ |\ (q',\eta',P')\in
  \Delta^*((t,\zeta,R),s)\}$), and\\
   $M_{(T,A,B),s}[(q,\eta,P),X][(q',\eta',P'),Y] = i, \alpha_1, \alpha_2$. 
\end{enumerate}
\end{lemma}
\begin{proof}
 It is easily shown that 
$M_{T',s}[S_1,X_{(q,\eta,P)}][S_2,Y_{(q',\eta',P')}]=i, \alpha_1, \alpha_2$ with $i \geq 0$ iff $S_2 =
\Delta^*(S_1,s)$. Let us show the two other conditions.

 Assume that 
$M_{T',s}[S_1,X_{(q,\eta,P)}][S_2,Y_{(q',\eta',P')}] = i, \alpha_1, \alpha_2$. The variable
update function is defined in $T'$ as follows: after reading string $s$,
 from state
$S_1$, all the variables $Z_{(t,\zeta,R)}$ such that $(t,\zeta,R)\not\in S_2$ are
 reset to $\epsilon$ (and therefore no variable can flow from
$S_1$ to them). In particular, if  $(q',\eta', P')\not\in S_2$, then no
variable can flow in $Y_{(q',\eta',P')}$ and $i=0$.

Now, assume that $(q',\eta',P')\in S_2$, and consider the sequence of states 
$S_1,S'_1,S'_2,\dots,S'_k,S_2$ of $T'$ on reading $s$. By definition
of the variable update, the variables that are used to update
$Y_{(q',\eta',P')}$ on reading the last symbol of $s$ from $S'_k$ 
are copies of the form $Z_{(t,\zeta,R)}$ such that $(t,\zeta,R)$ is the minimal
predecessor in $S'_k$ of $(q',\eta',P')$ (by $\Delta$). By induction, it is
easily shown that if some variable $Z_{(t,\zeta,R)}$ flows to $Y_{(q',\eta',P')}$ 
from $S_1$ to $S_2$ on reading $s$, then $(t,\zeta,R)$ is necessarily the
minimal ancestor (by $\Delta^*$) of $(q',\eta',P')$ on reading $s$. 
In particular if $(q,\eta,P)\not\in S_1$, then $i=0$. 

Finally, if $i>0$, then necessarily $(q,\eta,P)$ is the minimal ancestor in
$S_1$ of $(q',\eta',P')$ on reading $s$, from $S_1$ to $S_2$, and since $T'$
mimics the variable update of $(T,A,B)$ on the copies, we get 
that $M_{(T,A,B),s}[(q,\eta,P),X][(q',\eta',P'),Y]=i, \alpha_1, \alpha_2$.

The converse is shown similarly. \end{proof}

\subsection{1-boundedness and aperiodicity of $T'$} 
1-boundedness is an obvious consequence of the claim and the fact that $(T,A,B)$ is
$1$-bounded. Let us show that $M_{T'}$ is aperiodic. We know that
$M_{(T,A,B)}$ is aperiodic. Therefore there exists $n\in\mathbb{N}$ such
that for all strings $s\in\Sigma^*$,  $M_{(T,A,B),s}^n =
M_{(T,A,B),s}^{n+1}$.

We first show that $\forall$ $S_1,S_2\in Q'$, and all strings $s\in\Sigma^*$,
$\Delta^*(S_1,s^n) = S_2$ iff $\Delta^*(S_1,s^{n+1}) = S_2$.
  Indeed, 
  \begin{itemize}
\item $S_2 = \Delta^*(S_1,s^n)$, iff 
$$S_2 = \{ (q',\eta', P')\in\textsf{useful}(T,A,B)\ |\ \exists (q,\eta,P)\in S_1,\
(q,\eta,P)\flows^{s^n}_{T,A,B},(q',\eta',P')\},$$ iff
\item $S_2 = \{ (q',\eta',P')\in\textsf{useful}(T,A,B)\ |\ \exists (q,\eta,P)\in S_1,\
M_{(T,A,B),s}^n[(q,\eta,P),X][(q',\eta',P'),Y]=i, \alpha_1, \alpha_2$, $i \geq 0\text{ for some
}X,Y\in\varsst\}$, iff
\item by aperiodicity of $M_{(T,A,B)}$, 
$$S_2 = \{ (q',\eta',P')\in\textsf{useful}(T,A,B)\ |\ \exists (q,\eta,P)\in S_1,$$
$$~~~~~~M_{(T,A,B),s}^{n+1}[(q,\eta,P),X][(q',\eta',P'),Y]=i, \alpha_1, \alpha_2, i \geq 0 \\ \text{ for some
}X,Y\in\varsst\}$$ iff 
\item $S_2 = \Delta^*(S_1,s^{n+1})$.
\end{itemize}

Let $S_1,S_2\in Q'$ and $X_{(q,\eta,P)}, Y_{(q',\eta',P')}\in \varsst$. Let also
$s\in\Sigma^*$. We now show that the matrices corresponding to $s^n, s^{n+1}$ coincide, using Lemma \ref{lem:int}.
By the first condition of Lemma \ref{lem:int}, we have 
\begin{quote}
$M_{T',s}^n[S_1,X_{(q,\eta,P)}][S_2, Y_{(q',\eta',P')}] = i, \alpha_1, \alpha_2$ and condition $(1)$
of Lemma \ref{lem:int} holds, iff \\
$M_{T',s}^{n+1}[S_1,X_{(q,\eta,P)}][S_2,
Y_{(q',\eta',P')}] = i, \alpha_1, \alpha_2$ and condition $(1)$ of Lemma \ref{lem:int} holds. 
\end{quote}

\begin{itemize}
\item Indeed, $M_{T',s}^n[S_1,X_{(q,\eta,P)}][S_2, Y_{(q',\eta',P')}] = 0, \alpha_1, \alpha_2$ and,
$(q,\eta,P)\not\in S_1$ or $(q',\eta',P')\not\in S_2$ iff by Lemma \ref{lem:int},  
$\Delta^*(S_1,s^n) = S_2$, and $(q,\eta,P)\not\in S_1$ or $(q',\eta',P')\not\in
S_2$, iff \\
by what we just showed, $\Delta^*(S_1,s^{n+1}) = S_2$, and $(q,\eta,P)\not\in S_1$ or $(q',\eta',P')\not\in
S_2$, iff by Lemma \ref{lem:int}, $M_{T',s}^{n+1}[S_1,X_{(q,\eta,P)}][S_2,
Y_{(q',\eta',P')}] = 0, \alpha_1, \alpha_2$ and condition $(1)$ of Lemma \ref{lem:int} holds. 
\end{itemize}

Let us now look at condition $(2)$ of Lemma \ref{lem:int}, and  show that 
\begin{quote}
$M_{T',s}^{n}[S_1,X_{(q,\eta,P)}][S_2,
Y_{(q',\eta',P')}] = i, \alpha_1, \alpha_2$ and condition $(2)$ of Lemma \ref{lem:int} holds, iff \\
$M_{T',s}^{n+1}[S_1,X_{(q,\eta,P)}][S_2, Y_{(q',\eta',P')}] = i, \alpha_1, \alpha_2$ and condition
$(2)$ of Lemma \ref{lem:int} holds.
\end{quote}

\begin{itemize}
\item   We only show one direction, the other being
proved exactly similarly. Assume that \\
$M_{T',s}^{n}[S_1,X_{(q,\eta,P)}][S_2, Y_{(q',\eta',P')}] = i, \alpha_1, \alpha_2$ and $(q,\eta,P)\in
S_1$, $(q',\eta',P')\in S_2$, and $(q,\eta,P)$ is the minimal ancestor in $S_1$
of $(q',\eta',P')$, and $M_{(T,A,B),s}^n[(q,\eta,P),X][(q',\eta',P'),Y] = i, \alpha_1, \alpha_2$. 
By Lemma \ref{lem:int},  we have $\Delta^*(S_1,s^n) = S_2$, and therefore 
$\Delta^*(S_1,s^{n+1}) = S_2$. Now, we have
$(q,\eta,P) = \text{min}\ \{ (t,\zeta,R)\in S_1\ |\ (q',\eta',P')\in
\Delta^*((t,\zeta, R),s^n)\}$. Since $\Delta^*((t,\zeta,R),s^n) =
\Delta^*((t,\zeta,R),s^{n+1})$ for all $(t,\zeta,R)\in S_1$, we have
$(q,\eta,P) = \text{min}\ \{ (t,\zeta,R)\in S_1\ |\ (q',\eta',P')\in
\Delta^*((t,\zeta,R),s^{n+1})\}$.

 Finally, $M_{(T,A,B),s}^{n+1}[q,\eta,P,X][q',\eta',P',Y] = 
M_{(T,A,B),s}^{n}[q,\eta,P,X][q',\eta',P',Y]$ (by aperiodicity of $(T,A,B)$). 
By Lemma \ref{lem:int}, we obtain  $M_{T',s}^{n+1}[S_1,X_{(q,\eta,P)}][S_2,
Y_{(q',\eta',P')}] = i, \alpha_1, \alpha_2$ and hence condition $(2)$ of Lemma \ref{lem:int} is satisfied. 
\end{itemize}

Using Lemma \ref{lem:int} and the aperiodicity of $M_{(T,A,B)}$, we obtain the aperiodicity of $M_{T'}$.

\end{proof}

\end{document}